\newif\ifaos
\aosfalse

\ifaos
\documentclass[aos,submission]{imsart}
\else
\documentclass{article}
\usepackage[margin=1.7in]{geometry}
\usepackage[hang,flushmargin]{footmisc}
\fi

\RequirePackage{amsthm,amsmath,amsfonts,amssymb}
\RequirePackage[authoryear]{natbib}
\RequirePackage{xr-hyper}
\RequirePackage[colorlinks,citecolor=blue,urlcolor=blue]{hyperref}
\RequirePackage{graphicx}

\RequirePackage{mathtools}
\RequirePackage{array,multirow}
\RequirePackage{verbatim}
\RequirePackage{caption}
\RequirePackage{subcaption}
\RequirePackage{fancyvrb}
\RequirePackage{enumerate}
\RequirePackage{relsize}
\RequirePackage{placeins}

\RequirePackage{xparse}
\RequirePackage{xifthen}
\RequirePackage{enumitem}
\RequirePackage{pdfsync}
\RequirePackage{stefan_tex}
\theoremstyle{plain}
\newtheorem{prop}{Proposition}

\newtheorem{coro}[prop]{Corollary}
\newtheorem{lemm}[prop]{Lemma}
\newtheorem{theo}[prop]{Theorem}

\theoremstyle{remark}
\newtheorem{exam}{Example}

\newtheorem{rema}{Remark}

\ifaos 
\startlocaldefs 
\fi

\DeclareMathOperator{\var}{Var}

\DeclareMathOperator*{\vspan}{span}
\DeclareMathOperator*{\cspan}{\overline{span}}

\DeclareMathOperator*{\absconv}{absconv}
\DeclareMathOperator*{\conv}{conv}

\newcommand{\MM}{\mathcal M}

\newcommand{\tpsi}{\tilde{\psi}}
\newcommand{\influence}{\iota}
\newcommand{\dpsi}[1][]{\ifthenelse{\equal{#1}{}}{\dot{\psi}}{\dot{\psi}_{#1}}}
\newcommand{\dchi}[1][]{\ifthenelse{\equal{#1}{}}{\dot{\chi}}{\dot{\chi}_{#1}}}
\newcommand{\edchi}[1][]{\ifthenelse{\equal{#1}{}}{\dot{\chi}}{\dot{\chi}_{#1}}}
\renewcommand{\dh}[1][]{\ifthenelse{\equal{#1}{}}{\dot{h}}{\dot{h}_{#1}}}
\newcommand{\riesz}[1][]{\ifthenelse{\equal{#1}{}}{\gamma}{\gamma_{{#1}}}}
\newcommand{\hriesz}[1][]{\ifthenelse{\equal{#1}{}}{\hgamma}{\hgamma_{{#1}}}}
\newcommand{\ariesz}[1][]{\ifthenelse{\equal{#1}{}}{\tilde{gamma}}{\tilde{\gamma_}{{#1}}}}

\newcommand{\tangent}{\mathcal{T}}
\newcommand{\gammaipw}{\gamma^{\star}}
\newcommand{\excess}{\mathcal{L}}

\newcommand{\gapprox}{\tilde{g}}

\newcommand{\HH}{\mathcal{H}}

\DeclarePairedDelimiter\abs{\lvert}{\rvert}
\DeclarePairedDelimiter\norm{\lVert}{\rVert}
\DeclarePairedDelimiter\set{\{}{\}}
\DeclarePairedDelimiter\inner{\langle}{\rangle}

\renewcommand{\S}{\mathcal{S}}
\newcommand{\R}{\mathbb{R}}

\newcommand{\GG}{\mathcal{G}}
\newcommand{\F}{\mathcal{F}}

\def\independenT#1#2{\mathrel{\rlap{$#1#2$}\mkern2mu{#1#2}}}

\NewDocumentCommand{\Pn}{g}{\IfValueTF{#1}{\P{#1}_n}{\P_n}}
\RenewDocumentCommand{\P}{g}{\IfValueTF{#1}{P^{(#1)}}{P}}
\NewDocumentCommand{\einv}{g}{%
    \ensuremath{e^{-1}\IfValueT{#1}{\p{#1}}}}
\NewDocumentCommand{\E}{g}{%
	\ensuremath{\operatorname{{\mathbb E}}\IfValueT{#1}{ \p{#1} }}}
\RenewDocumentCommand{\L}{mg}{%
        \ensuremath{L_{#1}\IfValueT{#2}{(#2)}}}
\NewDocumentCommand{\M}{g}{\IfValueTF{#1}{\smash{\mathbb{M}^{(#1)}}}{\mathbb{M}}}
\NewDocumentCommand{\Mn}{g}{\IfValueTF{#1}{\smash{\mathbb{M}^{(#1)}_n}}{\mathbb{M}_n}}
\NewDocumentCommand{\I}{mmg}{{\IfValueTF{#3}{I_{#2}^{(#3)}}{I_{#2}}}(#1)}

\ifaos 
\endlocaldefs 
\fi

\ifaos
\fi

\begin{document}

\ifaos
\begin{frontmatter}

\title{Augmented Minimax Linear Estimation}
\runtitle{Augmented Minimax Linear Estimation}

\begin{aug}
\author[A]{\fnms{David} \snm{Hirshberg}\ead[label=e1]{davidahirshberg@stanford.edu}}
\and
\author[B]{\fnms{Stefan} \snm{Wager}\corref{}\ead[label=e2]{swager@stanford.edu}}
\address[A]{Department of Statistics, Stanford University, \printead{e1}}
\address[B]{Stanford Graduate School of Business, \printead{e2}}
\end{aug}
\else
\title{Augmented Minimax Linear Estimation}
\author{David A.~Hirshberg \and Stefan Wager}
\date{Stanford University}
\maketitle
\fi

\begin{abstract}
Many statistical estimands can expressed as continuous linear functionals of a conditional expectation function.
This includes the average treatment effect under unconfoundedness and generalizations for continuous-valued and personalized treatments.
In this paper, we discuss a general approach to estimating such quantities: we begin with a simple plug-in estimator based on
an estimate of the conditional expectation function, and then correct the plug-in estimator by subtracting a minimax linear estimate
of its error. We show that our method is semiparametrically efficient under weak conditions and observe promising performance on both real and simulated data.
\end{abstract}

\ifaos


\end{frontmatter}
\fi

\section{Introduction}
\label{sec:introduction}
Suppose we observe
$n$ independent and identically distributed samples $(Z_i, Y_i) \sim P$ with support in $\zz \times \RR$,
and we want to estimate a continuous linear functional of the form
\begin{equation}
\label{eq:psi}
\psi(m) = \EE[P]{h(Z_i, \, m)} \ \ \text{ at } \ \ m(z) = \EE[P]{Y_i \cond Z_i = z}. 
\end{equation}
Our main result establishes that we can build 
efficient estimators for a wide variety of such problems simply by subtracting from 
a plugin estimator $\psi(\hm)$ a minimax linear estimate of its error $\psi(\hm)-\psi(m)$. 

The following estimands from the literature on causal inference and missing data are of this type 
and can be estimated efficiently by our approach.

\begin{exam}[Mean with Outcomes Missing at Random]
\label{exam:mar}
We observe covariates $X_i$ and some but not all of the corresponding outcomes $Y_i^{\star}$.
We write $W_i \in \cb{0, \, 1}$ to indicate whether the outcome $Y_i^{\star}$ was observed, 
and define $Z_i=(X_i, W_i)$ and $Y_i=W_i Y_i^{\star}$; we then estimate the 
linear functional $\psi(m) = \EE[P]{m(X_i, \, 1)}$ at $m(x,w) = \EE[P]{Y_i \cond X_i=x, W_i=w}$.
This will be equal to the mean $\EE[P]{Y^{\star}_i}$ if, conditional on covariates $X_i$,
each outcome $Y_i^{\star}$ is independent of its nonmissingness $W_i$ \citep{rosenbaum1983central}. 
\end{exam}
\begin{exam}[Average Partial Effect]
\label{exam:ape}
Letting $Z_i = \p{X_i, \, W_i} \in \xx \times \RR$, we estimate the average of the derivative of
the response surface $m(x,w)$ with respect to $w$,
$\psi(m) = \EE[P]{\frac{\partial}{\partial w} \cb{m(X_i, \, w)}_{w = W_i} }$.
This estimand, and weighted variants of it, quantify the
average effect of a continuous treatment $W_i$ under exogeneity \citep*{powell1989semiparametric}.
\end{exam}
\begin{exam}[Average Partial Effect in the Conditionally Linear Model]
\label{exam:ape-linear}
In the setting of the previous example, we make the additional assumption that the 
regression function $m$ is conditionally linear in $w$, $m(x,w) = \mu(x) + w\tau(x)$.
The average partial effect is then $\psi(m) = \EE[P]{\tau(X_i)}$.
\end{exam}
\begin{exam}[Distribution Shift]
\label{exam:distribution-shift}
We estimate the effect of a shift in the distribution of the
conditioning variable $Z$ from one known distribution, $P_0$, to another, $P_1$, i.e., $\psi(m) =$ \smash{$\int m(z) (d\P_1(z) - d\P_0(z))$} for \smash{$m(z) = \EE[P]{Y_i \mid Z_i=z}$}.
Under exogeneity assumptions, this estimand can be used to compare policies for assigning personalized treatments, and estimators for it form a key building block in methods for estimation of optimal treatment policies.
\end{exam}

Below, we first discuss our estimator in the simple case that 
$h(z, \, m)$ in \eqref{eq:psi} does not depend on $z$, i.e., $h(z,m)=\psi(m)$. In this case, e.g., in Example~\ref{exam:distribution-shift},
we can evaluate
$\psi(m)$ without knowledge of the distribution $\P$ of $z$, and we say that our functional of interest 
$\psi(\cdot)$ is \emph{evaluable}.
From Section~\ref{sec:simple-setting} on, we will address the general case where $h$ also
depends on $z$ and so, even if we knew $m$ a-priori, we could only approximate $\psi(m)$
with a sample average \smash{$n^{-1}\sum_{i=1}^n h(Z_i,\,m)$.}

\subsection{Estimating Evaluable Linear Functionals}
\label{sec:amle}

Consider the estimation of $\psi(m)$ where $\psi(\cdot)$ is an evaluable mean-square-continuous linear functional.
The estimator we propose takes a plugin estimator $\psi(\hm)$, and then subtracts out an estimate of its error
$\psi(\hm)-\psi(m)=\psi(\hm-m)$ obtained as a weighted average of regression residuals,
\begin{equation}
\label{eq:aml-intro}
\hpsi = \psi(\hm) - \frac{1}{n} \sum_{i = 1}^n \hriesz_i \p{\hm(Z_i) - Y_i}.
\end{equation}
Our approach builds on a result of \citet*{chernozhukov2016locally} and \citet*{chernozhukov2018double},
who show that we can use the Riesz representer for $\psi$ to construct efficient estimators of this type.

To motivate this approach recall that, by the Riesz representation theorem,
any continuous linear functional $\psi(\cdot)$ on the square integrable functions from $\zz$ to $\R$
has a Riesz representer $\riesz[\psi](\cdot)$, i.e., a function satisfying \smash{$\int \riesz[\psi](z)f(z) d\P(z) = \psi(f)$} for all square-integrable functions $f$ \citep[e.g.,][Theorem 1.41]{peypouquet2015convex}.
Then, if we set \smash{$\hriesz_i = \riesz[\psi](Z_i)$} in \eqref{eq:aml-intro}, the second term
in the estimator acts as a correction for the error of \smash{$\psi(\hm)$} because
\begin{equation}
\label{eq:derivation}
\begin{split}
\psi(\hm) - \psi(m) &= \int \riesz[\psi](z)(\hm-m)(z) d\P(z)  
\approx \frac{1}{n}\sum_{i=1}^n \riesz[\psi](Z_i) (\hm(Z_i) - m(Z_i))  \\
&= \frac{1}{n}\sum_{i=1}^n \riesz[\psi](Z_i) (\hm(Z_i)-Y_i) + \frac{1}{n}\sum_{i=1}^n \riesz[\psi](Z_i) \p{Y_i - m(Z_i)}. 
\end{split}
\end{equation}
Thus, plugging the above expression into \eqref{eq:aml-intro}, we see that if we could compute
our estimator with the oracle Riesz representer weights $\riesz[\psi](Z_i)$, 
its error would very nearly be a weighted sum of mean-zero noise \smash{$n^{-1}\sum_{i=1}^n \riesz[\psi](Z_i) \varepsilon_i$} where \smash{$\varepsilon_i = Y_i - m(Z_i)$.}
This behavior is asymptotically optimal with a great deal of generality \citep[e.g.,][Proposition 4]{newey1994asymptotic}. 

Our goal will be to imitate the behavior of this oracle estimator without a-priori knowledge of the Riesz representer. 
One possible approach is to determine the form of the Riesz representer $\riesz[\psi](\cdot)$ by solving
analytically the set of equations that define it, 
\begin{equation}
\label{eq:riesz}
\int \riesz[\psi](z) f(z) d\P(z)  = \psi(f)\ \text{ for all $f$ satisfying }\ \int f(z)^2 d\P(z)  < \infty,
\end{equation}
then estimate it and plug the resulting weights $\hriesz_i = \hriesz[\psi](Z_i)$ into \eqref{eq:aml-intro}. 
In the context of our first example, the estimation of a mean with outcomes missing,
the Riesz representer is the inverse probability weight $\riesz[\psi](w,x) = w/e(x)$ where $e(x) = \P[W_i=1 \mid X_i=x]$, 
and this plug-in approach involves first obtaining an estimate \smash{$\he(x)$} of
treatment probabilities and then weighting by its inverse.
This is the well-known Augmented Inverse Probability Weighting (AIPW) estimator of \citet{robins1994estimation}. 
\citet{chernozhukov2016double} provide general results on the efficiency of
such estimators, provided \smash{$\hriesz[\psi](Z_i) - \riesz[\psi](Z_i)$} goes to zero fast enough in squared-error loss.

We take another approach. Considering our regression estimator $\hm$ and the design $Z_1 \ldots Z_n$ to be fixed,\footnotemark\,we
simply choose the weights $\hgamma \in \R^n$ that make our correction term $n^{-1} \sum_{i = 1}^n \hriesz_i \p{\hm(Z_i) - Y_i}$
a minimax linear estimator of what it is intended to correct for, $\psi(\hm-m)$.
To be precise, we first choose an absolutely convex set of functions $\ff$ which we
believe should contain the regression error $\hm - m$.
We then choose weights $\hgamma_i$ that perform best in terms of worst case mean squared error
over possible regression errors $\hm - m \in \ff$
and conditional variances satisfying $\Var[P]{Y_i \mid Z_i} \le \sigma^2$. 
This specifies the weights $\hriesz$ as the solution to a convex optimization problem,
\begin{equation*}
\hriesz = \argmin_{\gamma \in \R^n}\set*{ I_{\psi,\ff}^2(\gamma) + \frac{\sigma^2}{n^2}\norm{\gamma}^2 },
\quad I_{\psi,\ff}(\gamma) = \sup_{f \in \ff}\set*{\frac{1}{n} \sum_{i = 1}^n \gamma_i f(Z_i) - \psi(f)}.
\end{equation*}
The good properties of minimax linear estimators like this one are well known. \citet{donoho1994statistical} and related papers
\citep{armstrong2015optimal,cai2003note,donoho1991geometrizing,ibragimov1985nonparametric,
johnstone2015-gaussian-sequence,juditsky2009nonparametric} show that when a regression function $m$ is in a convex set $\ff$ and $Y_i \cond Z_i \sim N(0,\sigma_i^2)$,
a minimax linear estimator of a linear functional \smash{$\psi(m)$} will come within a factor 1.25 of the minimax risk over all estimators.
In addition to strong conceptual support, estimators of the type have been found to perform well
in practice across several application areas \citep{armstrong2015optimal,imbens2017optimized,zubizarreta2015stable}.

Methodologically, the main difference between our proposal and the references cited above
is that we use the minimax linear approach to debias a plugin estimator \smash{$\psi(\hm)$}
rather than as a stand-alone estimator.
Because we `augment' the minimax linear estimator by applying it after
regression adjustment in the same way that the AIPW estimator
augments the inverse probability weighting estimator, we refer to our
approach as the Augmented Minimax Linear (AML) estimator.
Our main result establishes semiparametric efficiency of the AML estimator
under considerable generality.

\footnotetext{If we estimate $\hm$ on an auxiliary sample, this is the case when we condition on both that sample and on $Z_1 \ldots Z_n$.
While it is not necessary to estimate $\hm$ on an auxiliary sample when estimating linear functionals,
it can buy us some additional flexibility. We discuss this in Section~\ref{sec:sample-splitting}.}

We note that the weights $\hgamma$ that underlie minimax linear estimation can be interpreted as a penalized least-squares solution to a set of estimating equations 
suggested by the definition \eqref{eq:riesz} of the Riesz representer $\riesz[\psi]$,
\begin{equation}
\label{eq:riesz-sample}
\frac{1}{n}\sum_{i=1}^n\riesz_i f(Z_i)  \approx \psi(f) \ \text{ for all }\ f \in \ff. 
\end{equation}
These estimating equations generalize covariate balance conditions
from the literature on the estimation of average treatment effects, and 
when analyzing our estimator we build on approaches used to study 
treatment effect estimators that use balancing weights \citep[e.g.,][]{athey2016approximate,graham1,imai2014covariate,kallus2016generalized,zubizarreta2015stable};
see Section~\ref{sec:balancing} for further discussion.

The restriction of $f$ to a strict subset $\ff$ of the square-integrable functions is 
necessary, as there are infinitely many square-integrable functions 
$f$ that agree on our sample $Z_1 \ldots Z_n$ and they need not even approximately agree in terms of $\psi(f)$. 
Our choice of this subset $\ff$, a set that characterizes our uncertainty about the regression error function $\hm - m$, 
focuses our estimated weights $\hgamma$ on the role they play in ensuring that \eqref{eq:riesz-sample} is satisfied for this function $f=\hm - m$. The size of this subset $\ff$, measured by, e.g., its Rademacher complexity, 
determines the accuracy with which these equations \eqref{eq:riesz-sample} can be simultaneously satisfied.
The smaller we can make $\ff$, i.e., the better the consistency guarantees we have for $\hm$, the more accurately we can solve
\eqref{eq:riesz-sample}. In practice, we may take $\ff$ to be a set of smooth functions, functions that are approximately sparse in some basis, 
functions of bounded variation, etc.

That our weights $\hriesz_i$ approximately solve the estimating equations \eqref{eq:riesz-sample}
does not imply that they estimate the Riesz representer \smash{$\riesz[\psi](\cdot)$} well in the mean-square
sense. 
However, to whatever degree the oracle weights $\riesz_i = \riesz[\psi](Z_i)$ also approximately solve 
\eqref{eq:riesz-sample}, it will imply that $\hriesz$ and $\riesz[\psi](\cdot)$ are close in the sense that
\begin{equation}
\label{eq:bilin-uniform}
\frac{1}{n}\sum_{i=1}^n [\hriesz_i-\riesz[\psi](Z_i)]f(Z_i) \approx 0 \ \text{ for all}\ f \in \ff.
\end{equation}
This property holds if and only if the vector with elements $\hriesz_i - \riesz[\psi](Z_i)$ is small or 
approximately orthogonal to every vector with elements $f(Z_i)$ for $f \in \ff$. 
And it implies that when $\hm - m \in \ff$, our estimator \eqref{eq:aml-intro} approximates the corresponding
oracle estimator, as the difference between them is $n^{-1}\sum_{i=1}^n [\hriesz_i-\riesz[\psi](Z_i)][(\hm-m)(Z_i) - \varepsilon_i]$. 

We state below a simple version of our main result. In essence, 
if an estimator $\hm$ converges to $m$ in mean square and our regression error $\hm - m$ is in a uniformly bounded Donsker class $\ff$ 
or more generally satisfies $(\hm - m)/O_P(1) \in \ff$, then our approach can be used to define an efficient estimator.

\subsection{Definitions}

As a measure of the scale of a function $f$ relative to an absolutely convex set $\ff$, 
we define the \emph{gauge} \smash{$\norm{f}_{\ff}$} $= \inf\set{\alpha > 0 : f \in \alpha\ff}$.
We will write $\ff_r$ to denote the localized class \smash{$\set{ f \in \ff : \norm{f}_{L_2(\P)} \le r}$},
$g\ff$ to denote the class of products $\set{ gf : f \in \ff}$, and
$h(\cdot, \ff)$ to denote the image class $\set{ h(\cdot, f) : f \in \ff}$.
We will write $\overline{\mm}$ to denote the closure of a subspace $\mm$ of the square-integrable functions
and $\mm_{\perp}$ to denote its orthogonal complement, and will write $\cspan \ff$ to denote the closure of $\vspan \ff$.
We will say that a set of functions $\ff$ from $\zz \to \R$ is pointwise bounded if $\sup_{f \in \ff} \abs{f(z)} < \infty$ for all $z \in \zz$,
uniformly bounded if $\sup_{f \in \ff}\norm{f}_{\infty} < \infty$ where $\norm{f}_{\infty}=\sup_{z \in \zz}\abs{f(z)}$, 
and pointwise closed if $f \in \ff$ whenever it is the limit of a sequence $f_j \in \ff$ in the sense that $\lim_{j \to \infty} f_j(z)=f(z)$ for all $z \in \zz$.

\subsection{Setting}
\label{sec:simple-setting}
We observe $(Y_1,Z_1) \ldots (Y_n,Z_n) \overset{iid}{\sim} P$ with $Y_i \in \R$ and $Z_i$ in an arbitrary set $\zz$.
We assume that $m(z)=\E[Y_i \mid Z_i=z]$ is in a subspace $\mm$ of the square integrable functions
 and that $v(z)=\Var[P]{Y_i \mid Z_i=z}$ is bounded. 
And we let $\ff$ be an absolutely convex set of square integrable functions.

Our estimand is $\psi(m)$ for a continuous linear functional $\psi(\cdot)$ on a subspace $\mm \cup \vspan \ff$
of the square integrable functions, which takes the form $\psi(m)=\E h(Z_i, m)$.
The Riesz representation theorem guarantees the existence and uniqueness of a function $\riesz[\psi] \in \cspan \ff$ 
satisfying the set of equations $\set{\E \riesz[\psi](Z) f(Z) = \psi(f) : f \in \cspan \ff}$.\footnotemark \ 
We call this function the Riesz representer of $\psi$ on the \emph{tangent space} $\cspan \ff$.
This generalizes our prior definition \eqref{eq:riesz}, coinciding when $\cspan\ff$ is the space of square integrable functions.

Our regularity and efficiency claims are relative to the set of all one-dimensional submodels $\P_t$ 
through $\P_0=\P$ for which, letting $(Y_t,Z_t) \sim P_t$, the regression functions \smash{$m_{\P_t}(z)=\E[Y_t \mid Z_t=z]$} are in $\mm$
and satisfy  \smash{$\lim_{t \to 0}\norm{m_{\P_t} - m_{\P}}_{L_2(\P)} = 0$} and the squares of \smash{$\epsilon_t = Y_t - m_{P_t}(Z_t)$} are uniformly integrable.
For these claims, we use the additional assumptions that there is a regular conditional probability $\P[Y_i \in \cdot \mid Z_i=z]$ and 
that $\mm_{\perp}$ has a dense subset of bounded functions.

\footnotetext{In this statement we implicitly work with the unique extension of the continuous functional $\psi(\cdot)$ defined on $\vspan \ff$
to a functional defined on its closure $\cspan \ff$ \citep[e.g.,][Theorem IV.3.1]{lang1993real}.}

\begin{theo}
\label{theo:simple}
In the setting above, choose finite $\sigma > 0$ and consider the estimator 
\begin{align}
\label{eq:aml}
&\hpsi_{AML} = \frac{1}{n} \sum_{i = 1}^n [ h(Z_i, \hm) - \hriesz_i \p{\hm(Z_i) - Y_i}] \quad \text{ where }\  \\
\label{eq:aml-primal}
&\hriesz = \argmin_{\gamma \in \R^n}\set*{I_{h,\ff}^2(\gamma) + \frac{\sigma^2}{n^2}\norm{\gamma}^2},\ \ I_{h,\ff}(\gamma) = \sup_{f \in \ff}\set*{ \frac{1}{n} \sum_{i = 1}^n [\gamma_i f(Z_i) - h(Z_i,f)]}. 
\end{align}
If $\ff$ is uniformly bounded and pointwise closed; $\ff$, $\riesz[\psi]\ff$, and $h(\cdot,\ff)$ are Donsker;
and $h(Z,\cdot)$ is pointwise bounded and mean-square equicontinuous on $\ff$ in the sense that 
$\sup_{f \in \F} \abs{ h(z,f) } < \infty$ for each $z \in \zz$ and 
$\lim_{r \to 0}\sup_{f \in \F_r}\norm{ h(\cdot,f) }_{L_2(P)} = 0$;
then our weights converge to the Riesz representer of $\psi$ on the tangent space $\cspan \ff$, i.e.,
\begin{equation}
\label{eq:gamma-consistency}
\frac{1}{n}\sum_{i=1}^n (\hgamma_i - \riesz[\psi](Z_i))^2 \to_P 0.
\end{equation}
If, in addition, $\hm$ has the tightness and consistency properties 
\begin{equation}
\label{eq:consistency-properties}
\norm{\hm - m}_{\ff} = O_P(1) \ \text{ and }\ \norm{\hm - m}_{\L{2}{\Pn}} = o_P(1) 
\end{equation}
then our estimator $\hpsi_{AML}$ is asymptotically linear, i.e., 
\begin{equation}
\label{eq:asymptotic-linearity}
\begin{split}
&\hpsi_{AML} - \psi(m) = \frac{1}{n}\sum_{i=1}^n \influence(Y_i, \, Z_i) + o_{P}(n^{-1/2}) \ \text{ where }\\
&\influence(y, \, z) = h(z,m) - \riesz[\psi](z)(m(z) - y) - \psi(m),
\end{split}
\end{equation}
and therefore $\sqrt{n}(\hpsi_{AML} - \psi(m))/V^{1/2} \Rightarrow \nn\p{0, \, 1}$ with $V = \EE[P]{\influence(Y, \, Z)^2}$.

Furthermore, an estimator satisfying \eqref{eq:asymptotic-linearity} is 
regular on the model class $\mm$ if $\mm \subseteq \cspan \ff$,
and asymptotically efficient if, in addition, $v(\cdot) \riesz[\psi](\cdot) \in \overline{\mm}$.\footnote{If an estimator satisfies
\eqref{eq:asymptotic-linearity}, 
a combination of two simple conditions implies efficiency: 
$\cspan \ff = \overline \mm$ and $v(\cdot)\overline{\mm} \subseteq \overline{\mm}$. The first says that
we correct for all error functions $\hm - m$ permitted by our assumption that $m \in \mm$, 
and waste no effort on those (in $\mm_{\perp}$) ruled out by it. The second
holds when the conditional variance $v(z)$ is sufficiently simple relative to $\overline{\mm}$,
e.g., when $v(z)$ is constant or when the model class $\mm$ is fully nonparametric in the sense that it contains an approximation to
every square integrable function.} 
\end{theo}

Theorem~\ref{theo:simple} follows from a finite sample result, Theorem~\ref{theo:simple-rate},
that we will discuss in Section~\ref{sec:main-results}. 
We end this section with a few remarks on the statistical behavior of the estimator, 
focusing on the choices of $\hm,\ff,\sigma$ that define a specific estimator $\hpsi$ of this type. 
We defer the discussion of computational issues to 
\ifaos
\else
Appendix~\ref{sec:computation}.
\fi

\begin{rema}
\label{rema:riesz-assumptions}
Our approach does not require knowledge of the functional form of the Riesz representer $\riesz[\psi](\cdot)$,
sparing us the trouble of solving \eqref{eq:riesz} analytically. 
\end{rema}

\begin{rema}
\label{rema:consistency}
If $\norm{m}_{\ff} < \infty$, the tightness and consistency properties \eqref{eq:consistency-properties}
are satisfied by the penalized least squares estimator $\hm=\argmin$ $n^{-1}\sum_{i=1}^n (Y_i - m(Z_i))^2 + \lambda \norm{m}_{\ff}$
for an appropriate choice of $\lambda$ (see Appendix~\ref{appendix:penalized-least-squares}).
For example, we might choose $\ff$ to be the absolutely convex hull $\set{ \sum_j  \beta_j \phi_j : \norm{\beta}_{\ell_1} \le 1}$
of a sequence of basis functions satisfying \smash{$\sum_{j=1}^{\infty} \E \phi_j^2(Z_i) < \infty$.}
It is Donsker \citep[Section 2.13.2]{vandervaart-wellner1996:weak-convergence}
and the corresponding estimator $\hm$ is $\ell_1$-penalized regression in this basis.
This approach is easy to implement and performs well in simuation when $\lambda$ is chosen by cross-validation.
In our simulations, we use a class of this type defined in terms a basis of scaled Hermite polynomials.
\end{rema}

\begin{rema}
\label{rema:semiparametric-models}
The choices we make for $\hm$ and $\ff$ reflect 
assumptions about the regression function $m$. In addition to nonparametric assumptions like smoothness, we may
make parametric or semiparametric assumptions. A semiparametric assumption 
distinguishes Examples~\ref{exam:ape} and \ref{exam:ape-linear}, which consider the average partial effect for arbitrary functions 
$m(x,w)$ and for functions of the form $m(x,w) = \mu(x) + w\tau(x)$ respectively. 

In the latter case, which we discuss in detail in
Section~\ref{sec:ape}, the tangent space $\cspan\ff$ is smaller than the space of all square integrable functions,
and the Riesz representer $\riesz[\ff]$ for $\psi(\cdot)$ will be the orthogonal projection onto $\cspan \ff$ of the Riesz representer $\riesz[L_2]$ for $\psi(\cdot)$ 
on the tangent space of all square-integrable functions. An important consequence is that,
under our efficiency condition $v\riesz[\psi]\in \overline\mm$,
the optimal asymptotic variance in Example~\ref{exam:ape-linear} is smaller than that in Example~\ref{exam:ape}.\footnotemark\
This reflects the ease of estimating the average partial effect in the conditionally linear model relative to the general case.

Naturally, such an estimator will be considered superefficient if we entertain the possibility that $m(x,w)$ does not have
the form $\mu(x) + w\tau(x)$, i.e., if our regularity condition $\mm \subseteq \cspan \ff$ is not satisfied. In this case,
our weights fail to adjust for the deviation $\hm - m$ for some possible regression function $m \in \mm$ in a neighborhood of $\hm$, 
and any gain in efficiency possible by doing so is, in a local minimax sense, spurious.
Characterization of the behavior of our estimator under this form of misspecification is important but beyond the scope of this paper.

This phenomenon is not unique to our approach; for additional discussion of the choice of tangent space when estimating a Riesz representer, see e.g., 
Remark 2.5 of \citet{chernozhukov2016double} and Section 3 of \citet{robins2007comment}. It pervades the literature on inference in high dimensional statistics,
which typically involves an estimate of the Riesz representer on an appropriate tangent space of high-dimensional parametric functions \citep[e.g.,][]{athey2016approximate, javanmard2014confidence, zhang2014confidence}. For example, when estimating a mean with outcomes missing at random in a high-dimensional linear model $m(x,w) = w x^T \beta$, 
$\riesz[\psi]$ is the best linear-in-$x$ approximation to the inverse propensity weights $w/e(x)$.

\footnotetext{
The difference in asymptotic variance between estimators using weights converging to $\riesz[L_2]$ (Example~\ref{exam:ape}) and weights converging to $\riesz[\ff]$ (Example~\ref{exam:ape-linear})
is $\E v(Z)[ \riesz[L_2]^2(Z) - \riesz[\ff]^2(Z)] = \E v(Z)[\riesz[L_2](Z) - \riesz[\ff](Z)]^2  + 2\E v(Z)\riesz[\ff](Z)[\riesz[L_2](Z) - \riesz[\ff](Z)].$
The first term in this decomposition is positive and the second term is zero \emph{if} $v\riesz[\ff] \in \cspan{\ff}$, as in this case 
$\E \riesz[L_2](Z) [v(Z)\riesz[\ff](Z)] = \psi(v\riesz[\ff]) = \E \riesz[\ff](Z)[v(Z)\riesz[\ff](Z)]$. 
}
\end{rema}

\begin{rema}
\label{rema:overlap}
Our assumption that $\psi(\cdot)$ has a square-integrable Riesz representer $\riesz[\psi]$,
equivalent to its mean-square continuity, 
is necessary in the sense that $\psi(m)$ does not have a regular estimator when it is violated 
\citep[Theorem 2.1][see Section~\ref{sec:pathwise-derivative} here for details]{van1991differentiable}.
 If $\ff$ has a finite uniform entropy integral, it is also sufficient. Theorem~\ref{theo:simple} requires no additional conditions on $\riesz[\psi]$ because
under this condition on $\ff$, the square integrability of $\riesz[\psi]$ implies our condition that $\riesz[\psi]\ff$ is Donsker \citep[Example 2.10.23]{vandervaart-wellner1996:weak-convergence}.

In the context of Example~\ref{exam:mar}, in which $\riesz[\psi](x,w)$ is the inverse probability weight $w/e(x)$ for 
$e(x)=P[W_i=1 \mid X_i=x]$, this means that all we require of $e(x)$ is that
$\E \riesz[\psi]^2(X_i,W_i) = \E 1/e(X_i) < \infty$.
\citet{d2017overlap} highlights the need for a weak condition like this, showing that the usual `strict overlap' condition that $e(x)$ is bounded away from zero
 implies strong constraints on the conditional distribution of $X_i \mid W_i$.
\citet{chen2008semiparametric} discusses the estimation of parameters defined by nonlinear moment conditions  
using overlap assumptions comparable to what we use here.

In simulation settings in which $\riesz[\psi](Z_i)$ has a spiky distribution, our estimator sometimes outperforms a double robust oracle estimator that weights using
the true Riesz representer $\riesz[\psi]$, while a typical double robust estimator performs substantially worse
than this oracle estimator. This suggests that common responses to limited overlap, like changing the estimand \citep[e.g.,][]{crump, li2018balancing} or 
assuming a semiparametric model as in Remark~\ref{rema:semiparametric-models},
may not be needed as frequently with our approach.
\end{rema}

\begin{rema}
\label{rema:riesz-estimation}
Although we assume no regularity conditions on the Riesz representer $\riesz[\psi]$,
 our weights $\hriesz_i$ still estimate it consistently. This is a universal consistency result,
in line with well known results about $k$-nearest neighbors regression and related estimators
\citep{lugosi1995nonparametric,stone1977consistent}. Heuristically, the reason
for this phenomenon is that the Riesz representer \smash{$\riesz[{\psi}]$} is the unique\footnotemark\ weighting function that sets a
population-analogue of $I_{h,\ff}$ to 0; because $\hgamma$ comes close to doing the same, 
it must also approximate \smash{$\riesz[{\psi}]$}. This universal consistency property
is not what controls the bias of our estimator \smash{$\hpsi$}. In fact, the rate of convergence of
\smash{$\hgamma_i$} to $\riesz[{\psi}](X_i)$ is in general too slow for standard arguments for
plugin estimators to apply. However, it plays a key role in understanding why we get efficiency under heteroskedasticity even though
we choose our weights by solving an optimization problem \eqref{eq:aml-primal} that is 
not calibrated to the conditional variance structure of $Y_i$.

\footnotetext{This uniqueness is violated when the tangent space $\cspan \ff$ that $\psi$ acts on is
not the space of all square integrable functions. However, the dual characterization Lemma~\ref{lemma:duality}
shows that our weights must converge to a function in this tangent space, and it follows that they
converge to the unique Riesz representer $\riesz[\psi]$ on this tangent space.}

To understand this phenomenon, observe that under the conditions of Theorem~\ref{theo:simple}, 
the conditional bias term 
$n^{-1}\sum_{i=1}^n h(Z_i, \hm-m) - \hriesz_i (\hm(Z_i)-m(Z_i))$ in our error is $o_P(n^{-1/2})$. It is therefore 
unnecessary to make an optimal bias-variance tradeoff by this sort of calibration
to get efficiency under heteroskedasticity and heteroskedasticity-robust confidence intervals;
the asymptotic behavior of our estimator is determined by the asymptotic behavior of our noise term 
\smash{$n^{-1}\sum_{i=1}^n \hriesz_i \varepsilon_i$} and therefore by the limiting weights $\riesz[\psi](Z_i)$.

For the same reason, it is not necessary to know the error scale $\norm{\hm - m}_{\ff}$ to form asymptotically valid confidence intervals.
We stress that this is an asymptotic statement; in finite samples, there are strong impossibility results
for uniform inference that is adaptive to the scale of an unknown signal \citep{armstrong2015optimal}. 
Furthermore, tuning approaches that estimate and incorporate individual variances $\sigma_i$ into the minimax weighting problem \eqref{eq:aml-primal}
like those discussed in \citet{armstrong2017finite} may offer some finite-sample improvement. 
\end{rema}

\subsection{Comparison with Double-Robust Estimation}
\label{sec:double-robustness}

Perhaps the most popular existing paradigm for building asymptotically efficient estimators 
in our setting is via constructions that first compute stand-alone estimates
\smash{$\hm(\cdot)$} and \smash{$\hriesz[\psi](\cdot)$} for the regression function and the Riesz representer,
and then plug them into the following functional form
\citep{chernozhukov2016locally,newey1994asymptotic,robins1},
\begin{equation}
\label{eq:DR}
\hpsi_{DR} = \frac{1}{n}\sum_{i=1}^n [h(Z_i, \hm) -  \hriesz[\psi](Z_i) \p{\hm(Z_i)-Y_i} ],
\end{equation}
or an asymptotically equivalent expression \citep[e.g.,][]{van2006targeted}.
This estimator has a long history in the context of many specific estimands,
e.g., the aforementioned AIPW estimator for the estimation of a mean with outcomes missing at random \citep*{cassel1976some,robins1994estimation}.
In recent work, \citet*{chernozhukov2018double} describe a general approach of this type,
making use of a novel estimator for the Riesz representer of a functional $\riesz[\psi]$ 
in high dimensions motivated by the Dantzig selector of \citet{candes2007dantzig}.

In considerable generality, this estimator $\hpsi_{DR}$ is efficient when we use sample splitting\footnotemark\ to construct $\hm$ 
and these estimators satisfy \citep{chernozhukov2016double, zheng2011cross}
\begin{equation}
\label{eq:bilinear-error-negligible}
 \frac{1}{n}\sum_{i=1}^n [\hriesz[\psi](Z_i)-\riesz_{\psi}(Z_i)] [ \hm(Z_i)-m(Z_i)] = o_P(n^{-1/2}).
\end{equation}
Taking the Cauchy-Schwarz bound on this bilinear form results in a well-known sufficient condition on
the product of errors, \smash{$\norm{\hriesz[\psi]-\riesz[\psi]}_{\L{2}{\Pn}}\norm{\hm-m}_{\L{2}{\Pn}}$} \smash{$= o_P(n^{-1/2})$}.
This phenomenon, that we can trade off accuracy in how well the two nuisance functions $m$ and $\riesz[\psi]$ are
estimated, is called \emph{double-robustness}. 

\footnotetext{In particular, this result holds if we use the cross-fitting construction of
\citet{schick1986asymptotically}, where separate data folds are used to estimate the nuisance
components \smash{$\hm$} and \smash{$\hriesz[\psi]$} and to compute the expression \eqref{eq:DR} given those estimates.
The three-way sample splitting scheme of \citet{newey2018cross}, discussed below, refines this by using different folds to estimate the two nuisance functions,
and the remaining ones to compute the expression \eqref{eq:DR}.
}

While the estimator $\hpsi_{AML}$ defined in \eqref{eq:aml} shares the form of $\hpsi_{DR}$, it is not designed to be double robust.
The weights $\hgamma$ used in $\hpsi_{AML}$ are optimized for the task of correcting the error of the plugin estimator 
$\psi(\hm)$ when our assumptions on the regression error function $\hm - m$ are correct. When this is the case and the class $\ff$ 
characterizing our uncertainty about this function is sufficiently small (e.g., Donsker), this allows us to be completely robust
to the difficulty of estimating the Riesz representer $\riesz[\psi]$. Our estimator will be efficient essentially because 
the error \smash{$\hriesz-\gamma_{\psi}$} will be sufficiently orthogonal to all functions $f \in \ff$ 
that \eqref{eq:bilinear-error-negligible} will be satisfied uniformly over the class
of possible regression error functions $\hm - m \in \ff$.
As the existence of an estimator \smash{$\hm$}
whose error \smash{$\hm - m$} is tight in the gauge of some Donsker class $\ff$
is equivalent to the existence of an \smash{$o_P(n^{-1/4})$}-consistent estimator of $m$,
relative to the aforementioned sufficient condition on the product of error rates, 
this characterization completely eliminates regularity requirements on the Riesz representer $\riesz[{\psi}]$ 
while requiring the same level of regularity on the regression function \smash{$m$}. 

This type of phenomenon is not unique to our approach. The higher order influence function estimator of \citet{mukherjee2017semiparametric} 
is efficient under the minimal H\"older-type smoothness conditions on $\riesz[\psi]$ and $m$. This includes the case
where either $m$ or $\riesz[\psi]$ admits an \smash{$o_P(n^{-1/4})$}-consistent estimator with no conditions on the other,
as well as possibilities interpolating these in which neither does \citep{robins2009semiparametric}.
Furthermore, \citet{newey2018cross} show that, if $\hm$ and $\hriesz[\psi]$ are 
appropriately tuned series estimators fit using a three-way cross-fitting scheme,  \smash{$\hpsi_{DR}$} is efficient under 
minimal or nearly minimal H\"older-type smoothness conditions. They also show that 
for this $\hm$, a cross-fit plug-in estimator \smash{$n^{-1}\sum_{i=1}^n h(Z_i, \hm)$} will
be efficient if $m$ is H\"older-smooth enough to admit an $o_P(n^{-1/4})$-consistent estimator, and 
beyond this regime exhibits some double robustness --- 
it is also efficient when $m$ is less smooth and $\riesz[\psi]$ is smooth enough.

The use of undersmoothed, i.e., less biased than variable, nuisance estimators 
seems to be an important ingredient in estimators that beat the error rate product bound \citep[see also][]{kennedy2020optimal,van2019efficient}. 
Both here and in \citet{newey2018cross}, $\riesz[\psi]$ is estimated by solving a set of Riesz representer estimating equations \eqref{eq:riesz-sample}
subject to weak regularization or constraints. 
Furthermore, when $\F$ is a ball in a reproducing kernel Hilbert space,
the minimax linear estimator ($\hpsi_{AML}$ with $\hm \equiv 0$) is
equivalently described as a plug-in using a undersmoothed ridge regression estimator $\hm$ \citep[Theorem 22]{kallus2016generalized}.
\citet{hirshberg2019minimax} show that this estimator is efficient essentially whenever $\norm{m}_{\F} < \infty$.


\subsection{Comparison with Minimax Linear and Balancing Estimators}
\label{sec:balancing}

As discussed above, our approach is primarily motivated as a refinement of conditional-on-design minimax linear estimators as
developed and studied by a large community over the past decades
\citep[e.g.,][]{donoho1994statistical, ibragimov1985nonparametric, juditsky2009nonparametric};
however, our focus is on its behavior in a random-design setting,
as in the literature on semiparametrically efficient inference and local asymptotic minimaxity, including results on doubly robust methods
\citep[e.g.,][]{bickel,robins1,van2006targeted}.
The conceptual distinction between these two settings is strong in causal inference and missing data
problems, where in the former we consider an adversary that chooses $m(\cdot)$ having observed the realized covariates and pattern of missing data,
and in the latter we consider an adversary that chooses $m(\cdot)$ having observed no part of the realized data. 

We are aware of three estimators that can be understood as special cases of our augmented
minimax linear estimator \eqref{eq:aml}.
In the case of parameter estimation in high-dimensional linear models, \citet{javanmard2014confidence}
propose a type of debiased lasso that combines a
lasso regression adjustment with weights that debias the $\ell_1$-ball, a convex class known to
capture the error of the lasso; \citet*{athey2016approximate} develop a related
idea for average treatment effect estimation with high-dimensional linear confounding;
and \citet{kallus2016generalized, kallus2018balanced} proposes analogs
for treatment effect estimation and policy evaluation, a special case of Example~\ref{exam:distribution-shift},
that adjust for nonparametric confounding using weights that debias the unit ball of a reproducing kernel Hilbert space.   
The contribution of our paper
relative to this line of work lies in the generality of our results, and also in characterizing the
asymptotic variance of the estimator under heteroskedasticity and proving efficiency
in the fixed-dimensional nonparametric setting. Given heteroskedasticity, the aforementioned papers
prove $\sqrt{n}$-consistency but do not characterize the asymptotic variance directly in terms of the distribution of the data; 
instead, they express the variance in terms of the solution to an optimization problem analogous to \eqref{eq:aml-primal}.

In the special case of mean estimation with outcomes missing at random, the optimization
problem \eqref{eq:aml-primal} takes on a particularly intuitive form, with
\begin{equation}
I_{h,\ff}(\gamma) = \sup_{f \in \ff} \cb{\frac{1}{n} \sum_{i = 1}^n \p{1 - W_i \gamma_i} f(X_i, \, 1)}
\end{equation}
measuring how well the \smash{$\gamma$}-weighted average of $f(x,1)$ over the units with observed outcomes
matches its average over everyone. In other words, the minimax linear weights enforce
``balance'' between these subsamples, which has been emphasized as fundamental to this problem by several authors
including \citet{rosenbaum1983central} and \citet*{hirano2003efficient}. Recently there has been considerable
interest in the use of balancing weights, chosen to control $I_{h,\ff}$ or a variant,
in linear estimators and in augmented linear estimators \eqref{eq:aml} like those we consider here \citep{athey2016approximate,chan2015globally,graham1,graham2,hainmueller,imai2014covariate,kallus2016generalized,ning2017high,wang2017approximate,wong2017kernel,zhao2016covariate,zubizarreta2015stable}. 
In addition to generalizing beyond the missing-at-random problem, our Theorem \ref{theo:simple-rate} provides the sharpest results we are aware of for
balancing-type estimators in this specific problem. 

To do this, we bring together arguments from two strands of the balancing literature.
The first focuses on balancing small finite-dimensional classes, 
and in several instances it has been shown that when tuned so that $I_{h,\ff}(\hgamma)$ is sufficiently small, 
the linear estimator is efficient under strong assumptions on both $m$ and $\riesz[\psi]$ \citep{chan2015globally,fan2016improving,graham1,wang2017approximate}.
The arguments used to establish these results rely on the convergence of $\hriesz$ to $\riesz[\psi]$ at sufficient rate, much like those
used with the estimators discussed in the previous section. 
The second focuses on balancing high or infinite-dimensional classes, 
and in several instances it has been shown that when tuned so that \smash{$I_{h,\ff}(\hgamma) = O_P(n^{-1/2})$},
a level of balance that is attainable under assumptions comparable to ours, the linear estimator is $\sqrt{n}$-consistent and 
the augmented linear estimator is $\sqrt{n}$-consistent and asymptotically unbiased \citep{athey2016approximate,kallus2016generalized,wong2017kernel}. 
The arguments used to establish these results fundamentally rely on balance to bound the estimator's bias, 
and do not fully characterize the estimator's asymptotic distribution. Our argument is a refinement of this one,
using balance to do the bulk of the work, but relying on the convergence of the balancing weights $\hgamma$ to $\riesz[\psi]$ 
to characterize the asymptotic distribution of our estimator and to establish asymptotic unbiasedness under weaker conditions.

\section{Estimating Linear Functionals}
\label{sec:estimating-linear-functionals}

In this section, we give a more general characterization of the behavior of our estimator. We begin by sketching our argument, 
which is based on a decomposition of our estimator's error into a bias-like term and a noise-like term. 
We consider error relative to a sample-average version of our estimand, $\tilde{\psi}(m) = n^{-1}\sum_{i=1}^n h(Z_i,m)$,
as the difference $\psi(m)-\tpsi(m)$ is out of our hands:
\begin{equation}
 \label{eq:error-decomp}
\begin{split}
&\hpsi_{AML} - \tilde{\psi}(m)  
= \frac{1}{n}\sum_{i=1}^n h(Z_i, \hm) - \hgamma_i \p{\hm(Z_i) - Y_i} - h(Z_i, m)  \\
&\quad = \frac{1}{n}\sum_{i=1}^n \underbrace{h(Z_i,\hm-m) - \hgamma_i (\hm-m)(Z_i)}_{\text{bias}} 
  + \underbrace{\hgamma_i \p{Y_i - m(Z_i)}}_{\text{noise}}.
\end{split}
\end{equation}
In Appendix~\ref{sec:finite-sample-proofs}, we prove finite sample bounds on the bias term and 
the difference between the noise term and that of the oracle estimator with weights $\riesz[\psi](Z_i)$.
Our estimator will be asymptotically linear, with the influence function of the oracle estimator,
if both of these quantities are $o_p(n^{-1/2})$. We establish these bounds in three steps.

\paragraph*{Step 1}
We bound \smash{$n^{-1}\sum_{i=1}^n (\hgamma_i - \gamma^{\star}_i)^2$} for $\gamma^{\star}_i = \riesz[\psi](Z_i)$.
To do this,  we work with a dual characterization of our weights $\hgamma_i$ as evaluations $\hriesz[\psi](Z_i)$ of a penalized least squares estimate
of the Riesz representer $\riesz[\psi]$. 
\begin{equation}
\label{eq:dual}
\begin{split}
\hriesz[\psi] &= \argmin_{g} \set*{ \norm{g}_{L_2(\Pn)}^2 - \frac{2}{n}\sum_{i=1}^n h(Z_i, g) + \frac{\sigma^2}{n}\norm{g}_{\ff}^2 } \\
	      &= \argmin_{g} \set*{  \norm{g - \riesz[\psi]}_{L_2(\Pn)}^2 - \frac{2}{n}\sum_{i=1}^n h_{\riesz[\psi]}(Z_i, g) + \frac{\sigma^2}{n}\norm{g}_{\ff}^2}
\end{split}
\end{equation}
where $h_{\riesz}(z,f) = h(z,f) - \riesz(z)f(z)$. Here the term involving $h_{\riesz[\psi]}$ plays the role of `noise' in our least squares problem, 
as it has mean zero for any function $f \in \mm$. The first characterization is established using strong duality in Lemma~\ref{lemma:duality}
and the second is derived by completing the square.

\paragraph*{Step 2} We bound the difference between our noise term and that of the oracle estimator,
$n^{-1}\sum_{i=1}^n (\hgamma_i - \gamma^{\star}_i)(Y_i - m(Z_i))$, using the result of Step 1.

\paragraph*{Step 3} We bound our bias term by $\norm{\hm - m}_{\ff}I_{h, \ff}(\hgamma)$,
where as a consequence of the definition of our weights $\hgamma$ in \eqref{eq:aml-primal}, 
\begin{equation}
\label{eq:gamma-optimality-cond}
I_{h, \ff}^2(\hgamma) \le I_{h,\ff}^2(\gamma^{\star}) + \frac{\sigma^2}{n^2} \sum_{i=1}^n \p{{\gamma^{\star}_i}^2 - \hgamma_i^2}. 
\end{equation}
The first term on the right side can be characterized using empirical process techniques, as $I_{h,\ff}(\gamma^{\star})$ 
is the supremum of the empirical measure indexed by the class of mean-zero functions 
$h_{\riesz[\psi]}(\cdot,\ff)$. And the second term can be shown, using 
some simple arithmetic, to be $o_p(n^{-1})$ when $\hriesz$ is consistent. 
Thus, our bias term will be bounded by $\norm{\hm - m}_{\ff}[I_{\ff}(\gamma^{\star}) + o_p(n^{-1/2})]$.

\paragraph*{Step 3'}
We refine this bound to take advantage of the consistency of $\hm$. 
To do this, we show that our estimator behaves essentially the same way as an oracle
that knows a sharp bound $\norm{\hm - m}_{L_2(\Pn)} \le \rho$ on our regression error
and uses a refined model class \smash{$\F_{\rho}' = \{ f : \norm{f}_{\F}^2 +$} \smash{$ \rho^{-2}\norm{f}_{L_2(\Pn)}^2$ $\le 1\}$}
in place of $\F$. The key insight is that this substitution changes the dual \eqref{eq:dual} and its solution $\hgamma$ very little,
so replacing $\F$ with $\F_{\rho}'$ in our bound \eqref{eq:gamma-optimality-cond} yields an inequality that is approximately satisfied.
Given the assumptions of Theorem~\ref{theo:simple}, 
the resulting refined bias term bound will be \smash{$o_p(n^{-1/2})$}, as 
\smash{$\norm{\hm - m}_{\F_{\rho}'} = O_p(1)$} for $\rho \to 0$ given our
tightness and consistency assumptions \eqref{eq:consistency-properties} and
\smash{$I_{h, \F_{\rho}'}(\gammaipw) = o_p(n^{-1/2})$} when $\rho \to 0$
given our Donskerity and equicontinuity assumptions.\\

Following a few definitions, we will state our main result.
Due to space constraints, all proofs are given in the appendix.

\subsection{Finite sample results}
\label{sec:main-results}
To characterize the size of a set $\GG$, we will use its \emph{Rademacher complexity}, $R_n(\GG) = \E \sup_{g \in \GG}\abs{n^{-1}\sum_{i=1}^n \epsilon_i g(Z_i)}$ where $\epsilon_i = \pm 1$ each with probability $1/2$
independently and independently of the sequence $Z_1 \ldots Z_n$, as well as the uniform bound \smash{$M_{\infty}(\GG)$} \smash{$= \sup_{g \in \GG} \norm{g}_{\infty}$}.
Letting $h_{\gamma}(z,f)=h(z,f)-\gamma(z)f(z)$, our bound depends
on the Rademacher complexity of the classes $\ff_r$,
$h_{\riesz[\psi]}(\cdot,\ff_r)$, and $h_{\tilde \riesz}(\cdot,\ff_r)$ 
for a regularized approximation $\tilde\riesz$ to $\riesz[\psi]$. The regularity of that approximation, 
and therefore the regularity of $\riesz[\psi]$ itself, will be a factor in a higher order term.
Without loss of generality, we will write our weights as function evaluations $\hriesz_i=\hriesz(Z_i)$,
and we will write $a \vee b$ and $a \wedge b$ respectively for the maximum and minimum of $a$ and $b$
and $a \lesssim b$ and $a \ll b$ meaning $a = O(b)$ and $a = o(b)$.

\begin{theo}
\label{theo:simple-rate}
In the setting described in Section~\ref{sec:simple-setting}, 
consider the estimator $\hpsi_{AML}$ defined in \eqref{eq:aml} with $\sigma > 0$ and $\ff$ 
a uniformly bounded absolutely convex set of functions for which $h(\cdot,\ff)$ is pointwise bounded.
Let $\riesz[\psi]$ be the Riesz representer of $\psi$ on the tangent space $\cspan \ff$ and 
$\tilde \riesz$ minimize \smash{$\norm{\riesz[\psi] - \riesz}_{L_2(Q)}^2 + (\sigma^2/n)\norm{\riesz}_{\F}^2\ $} for $Q=\P$ or $Q=\Pn$.
If $\F$ is \smash{$\norm{\cdot}_{L_2(Q)}$}-closed,
this argmin exists and is unique, and for any positive $\delta$, on the intersection of an event of probability \smash{$1-4\delta-3\exp(-c_2nr_Q^2/M_{\infty}^2(\F))$}
and one on which $\norm{\hm - m}_{\F} \le s_{\F}$ and $\norm{\hm - m}_{L_2(\Pn)} \le s_{L_2(\Pn)}$,
\begin{equation}
\label{eq:sample-riesz-rate} 
\begin{aligned}
&\norm{\hriesz - \tilde\riesz}_{L_2(\Pn)}^2 
\le 6\p{nr^4/\sigma^2 + \norm{\tilde\riesz}_{\ff}r^2} \vee 8 r^2 \quad \text{ for } \quad  r = r_Q \vee r_M, \\
r_Q &= \inf\set{r > 0 : R_n(\F_{c_0r}) \le c_1 r^2/M_{\infty}(\F)},\\
r_M &= \begin{cases} \inf\set{ r > 0 : R_n(h_{\tilde\riesz}(\cdot,\F_r))  \ \le \ \delta r^2/2} & \text{ for } \quad Q=\P,  \\
  		     \inf\set{ r > 0 : R_n(h_{\riesz[\psi]}(\cdot,\F_r)) \le \delta r^2/2} & \text{ for } \quad Q=\Pn, \end{cases}  \\
\end{aligned}
\end{equation}
and for $\influence_{\riesz}(y,z) = h(z,m) - \riesz(z)(m(z) - y) - \psi(m)$
and any positive $\epsilon \le 9/16$,
\begin{equation}
\label{eq:remainder-rate} 
\begin{aligned}
&\sqrt{n}\abs{\hpsi_{AML} - \psi(m) - n^{-1}\sum_{i=1}^n \influence_{\tilde\riesz}(Y_i,Z_i)} 
\le (1/\sqrt{\delta}) \norm{v}_{\infty} \norm{\hriesz - \tilde\riesz}_{L_2(\Pn)} \\ 
&\quad +  \sqrt{2n}s_{\F}\ \phi\p{\frac{s_{L_2(\Pn)}}{s_{\F}} \vee c_0 r \vee \frac{6 \sigma}{\epsilon \sqrt{n}}} \p{1+2\epsilon\ /\ \sqrt{1-\epsilon^2/36}}\\
&\quad +  \sqrt{2}\sigma s_{\F} \p{\norm{\riesz[\psi]}_{L_2(\Pn)} \wedge \norm{\riesz[\psi]}_{L_2(\Pn)}^{1/2} \norm{\hgamma-\riesz[\psi]}_{L_2(\Pn)}^{1/2}}\ /\ \sqrt{1-\epsilon^2/36}. 
\end{aligned}
\end{equation}
Here $c_0 \ldots c_2$ are universal constants and 
\begin{align*} 
\phi(\rho) &= \frac{2R_n(h_{\riesz[\psi]}(\cdot, \F_{\sqrt{2}\rho}))}{\delta} 
	   \vee \frac{216}{\epsilon^2}\p{r^2 + \frac{\sigma^2 \norm{\tilde\riesz}_{\ff}}{n}}
	   \vee \frac{ 36\sigma^2 \norm{\riesz[\psi]}_{L_2(\P)}}{\epsilon^2\sqrt{\delta} c_0 n r} \vee \frac{288\sigma^2}{\epsilon^2 n}.
\end{align*}
\end{theo}
Generalization to classes $\F$ that are not uniformly bounded is discussed in Appendix~\ref{sec:not-uniformly-bounded}.
We will briefly interpret this result by considering several asymptotic settings. 
Throughout, we will use the bounds above for $Q=\Pn$ and
the bound $\norm{\tilde\riesz}_{\F} \le (\sqrt{n}/\sigma)\norm{\riesz[\psi]}_{L_2(\Pn)}$.\footnote{
This bound holds because \smash{$\norm{\riesz - \riesz[\psi]}_{L_2(\Pn)}^2 + (\sigma^2/n)\norm{\riesz}_\F^2$}
is smaller at its minimizer than at $\riesz=0$.} 

\subsection{Nonparametric asymptotics}
\label{sec:nonparametric-asymptotics}
In the asymptotic setting we considered in the introduction,
in which the distribution $\P$, the class $\F$, and the tuning parameter $\sigma$ are fixed,
this result implies Theorem~\ref{theo:simple}.
The key steps of the proof are as follows. 
\begin{enumerate}
\item 
As $\riesz[\psi]$ is fixed, the regularized approximation $\tilde\riesz$ converges to $\riesz[\psi]$
in \smash{$\norm{\cdot}_{L_2(\Pn)}$} as the weight of regularization $\sigma^2/n \to 0$, so our `influence function' $\influence_{\tilde\riesz}$ converges to 
the limit $\influence_{\riesz[\psi]}$. 

\item
Given our tightness and consistency assumptions
\eqref{eq:consistency-properties}, we can take $s_{\F} \ge \norm{\hm-m}_{\F}$ to be of constant order and \smash{$s_{L_2(\Pn)} \ge \norm{\hm - m}_{L_2(\Pn)}$} to be converging to zero on a high probability event.
Thus, our remainder bound \eqref{eq:remainder-rate} goes to zero if 
$\sqrt{n}\phi(s_n) \to 0$ for any sequence $s_n$ converging to zero and \smash{$r \ll n^{-1/4}$} and therefore $\norm{\hriesz - \tilde\riesz}_{L_2(\Pn)} \to 0$ (via \ref{eq:sample-riesz-rate}).

\item
Both of these conditions hold if
$\lim_{t \to 0}\sqrt{n}R_n(\F_{t}) = \lim_{t \to 0}\sqrt{n}R_n(h_{\riesz[\psi]}(\cdot,\F_{t})) = 0$.
The first limit is zero because $\F$ is Donsker. 
And the second is zero for the same reason, as
$h_{\riesz[\psi]}(\cdot,\F_t) \subseteq \HH_{\omega(t)}$ where $\HH = h_{\riesz[\psi]}(\cdot,\F)$
is Donsker and $\omega(t)=\sup_{f \in \F_t}\norm{h_{\riesz[\psi]}(\cdot, f)}_{L_2(\P)}$ 
satisfies $\lim_{t \to 0}\omega(t) = 0$ under our equicontinuity and uniform boundedness assumptions. 
\end{enumerate}

\subsection{High dimensional asymptotics}
\label{sec:high-dim-asymptotics}
Now we consider estimation of the mean with outcomes missing at random (Example~\ref{exam:mar}) in the high dimensional linear model,
i.e., with $m(x,w)=wx^T\beta$ for $\beta \in \R^p$. In this setting, $\hpsi_{AML}$ for the class $\F = \set{m(x,w)=wx^T\beta : \norm{\beta}_{1} \le 1}$
is the ``approximate residual balancing'' estimator proposed in \citet{athey2016approximate}. We can derive from 
Theorem~\ref{theo:simple-rate} the main result from that paper: that this estimator is $\sqrt{n}$-consistent 
and an associated $t$-statistic is asymptotically standard normal.
Furthermore, Theorem~\ref{theo:simple-rate} also characterizes the limit of the weights $\hriesz$, and therefore the asymptotic variance of the estimator,
as a simple function of the distribution $\P$.

Specifically, suppose the coordinates of the covariates are bounded, 
$\norm{\riesz[\psi]}_{L_2(\P)}$ is bounded  (an overlap assumption), 
and \smash{$\norm{\hat{\beta} - \beta}_1 = O_P(s_\F)$} for \smash{$s_\F \ll 1/\sqrt{\log(p)}$.}
As discussed in \citet{athey2016approximate}, when \smash{$\hat \beta$} is estimated via the lasso, 
the third holds under standard sparsity and restricted eigenvalue conditions.
Then for any choice of tuning parameter $\sigma$ satisfying \smash{$\sqrt{\log(p)} \ll \sigma \ll 1/s_\F$},
 \smash{$\hat{\psi}_{AML} - \psi(m)$} is first-order equivalent to $n^{-1}\sum_{i=1}^n \influence_{\tilde\riesz}(Y_i,Z_i)$,
as our remainder bound \eqref{eq:remainder-rate} is vanishingly small.

To check this, note that by the finite class lemma of \citet[Lemma 5.2]{massart2000some},
\[ R_n(\F) \vee R_n(h_{\riesz[\psi]}(\cdot,\F)) \lesssim (1 \vee \norm{\riesz[\psi]}_{L_2(\P)})\sqrt{\log(p)/n}. \]
Thus, \eqref{eq:sample-riesz-rate} implies the convergence of $\hriesz$ to $\tilde\riesz$,
as $r^2 \lesssim R_n(\F) \vee R_n(h_{\riesz[\psi]}(\cdot,\F))$. 
It follows that the first term in our remainder bound \eqref{eq:remainder-rate} vanishes.
The second term vanishes as well,
as it is proportional to \smash{$\sqrt{n} s_\F \cdot \phi(x) \ll \sqrt{n/\log(p)} \cdot \phi(x)$} for some $x$ 
and \smash{$\phi(x) \lesssim R_n(h_{\riesz[\psi]}(\cdot,\F)) \lesssim \sqrt{\log(p)/n}$.}
So does the third term, as $\sigma s_\F \ll (1/s_\F) s_\F$.

\subsection{Sieve asymptotics}
\label{sec:sieve-asymptotics}
In the sieve asymptotics often considered \citep[e.g.,][]{newey2018cross, wang2017approximate},
we do not characterize the regression function $m$ by membership in a set $\F$ directly,
but instead by the existence of an element $\tilde m \in \F$ that approximates 
it with a certain degree of accuracy.
Our argument requires modification for this asymptotic setting,
as our bound $\norm{\hm - m}_{\ff}I_{h, \ff}(\hgamma)$ on the `bias term' 
in our error decomposition \eqref{eq:error-decomp} 
will tend to be vacuous: when $\hm - m \not \in \vspan \F$, $\norm{\hm - m}_{\F}=\infty$. 
We can modify our error decomposition as follows.
\begin{equation*}
\begin{split}
&\hpsi_{AML} - \tilde{\psi}(m)  
= \frac{1}{n}\sum_{i=1}^n h(Z_i,\hm-\tilde m) - \hgamma_i (\hm-\tilde m)(Z_i) 
+ \frac{1}{n}\sum_{i=1}^n \hgamma_i \p{Y_i - m(Z_i)} \\
&\quad 
+ \frac{1}{n}\sum_{i=1}^n h_{\riesz[\psi]}(Z_i,\tilde m-m)
+ \frac{1}{n}\sum_{i=1}^n (\riesz[\psi] - \tilde\riesz)(\tilde m-m)(Z_i)
+ \frac{1}{n}\sum_{i=1}^n (\tilde\riesz - \hriesz)(\tilde m-m)(Z_i).
\end{split}
\end{equation*}
The sum of the first two terms tends to converge to the influence function average $n^{-1}\sum_{i=1}^n \influence_{\tilde\riesz}(Y_i,Z_i)$.
The proof of 
Theorem~\ref{theo:simple-rate} implies the remainder satisfies the bound \eqref{eq:remainder-rate} for $s_\F \ge \norm{\hm-\tilde m}_\F$ and $s_{L_2(\Pn)} \ge \norm{\hm - \tilde m}_{L_2(\Pn)}$. We briefly discuss the remaining terms.

The third term is the sample average of 
a deterministic function with mean zero. 
It is negligible if our approximation is consistent in 
the sense that $\E[h_{\riesz[\psi]}^2(Z_i, \tilde m - m)] \to 0$.

The fourth term is the  the empirical inner product of two approximation errors.
It is comparable to the corresponding population inner product 
$\E[(\riesz[\psi]-\tilde\riesz)(\tilde m-m)(Z_i)]$, 
which can be analyzed deterministically using properties of the approximations.

The fifth term is the empirical inner product between the approximation error $\tilde m - m$
and $\tilde\riesz - \hriesz$, which satisfies  
$\norm{\tilde\riesz - \hriesz}_{\F} = O_P( nr^2/\sigma^2 + \norm{\tilde\riesz}_{\F})$ for $r$ as in \eqref{eq:sample-riesz-rate}
(see Appendix~\ref{sec:consistency}). We can sometimes get a useful bound on this inner product
based on the approximate orthogonality of $\tilde m - m$ to functions in $\F$.
This is natural when $\F$ is a subspace and $\tilde m$ is the $L_2(\P)$ 
orthogonal projection of $m$ onto it, as in that case $\tilde m - m$ is orthogonal to any element of $\vspan \F$.

\citet{newey2018cross}, working with subspaces $\F$ of finite sample-size-dependent dimension,
used techniques along these lines to characterize a cross-fit variant of the estimator we discuss, 
showing efficiency under near-minimal assumptions. The extension of their argument is a
promising area for future work \citep[see also][]{kennedy2020optimal}.  
\ \\

We conclude the section with a few practical considerations.

\subsection{The role of the tuning parameter $\sigma$}
\label{sec:tuning-sigma}
We generally recommend that the tuning parameter $\sigma$ be chosen without consideration of sample size. 
The simple heuristic $\sigma^2 \approx \max_{i \le n}\Var{Y_i \mid Z_i}$ arises 
from the minimax interpretation of our estimator, in which $\sigma^2$ is a bound on the conditional variance.\footnotemark\ 
However, \smash{$\hpsi_{AML}$} is fairly robust to our choice of $\sigma$, and Theorem~\ref{theo:simple-rate} 
justifies a wide range of choices.

\footnotetext{In our minimax framework in Section~\ref{sec:amle}, we also assume that $\norm{\hm - m}_{\ff} \le 1$.
If we instead believe that $\norm{\hm - m}_{\ff} \approx \alpha$, our heuristic suggests $\sigma^2 \approx \alpha^{-2} \max_{i \le n}\Var{Y_i \mid Z_i}$.}

To consider the impact of $\sigma$, we look at the role it plays in the dual characterization \eqref{eq:dual} of our weights.
As discussed above, this is a penalized least squares problem for estimating $\riesz[\psi]$.
From this perspective, taking $\sigma$ to be of constant order is regularizing very weakly, 
and we can improve the rate of convergence of $\hriesz$ to our regularized approximation $\tilde\riesz$ by increasing $\sigma$. On the other hand, 
consideration of the primal \eqref{eq:aml-primal} shows that this comes at a cost in terms of the 
maximal conditional bias $I_{h,\ff}(\hriesz)$, and if we have confidence that $\hm-m$ is in a small class \smash{$\ff$},
we can decrease $\sigma$ so that \smash{$I_{h,\ff}(\hriesz)$} and therefore our bias is zero or nearly zero.  
Recalling our discussion in Section~\ref{sec:double-robustness}, our choice of $\sigma$ essentially trades off between two properties of 
the error $\hriesz[\psi]-\riesz[\psi]$: its degree of orthogonality to the specific functions in $\ff$, and its degree of `orthogonality' to all 
square integrable functions, i.e., its magnitude $\norm{\hriesz - \riesz[\psi]}_{L_2(\P)}$.

When we choose $\sigma$ proportional to $\sqrt{n}r$, $\hpsi_{AML}$ is essentially a standard doubly robust estimator.
Our estimate of $\riesz[\psi]$ is not undersmoothed as discussed in Section~\ref{sec:double-robustness};
with this tuning, if $\norm{\riesz[\psi]}_{\ff} < \infty$, our weights converge to $\riesz[\psi]$ in empirical mean square at the rate $r$,
typically the minimax rate for estimating $\riesz[\psi]$ satisfying $\norm{\riesz[\psi]}_{\ff} < \infty$ 
(see Appendix~\ref{sec:estimating-riesz-optimal}).
The asymptotic linearity of \smash{$\hpsi_{AML}$} may then follow from the rate-product condition \smash{$\norm{\hriesz[\psi] - \riesz[\psi]}_{L_2(\Pn)}$} \smash{$\norm{\hm - m}_{L_2(\Pn)} = o_P(n^{-1/2})$},
 which is a sufficient condition when we use sample splitting to fit $\hm$.\footnotemark\ 
However, to improve our rate of convergence, we sacrifice orthogonality of $\hriesz[\psi] - \riesz[\psi]$ to possible realizations of $\hm - m$ in $\ff$. 
This makes our estimator sensitive to the rate of convergence of $\hm-m$.
We see this in our bound \eqref{eq:remainder-rate}; the term proportional to $\sigma$ will be large. 

\footnotetext{It is common to use sample splitting to fit $\hriesz[\psi]$ as well. Our bound \eqref{eq:sample-riesz-rate}
does not justify this, as it concerns empirical mean squared error on the sample used to estimate $\hriesz[\psi]$.
However, in the course of our proof in Appendix~\ref{sec:finite-sample-proofs}, we show that with this tuning,
$\hriesz[\psi]$ converges to $\riesz[\psi]$ in population mean square at the rate $r$, which is sufficient.} 
\subsection{Flexible regression adjustments and cross-fitting}
\label{sec:sample-splitting}
In some applications, we may want to base our regression adjustment on flexible, adaptive methods like boosting, random forests, or neural networks.
In this case, it may be hard to argue that $\norm{\hm - m}_{\ff} = O_P(1)$ because $\hm$ itself is irregular.
And the violation of this assumption may result in bias. For example, when we take $\ff$ to be a class of smooth functions, 
the weights $\hgamma$ that we use in $\hpsi_{AML}$ will control its bias only when $\hm-m$ is smooth. In this sense,
a nonsmooth estimator $\hm$ is incompatible with this smooth class $\ff$. This problem is easy to fix,
as we can ensure compatibility for any estimator $\hm$ simply by including it in $\ff$. A natural approach 
is to choose a class $\GG$ intended to capture $m$, and let $\ff$ be the absolutely convex hull of $\hm - \GG$. 
For this class, $\norm{\hm - m}_{\ff} \le \norm{m}_{\GG}$.

This set $\ff$ is random, presumably depending on $Y_1 \ldots Y_n$ through $\hm$,
and a problem arises because of the dependence this induces between $\hgamma_i$ and $Y_i$:
the `noise term' in \eqref{eq:error-decomp} can have nonzero mean. We can sidestep this problem by cross-fitting \citep{schick1986asymptotically}, i.e.,
fitting $\hm$ using a subsample of our observations, and defining $\hpsi_{AML}$ in terms of it on the remaining observations.
We will call the former sample the \emph{auxiliary sample} and the latter the \emph{estimation sample}. 
Asymptotic linearity can be established by Theorem~\ref{theo:simple-rate}, applied conditionally on the auxiliary sample.
We get efficiency, under the conditions stated in Theorem~\ref{theo:simple}, by averaging over multiple splits of the sample.

We can generalize this construction by training multiple candidate estimators $\hm_1 \ldots \hm_K$ on the auxiliary sample
and taking $\ff$ to be the absolutely convex hull of $\set{\hm_1 \ldots \hm_K} - \GG$.
We then define $\hpsi_{AML}$ using an estimator $\hm$ chosen from $\hm_1 \ldots \hm_K$ 
or their absolutely convex hull,
e.g., by minimizing empirical mean squared error or a targeted 
loss function \citep[see e.g.,][]{juditsky2000functional, van2003unified}. In addition 
to allowing irregular regression estimators $\hm$, this approach offers robustness to the irregularity 
of the regression function $m$ itself; recalling Section~\ref{sec:sieve-asymptotics}, 
$\norm{\hm - \tilde m}_{\ff}$ and $\norm{\tilde m - m}_{L_2(\Pn)}$
will be small for some $\tilde m$ when $m$ is approximated well by a function in $\GG$ or in $\hm_1 \ldots \hm_K$. 
In ideal conditions, the theorem below justifies the use of up to 
$K = o(n^{1/(2+\alpha)})$ candidates when $\HH \in \set{ \GG,\ \riesz[\psi]\GG,\ h(\cdot,\GG)}$
satisfy the metric entropy bound \smash{$\log \hat N(\HH,\tau) \le \tau^{-\alpha}$} for $\alpha < 2$. 

\begin{theo}
\label{theo:simple-hull}
In the setting of Theorem~\ref{theo:simple},
let $\GG \subseteq \mm $ be an absolutely convex and pointwise closed set,
and let $\ff_n$ be the absolutely convex hull of $\set{m_1 \ldots m_{K_n}} - \GG$
for $m_1 \ldots m_{K_n} \in \cspan \GG$.
Define \smash{$\hpsi_{AML}$} as in \eqref{eq:aml} with $\ff =\ff_n$.
It is asymptotically linear, satisfying \eqref{eq:asymptotic-linearity} with $\riesz[\psi]$ denoting the Riesz representer of $\psi(\cdot)$ on the tangent space $\cspan\GG$, if
\begin{enumerate}
\item $\norm{\hm - m}_{\ff_n} = O_P(1)$ and $\norm{\hm - m}_{L_2(\Pn)} = O_P(s_n)$ \ for \ $s_n \to 0$ 
\item for all $\chi \in \set{ f \to f,\ f \to \riesz[\psi]f,\ f \to h(\cdot,f)}$,
\begin{enumerate}
\item $\chi(\GG)$ is Donsker, 
\item $\sup_{f \in \ff_n}\norm{\chi(f)}_{L_p(\P)}$ is bounded uniformly in $n$ for some $p \in (2,\infty]$, 
\item when $a_n \to 0$ sufficiently slowly,
\[ (K_n+1)\hat N(\chi(\GG),\ a_n \log(K_n+1)^{-1/2}) = o_P(\omega'_{\chi,\ff_n}(n^{-1/4} \vee s_n)^{-2}). \]
\end{enumerate}
\end{enumerate}
Here $\hat N(\HH, \tau)$ is the minimal number of $\norm{\cdot}_{L_2(\Pn)}$-balls of radius $\tau$ covering $\HH$ and
\[ \omega_{\chi,\ff_n}'(r) = \omega_{\chi,\ff_n}(r) \vee r^{(p-2)/(p-1)} \ \text{ where }\
 \omega_{\chi,\ff_n}(r) =\sup_{f \in \ff_n : \norm{f}_{L_2(\P)} \le r} \norm{\chi(f)}_{L_2(\P)}. \]

\end{theo}

Candidates $\hm_1 \ldots \hm_K$ need not be good estimators of $m$ individually. We may benefit, for example, 
from including indicators for strata of estimates of $\riesz[\psi]$ and $m$, motivated 
by the ideas of propensity score and prognostic score stratification in causal inference \citep{rosenbaum1984reducing}.

\begin{rema}
\label{rema:hull-k-constant}
In the case most similar to that of Theorem~\ref{theo:simple}, in which $K_n = O(1)$
and \smash{$\sup_{f \in \ff_n}\norm{f}_{\infty}$} is bounded uniformly in $n$, the assumptions of Theorem~\ref{theo:simple-hull} essentially 
reduce to those of Theorem~\ref{theo:simple} and additional $L_p$ boundedness assumptions on $\riesz[\psi]\ff_n$ and $h(\cdot,\ff_n)$
from (2b). In particular, for any $s_n \to 0$, (2c) is implied by the equicontinuity of $h(Z,\cdot)$ on $\ff_n$ in the sense that $\lim_{r \to 0}\sup_{n}\omega_{\chi,\ff_n}(r)=0$ for $\chi(f)=h(\cdot,f)$. 
\end{rema}

\section{Estimating the Average Partial Effect in a Conditionally Linear Outcome Model}
\label{sec:ape}

As a concrete instance of our approach, we consider the problem
of estimating an average partial effect, assuming a conditionally linear treatment effect model.
A statistician observes features $X \in \xx$, a treatment dose $W \in \RR$, and an outcome
$Y \in \RR$ and wants to estimate $\psi$, where
\begin{equation}
\label{eq:ape_again}
\psi = \EE{\tau(X)} \ \text{ assuming } \ \EE{Y \cond X = x, \, W = w} = \mu(x) + w \,\tau(x).
\end{equation}
By Theorem~\ref{theo:simple}, our AML estimator will be efficient for
$\psi$ under regularity conditions when $\Var{Y_i \cond X_i, \, W_i} = v(X_i)$ is only a function of $X_i$.

In the classical case of an unconfounded binary
treatment, the model \eqref{eq:ape_again} is general and the estimand $\psi$ corresponds
to the average treatment effect \citep{rosenbaum1983central,imbens2015causal}. At the other
extreme, if $W$ is real valued but $\tau(x) = \tau$ is constrained not to depend on $x$, then
\eqref{eq:ape_again} reduces to the partially linear model as studied by \citet{robinson1988root}.
The specific model \eqref{eq:ape_again} has recently been studied by \citet*{athey2016generalized}, \citet*{graham2018semiparametrically},
and \citet*{zhao2017selective}.
We consider the motivation for \eqref{eq:ape_again} in Section \ref{sec:application} in the
context a real-world application; here, we focus on estimating $\psi$ in this model.

Both $\mu(\cdot)$ and $\tau(\cdot)$ in the model \eqref{eq:ape_again} are assumed to have
finite gauge with respect to an absolutely convex class $\HH$, and we define
\begin{equation}
\label{eq:F_ape}
\F_{\HH} = \cb{m : m(x, \, w) = \mu(x) + w \tau(x), \ \Norm{\mu}_\HH^2 + \Norm{\tau}_\HH^2 \leq 1}.
\end{equation}
We can simplify the definition \eqref{eq:aml-primal} of the minimax weights for this class. 
\begin{equation}
\label{eq:ape_gamma}
\begin{split}
\hgamma = \argmin_{\gamma \in \R^n}
\sup_{\mu \in \HH}\sqb{\frac{1}{n} \sum_{i = 1}^n \gamma_i \mu(X_i)}^2 
+ \sup_{\tau \in \HH}\sqb{\frac{1}{n} \sum_{i = 1}^n \p{W_i\gamma_i - 1} \tau(X_i)}^2 + \frac{\sigma^2 \norm{\gamma}^2}{n^2}.
\end{split}
\end{equation}
Given these weights, the augmented minimax linear estimator is  
\begin{equation}
\label{eq:ape_aml}
\hpsi_{AML} = \frac{1}{n} \sum_{i = 1}^n \p{\htau(X_i) - \hgamma_i \p{\hmu(X_i) + W_i \htau(X_i) - Y_i}}.
\end{equation}
Our formal results above give conditions under which it is asymptotically efficient.
In this section, our goal is to explore the behavior of this estimator empirically.
For comparison, we introduce some alternatives.
The first is the minimax linear estimator $\hpsi_{MLIN} = n^{-1}\sum_{i=1}^n \hgamma_i Y_i$,
i.e., $\hpsi_{AML}$ with $\hm \equiv 0$.
The others are variants of the doubly robust estimator \smash{$\hpsi_{DR}$.}
In this setting, the Riesz representer has the form $\riesz[\psi](x, w) = (w - e(x)) / v_w(x)$ with
$e(x) = \EE{W \cond X = x}$ and $v_w(x) = \Var{W \cond X = x}$, so we consider 
a natural doubly robust estimator based on plug-in estimates of these quantities,\footnote{For example,
a random forest version of this estimator is available in the \texttt{grf} package of
\citet*{athey2016generalized}. In the binary treatment assignment case $W_i \in \cb{0, \, 1}$,
we know that $v_w(x) = e(x)(1 - e(x))$; and if we set \smash{$\hv_w(x) = \he(x)(1 - \he(x))$}, then the estimator
in \eqref{eq:ape_dr} is equivalent to the augmented inverse-propensity weighted estimator of
\citet*{robins1994estimation}. For more general $W_i$, however, $v_w(x)$ is not necessarily determined by $e(x)$
and so we need to estimate it separately.}
\begin{equation}
\label{eq:ape_dr}
\hpsi_{DR} = \frac{1}{n} \sum_{i = 1}^n \p{\htau(X_i) - \p{\frac{W_i - \he(X_i)}{\hv_w(X_i)}} \p{\hmu(X_i) + W_i \htau(X_i) - Y_i}}.
\end{equation}
Below, we numerically compare the relative merits of minimax linear, augmented
minimax linear, and plug-in doubly robust estimation of the average partial effect.

\subsection{A Simulation Study}
\label{sec:simu}

To better understand the merits of different approaches to average partial effect estimation, we
conduct a simulation study. 
As baselines, we consider the {\bf plug-in doubly robust estimator} defined in \eqref{eq:ape_dr}, 
where $\he(\cdot)$ and $\hv_w(\cdot)$ are fit separately,
and an {\bf oracle doubly robust estimator} that uses the same
functional form \eqref{eq:ape_dr} but with oracle values of $e(X_i)$ and $v_w(X_i)$.
We compare these baselines to an {\bf augmented minimax linear estimator} (AML) 
that uses minimax linear weights for a class $\ff_{\hh}$ as described in \eqref{eq:ape_aml}, as well as an 
{\bf augmented minimax linear estimator over an extended class} (AML+), a variant 
that uses the same functional form but with the minimax linear weights for an extended class $\ff_{\hh_+}$
that includes a set of estimated functions. We also consider the simpler {\bf minimax linear estimator}
for each class. We provide further implementation details below.


\subsubsection{Construction of Augmented Minimax Linear Estimators}

We first describe how we implement our approach,
an augmented minimax linear estimator for the class $\ff_{\hh}$ described in the section above \eqref{eq:F_ape}.
We take $\hh$ to be the absolutely convex hull of a mean-square summable set of basis functions as described in Remark~\ref{rema:consistency}.
Specifically, we use a basis sequence $\phi_j=a_j\phi_j'$, where $\phi_j'$ are $d$-dimensional
interactions of Hermite polynomials that are orthonormal with respect to the standard normal
distribution. The sequence of weights $\cb{a_j}$ varies with order $k$ of the polynomial $\phi_j$;
\smash{$a_j =1/(k\sqrt{n_{k,d}})$} where $n_{k,d}$ is the number of terms of order $k$.
Observe that \smash{$\sum_{j=1}^{\infty} a_j^2$} \smash{$=\sum_{k=1}^{\infty}1/k^2 < \infty$} and therefore \smash{$\sum_{j=1}^{\infty}\E \phi_j^2(X) < \infty$} for standard normal $X$ or $X$ with bounded density with respect to the standard normal.

Following our discussion in Remark \ref{rema:consistency}, we take
an $\ell_1$-penalized least squares approach to estimating the regression function $m$.
Rather than using a fully nonparametric estimate $\hat m(x,w)$,
which would not be in our class $\ff_{\hh}$, 
we fit a conditionally linear model $\hmu(x) + w\htau(x)$
using the $R$-lasso method proposed by \citet{nie2017learning}.
To do this, we first estimate the marginal response function \smash{$r(x)=\EE{Y_i \cond X_i = x}$} and \smash{$e(x)$}
via a cross-validated lasso \citep{tibshirani1996regression}
on the basis $\phi(x)$.\footnotemark\ We then fit $\tau_{\beta}(x)= \phi(x)^T\beta$ 
by minimizing the $\ell_1$-penalized R-loss
$n^{-1}\sum_{i=1}^n [Y_i - \hat r(X_i) - (W-\hat e(X_i))\tau_{\beta}(X_i)]^2 + \lambda\norm{\beta}_{\ell_1}$,
with $\lambda$ chosen by cross-validation.
Finally, we set \smash{$\hmu(x) = \hat r(x)- \htau(x) \he(x)$}.
As discussed in \citet{nie2017learning}, this method is appropriate when the treatment effect function $\tau(x)$
is simpler than \smash{$r(x)$} and \smash{$e(x)$},
and allows for faster rates of convergence on $\tau(x)$ than the other regression components
whenever the nuisance components can be estimated at $o_p(n^{-1/4})$ rates in root-mean squared error.

\footnotetext{We emphasize that, although we use lasso software for fitting $\beta$, we do not follow the default
practice of standardizing the basis functions before applying the $\ell_1$-penalty. Rather, we estimate coefficients $\beta$ for the square-summable
basis \smash{$\phi_1,\phi_2,\ldots$} using a penalty proportional to $\norm{\beta}_{\ell_1}$. 
As discussed in Remark~\ref{rema:consistency}, this is penalized least squares estimation of the functions $r$ and $e$ (and $v_{w}$, which we discuss later) with a penalty proportional to the gauge of a Donsker class, where that Donsker class is the absolutely convex hull of $\phi_1,\phi_2,\ldots$.}

We consider two options for the bias-correcting weights $\hgamma$. The simpler option is to use
the minimax weights for the class $\ff_{\hh}$ described in \eqref{eq:F_ape}.
This choice is directly motivated by our formal results given in Theorem \ref{theo:simple}.
As an alternative, motivated by popular idea of propensity-stratified estimation in the causal inference
literature \citep{rosenbaum1984reducing}, we use minimax weights for an extended class $\ff_{\hh_+}$
where $\hh_+$ extends $\hh$ by adding to our basis expansion $\phi(x)$ the following random basis functions:
\begin{itemize}
\item Multi-scale strata of the estimated average treatment intensity \smash{$\he(X_i)$}
(we balanced over histogram bins of width 0.05, 0.1, and 0.2),
\item Basis elements obtained by depth-3 recursive dyadic partitioning (i.e., pick a feature, split along its median, and recurse), and
\item Leaves generated by a regression tree on the $W_i$ \citep{breiman1984classification}.
\end{itemize}
The underlying idea is that we may be able to improve the practical performance of the method by opportunistically adding
a small number of basis functions that help mitigate bias in case of misspecification (i.e., when $\mu$ and
$\tau$ do not have finite gauge \smash{$\norm{\cdot}_\hh$}). The motivation for focusing on transformations of
\smash{$\he(X_i)$} is that accurately stratifying on \smash{$e(X_i)$} would
suffice to eliminate all confounding in the model \eqref{eq:ape_again}.\footnote{In the case of binary
treatments $W_i$, this corresponds to the classical result of \citet{rosenbaum1983central}, who showed
that the propensity score is a balancing score. With non-binary treatments, \smash{$\EE{W_i \cond X_i}$}
is not in general a balancing score \citep{imbens2000role}; however, it is 
a balancing score for our specific model \eqref{eq:ape_again}.} 
Because $\F_{\hh+}$ is a function of $Z_1 \ldots Z_n$ for $Z_i=(X_i,W_i)$,
it is not necessary to cross-fit as described in 
Section~\ref{sec:sample-splitting} to avoid bias from the `noise term'. 
With both $\ff_{\hh}$ and $\ff_{\hh+}$, we take $\sigma^2=1$ in \eqref{eq:ape_gamma}.

\subsubsection{Baselines and Software Details}

The baselines we consider combine the aforementioned 
regression $\hmu(x) + w\htau(x)$ with various weighting schemes. 
The weights used in the plug-in double robust estimator \eqref{eq:ape_dr} 
involve $\he$ as estimated above and an estimate of $v_w(x)=\Var{W \mid X=x}$,
which we fit by cross-validated lasso regressing \smash{$(W_i-\he_{f_i,\hlambda_e}(X_i))^2$} on $\phi(X_i)$. 
The weights used in the double-robust oracle substitute the true values of $e(x)$ and $v_w(w)$ in our simulated design.

Ten-fold cross-fitting is used throughout: where $\htau(X_i)$ and $\hmu(X_i)$ appear in
\eqref{eq:ape_aml} and \eqref{eq:ape_dr}, we use estimators $\htau^{(-i)}$ and $\hmu^{(-i)}$
trained on the folds that do not include unit $i$. This 
reduces dependence on $(Y_i,X_i,W_i)$ and therefore mitigates potential own-observation bias 
in $\hpsi_{DR}$ \citep[see e.g.,][]{chernozhukov2016double}. However, we do get some dependence
through the estimates of $\hat r$ and $\he$ used to train $\htau$ and
through lasso tuning parameters, which are chosen once for all $i$ by cross-validation.
While this dependence could be eliminated using a computationally demanding nested sample splitting scheme,
we here follow the approach taken in the \texttt{grf} package of \citet*{athey2016generalized} and use a simplified scheme
described in Appendix~\ref{sec:simu_details}.
 Our theoretical results for $\hpsi_{AML}$ do not formally justify the use of this cross-fitting scheme,
as $\hm^{(-i)}(x,w) = \hmu^{(-i)}(x) + w\htau^{(-i)}(x)$ is a function of the fold indicator $f_i$ as
well as $x,w$, and for this reason $\norm{\hm - m}_{\ff_{\hh}} = \infty$;
however, this does not seem to cause problems in our simulations.

All methods are implemented in the \texttt{R} package \texttt{amlinear}, and replication files are
available at \url{https://github.com/davidahirshberg/amlinear}.
We computed minimax linear weights via the cone solver \texttt{ECOS} \citep*{domahidi2013ecos},
available in \texttt{R} via the package \texttt{CVXR} \citep{CVXR}.
When needed, we run penalized regression using the \texttt{R} package
\texttt{glmnet} \citep*{friedman2010regularization}.

\subsubsection{Simulation Design}
\label{sec:spec}

We considered data-generating distributions of the form
\begin{equation*}
\begin{split}
X_i \sim \nn\p{0, \, I_{d \times d}}, \ \ W_i  \cond X_i \sim \law_{X_i}, \ \
Y_i \cond X_i, \, W_i = \nn\p{b(X_i) + W_i \tau(X_i) , \, 1},
\end{split}
\end{equation*}
for different choices of
dimension $d$,
treatment assignment distribution $\law_{X_i}$,
baseline main effect $\mu(\cdot)$ and
treatment effect function $\tau(\cdot)$.
We considered the following 4 setups, each of which depends on a sparsity level
$k$ that controls the complexity of the signal.
\begin{enumerate}
\item Beta-distributed treatment,
\smash{$W_i \cond X_i \sim B(\alpha(X_i), \, 1-\alpha(X_i))$}, with
$\zeta(x) = \sum_{j = 1}^k x_{j}/\sqrt{k}$,
$\eta(x) = \sign(\zeta(x)) \zeta^2(x)$,
$\alpha(x) = \max\{0.05, \, \min\{0.95, $ $ 1/(1 + \exp[-\eta(x)]) \}\}$,
$\mu(x) = \eta(x) + 0.2 (\alpha(x) - 0.5)$, and
$\tau(x) = -0.2$.
\item Scaled Gaussian treatment,
\smash{$W_i \cond X_i \sim \nn\p{\lambda(X_i), \, \lambda^2(X_i)}$}, with
$\eta(x) = 2^{k-1} \prod_{j = 1}^k x_j$,
$\mu(x) = \sign(\eta(x)) \sqrt{\abs{\eta(x)}}$,
$\lambda(x) = 0.1 \, \text{sign}(\mu(x)) + \mu(x)$, and
$\tau(x) = \max\cb{x_{1} + x_{2}, \, 0}/2$.
\item Poisson treatment,
\smash{$W_i \cond X_i \sim \text{Poisson}(\lambda(X_i))$}, with
$\tau(x) = k^{-1} \sum_{j =1}^k$ $ \cos\p{\pi x_j /3}$,
$\lambda(x) = 0.2 + \tau^2(x)$, and
$\mu(x) = 4d^{-1}\sum_{j = 1}^d x_{j} + 2\lambda(x)$.
\item Log-normal treatment,
\smash{$\log(W_i) \cond X_i \sim \nn\p{\lambda(X_i), \, 1/3^2}$}, with
$\zeta(x) = \sum_{j = 1}^k$ $x_{j}/\sqrt{k}$,
$\mu(x) = \max\cb{0, \, 2\zeta(x)}$,
$\lambda(x) = 1 / (1 + \exp[-\sign(\zeta(x))\zeta^2(x)])$, and
$\tau(x) = \sin\p{2\pi x_{1}}$.
\end{enumerate}

\subsection{Results}

We first compare our augmented minimax linear estimators with the corresponding minimax linear estimators.
Figure \ref{fig:augment} compares the resulting mean-squared errors for $\psi$ across
several variants of the simulation design (the exact parameters used are the same as those used in Table \ref{tab:simu_results}).
The left panel shows results where the weights are minimax over $\ff_{\hh}$, while the right panel has minimax weights over $\ff_{\hh_+}$.
 
Overall, we see that the augmented minimax linear estimator is sometimes comparable to the
minimax linear one and sometimes substantially better. Thus, while results of \citet{donoho1994statistical} and \citet{armstrong2015optimal} 
imply that the augmented estimator can be little better than the minimax linear estimator for a convex signal class $\ff$ 
in terms of its behavior at a few specific signals $m \in \ff$,
this does not appear representative of behavior in general.
Furthermore, as the bias of our augmented estimator is bounded as a proportion of \smash{$\norm{\hm-m}_{\ff}$}
rather than \smash{$\norm{m}_{\ff}$}, our approach offers a natural way to accomodate signals in some non-convex signal classes: 
those for which, for some choice of $\hm$, the regression error function $\hm - m$ is well-characterized in terms of some strong norm $\norm{\cdot}_{\ff}$.
This can be the case, for example, when estimating a vector of regression coefficients $\beta$ by the lasso: 
\smash{$\norm{\hbeta - \beta}_{\ell_1}$} will be small
either if \smash{$\norm{\beta}_{\ell_1}$} is small or,
to a degree determined by incoherence properties of $\phi(X)$, if $\beta$ is sparse \citep[e.g.][]{lecue2018regularization}.
This phenomenon offers some explanation for the good behavior we observe empirically, as 
the functions $\mu(x)=\phi(x)^T \beta_{\mu}$ and $\tau(x) = \phi(x)^T \beta_{\tau}$ defining our signal $m(x,w)=\mu(x)+w\tau(x)$ have some degree of sparsity 
and \smash{$\norm{\hm - m}_{\ff_{\hh}}^2=\norm{\hbeta_{\mu} - \beta_{\mu}}_{\ell_1}^2 + \norm{\hbeta_{\tau} - \beta_{\tau}}_{\ell_1}^2$}.

\begin{figure}[t]
\begin{center}
\begin{tabular}{cc}
\includegraphics[width=0.49\textwidth]{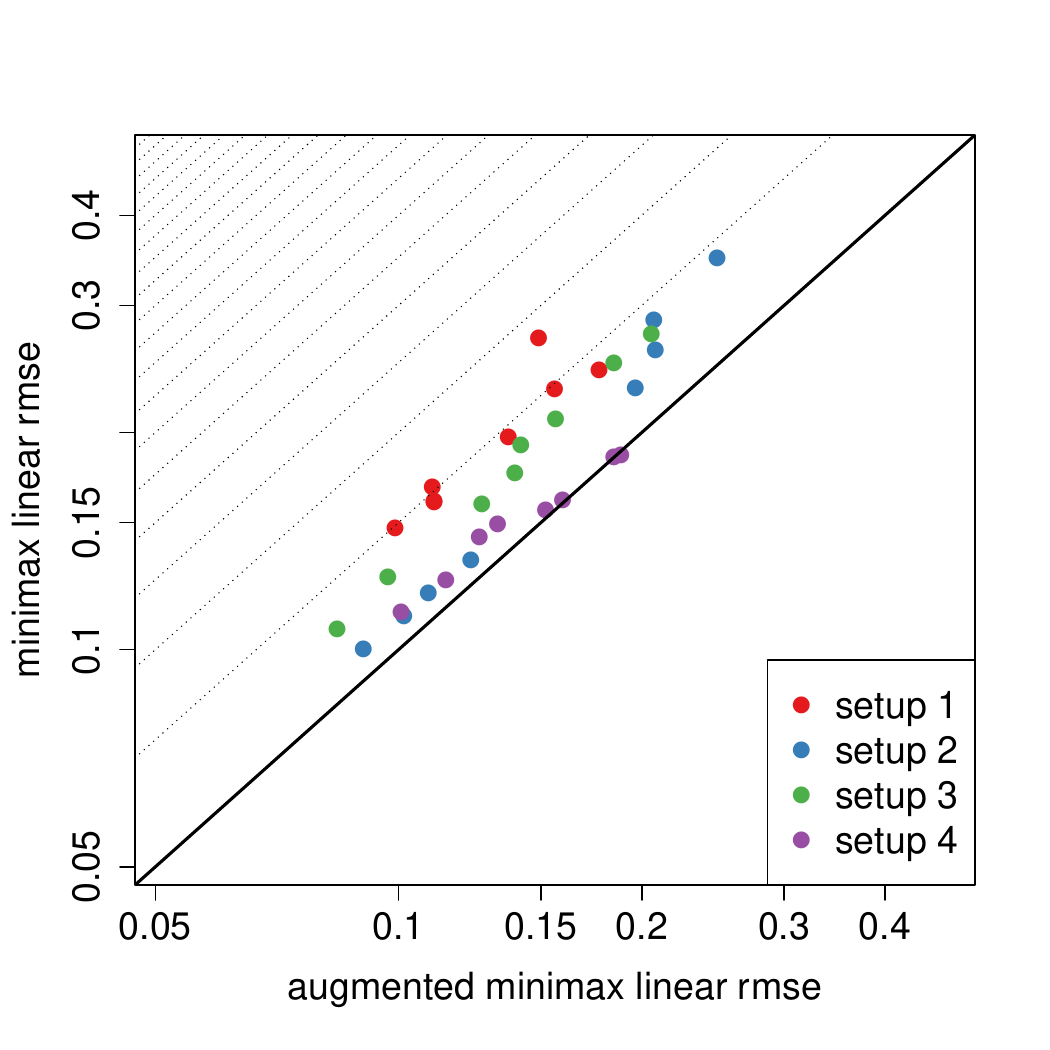} &
\includegraphics[width=0.49\textwidth]{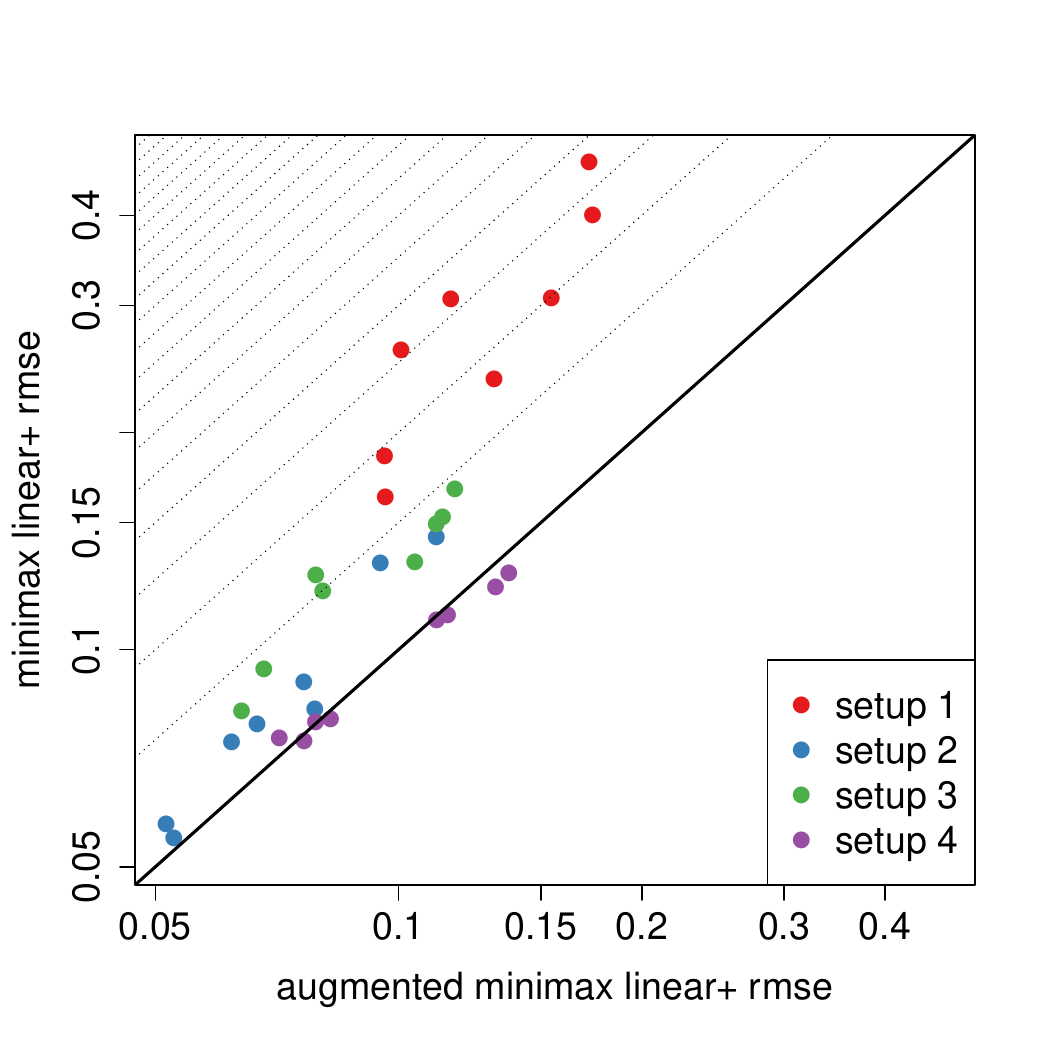}
\end{tabular}
\caption[Comparing augmented minimax linear estimation with linear estimation.]{Comparing augmented minimax linear estimation with minimax linear estimation.
The solid line $y=x$ indicates equivalent performance and the dotted lines indicate improvements of 50\%, 100\%, 150\%, etc. in root mean squared error.}
\label{fig:augment}
\end{center}
\end{figure}

In Table \ref{tab:simu_results}, we compare augmented minimax linear estimation
with doubly robust estimators, both using an estimated and an oracle Riesz representer.
In terms of mean-squared error, our simple AML estimator already performs well relative to
the main baseline (i.e., plug-in doubly robust estimation), and the AML+ estimator does better yet. Perhaps more surprisingly, our methods sometimes also
beat the doubly robust oracle, achieving comparable control of bias with a substantial decrease in variance. 
This reduction in variance arises from shrinkage due to the penalty term in \eqref{eq:aml-primal}.
It costs us little bias then because, although the oracle weights must be large to control bias
for all square integrable regression errors $\hm - m$ (i.e., to solve \ref{eq:riesz}), 
large weights are not necessary to control bias for $\hm - m$ in $\ff$ (i.e., to solve \ref{eq:riesz-sample}).

In terms of coverage, some of our simulation designs are extremely difficult and all non-oracle estimators
have substantial relative bias. However, in settings 1 and 4, the asymptotics
appear to kick in and our estimators get close to nominal coverage.

\setlength{\tabcolsep}{5pt}
\begin{table}[t]
\centering
\makebox[\textwidth]{

\begin{tabular}{||c|rrr||ccc|ccc|ccc|ccc||}
  \hline
\hline
& \multicolumn{3}{r||}{method} &
\multicolumn{3}{c|}{double rob.~plugin} &
\multicolumn{3}{c|}{augm.~minimax} &
\multicolumn{3}{c|}{augm.~minimax+} &
\multicolumn{3}{c||}{double rob.~oracle} \\ \hline
& $n$ & $p$ & $k$ & rmse & bias & covg & rmse & bias & covg & rmse & bias & covg & rmse & bias & covg \\ 
  \hline
\parbox[t]{2.8mm}{\multirow{8}{*}{\rotatebox[origin=c]{90}{setup 1}}}
 & 600 & 6 & 3 & \bf 0.13 & 0.03 & 0.98 & 0.14 & 0.03 & 0.98 & \bf 0.13 & 0.00 & 0.98 & 0.18 & -0.01 & 0.96 \\ 
   & 600 & 6 & 4 & 0.16 & 0.06 & 0.92 & 0.16 & 0.04 & 0.94 & \bf 0.15 & 0.03 & 0.93 & 0.21 & 0.00 & 0.92 \\ 
   & 600 & 12 & 3 & 0.22 & 0.09 & 0.78 & 0.18 & -0.00 & 0.87 & \bf 0.17 & 0.05 & 0.90 & 0.27 & -0.04 & 0.90 \\ 
   & 600 & 12 & 4 & 0.21 & 0.14 & 0.78 & \bf 0.15 & 0.01 & 0.94 & 0.17 & 0.09 & 0.90 & 0.23 & -0.03 & 0.93 \\ 
   & 1200 & 6 & 3 & \bf 0.10 & 0.03 & 0.94 & 0.11 & 0.06 & 0.92 & \bf 0.10 & 0.02 & 0.96 & 0.12 & 0.00 & 0.98 \\ 
   & 1200 & 6 & 4 & 0.11 & 0.03 & 0.94 & 0.11 & 0.05 & 0.92 & \bf 0.10 & 0.02 & 0.96 & 0.13 & 0.00 & 0.94 \\ 
   & 1200 & 12 & 3 & 0.11 & 0.02 & 0.90 & \bf 0.10 & 0.01 & 0.95 & \bf 0.10 & 0.02 & 0.94 & 0.14 & 0.00 & 0.94 \\ 
   & 1200 & 12 & 4 & 0.15 & 0.06 & 0.86 & \bf 0.11 & 0.00 & 0.92 & 0.12 & 0.04 & 0.90 & 0.16 & -0.00 & 0.94 \\ 
   \hline
   \parbox[t]{2.8mm}{\multirow{8}{*}{\rotatebox[origin=c]{90}{setup 2}}}
 & 600 & 6 & 1 & 0.15 & 0.12 & 0.52 & 0.11 & 0.09 & 0.74 & \bf 0.08 & 0.02 & 0.94 & 0.09 & 0.00 & 0.92 \\ 
   & 600 & 6 & 2 & 0.23 & 0.22 & 0.08 & 0.21 & 0.20 & 0.04 & \bf 0.09 & 0.07 & 0.85 & 0.10 & 0.00 & 0.94 \\ 
   & 600 & 12 & 1 & 0.16 & 0.14 & 0.44 & 0.12 & 0.11 & 0.62 & \bf 0.08 & 0.03 & 0.93 & 0.08 & 0.00 & 0.98 \\ 
   & 600 & 12 & 2 & 0.27 & 0.26 & 0.02 & 0.25 & 0.24 & 0.00 & \bf 0.11 & 0.09 & 0.76 & 0.10 & 0.01 & 0.95 \\ 
   & 1200 & 6 & 1 & 0.12 & 0.11 & 0.30 & 0.09 & 0.08 & 0.52 & \bf 0.05 & 0.01 & 0.95 & 0.06 & -0.00 & 0.96 \\ 
   & 1200 & 6 & 2 & 0.20 & 0.20 & 0.00 & 0.20 & 0.19 & 0.00 & \bf 0.06 & 0.04 & 0.90 & 0.06 & -0.00 & 0.96 \\ 
   & 1200 & 12 & 1 & 0.12 & 0.11 & 0.31 & 0.10 & 0.09 & 0.48 & \bf 0.05 & 0.01 & 0.96 & 0.06 & -0.00 & 0.98 \\ 
   & 1200 & 12 & 2 & 0.22 & 0.22 & 0.00 & 0.21 & 0.20 & 0.00 & \bf 0.07 & 0.04 & 0.86 & 0.07 & 0.00 & 0.94 \\ 
   \hline
   \parbox[t]{2.8mm}{\multirow{8}{*}{\rotatebox[origin=c]{90}{setup 3}}}
 & 600 & 6 & 3 & 0.23 & 0.23 & 0.04 & 0.14 & 0.13 & 0.44 & \bf 0.11 & 0.09 & 0.72 & 0.08 & -0.00 & 0.96 \\ 
   & 600 & 6 & 4 & 0.20 & 0.20 & 0.12 & 0.13 & 0.11 & 0.54 & \bf 0.10 & 0.09 & 0.72 & 0.07 & -0.00 & 0.96 \\ 
   & 600 & 12 & 3 & 0.25 & 0.24 & 0.03 & 0.21 & 0.20 & 0.10 & \bf 0.12 & 0.10 & 0.70 & 0.08 & -0.01 & 0.95 \\ 
   & 600 & 12 & 4 & 0.21 & 0.20 & 0.09 & 0.18 & 0.17 & 0.16 & \bf 0.11 & 0.10 & 0.72 & 0.08 & -0.01 & 0.94 \\ 
   & 1200 & 6 & 3 & 0.20 & 0.19 & 0.01 & 0.10 & 0.09 & 0.55 & \bf 0.07 & 0.05 & 0.78 & 0.05 & -0.01 & 0.97 \\ 
   & 1200 & 6 & 4 & 0.18 & 0.18 & 0.01 & 0.08 & 0.07 & 0.68 & \bf 0.06 & 0.05 & 0.85 & 0.05 & -0.01 & 0.96 \\ 
   & 1200 & 12 & 3 & 0.23 & 0.22 & 0.00 & 0.16 & 0.15 & 0.02 & \bf 0.08 & 0.07 & 0.76 & 0.05 & -0.00 & 0.96 \\ 
   & 1200 & 12 & 4 & 0.19 & 0.19 & 0.00 & 0.14 & 0.14 & 0.13 & \bf 0.08 & 0.07 & 0.70 & 0.05 & 0.00 & 0.94 \\ 
   \hline
   \parbox[t]{2.8mm}{\multirow{8}{*}{\rotatebox[origin=c]{90}{setup 4}}}
 & 600 & 6 & 4 & 0.22 & 0.16 & 0.84 & 0.16 & -0.03 & 0.94 & \bf 0.11 & -0.02 & 1.00 & 0.16 & 0.03 & 0.94 \\ 
   & 600 & 6 & 5 & 0.20 & 0.14 & 0.88 & 0.15 & -0.05 & 0.93 & \bf 0.11 & -0.02 & 1.00 & 0.15 & 0.00 & 0.93 \\ 
   & 600 & 12 & 4 & 0.23 & 0.15 & 0.86 & 0.18 & -0.09 & 0.88 & \bf 0.14 & -0.04 & 0.96 & 0.17 & -0.01 & 0.91 \\ 
   & 600 & 12 & 5 & 0.24 & 0.17 & 0.82 & 0.19 & -0.09 & 0.89 & \bf 0.13 & -0.05 & 0.97 & 0.17 & -0.01 & 0.94 \\ 
   & 1200 & 6 & 4 & 0.13 & 0.09 & 0.90 & 0.10 & -0.03 & 0.94 & \bf 0.07 & -0.01 & 1.00 & 0.10 & 0.00 & 0.96 \\ 
   & 1200 & 6 & 5 & 0.14 & 0.08 & 0.91 & 0.11 & -0.05 & 0.94 & \bf 0.08 & -0.01 & 1.00 & 0.11 & 0.00 & 0.94 \\ 
   & 1200 & 12 & 4 & 0.14 & 0.08 & 0.88 & 0.13 & -0.07 & 0.88 & \bf 0.08 & -0.02 & 0.98 & 0.11 & -0.00 & 0.94 \\ 
   & 1200 & 12 & 5 & 0.14 & 0.09 & 0.87 & 0.13 & -0.07 & 0.90 & \bf 0.08 & -0.02 & 1.00 & 0.11 & -0.00 & 0.96 \\ 
   \hline
\hline
\end{tabular}

}
\caption[Performance of AMLE and baselines in simulation.]{Performance of 4 methods described in Section \ref{sec:simu} on
the simulation designs from Section \ref{sec:spec}. We report root-mean squared error,
bias, and coverage of 95\% confidence intervals averaged over 200 simulation replications.}
\label{tab:simu_results}
\end{table}
\setlength{\tabcolsep}{6pt}

\section{The Effect of Lottery Winnings on Earnings}
\label{sec:application}

To test the behavior of our method in practice, we revisit a study of
\citet*{imbens2001estimating} on the effect of lottery winnings on
long-term earnings. It is of considerably policy interest to understand
how people react to reliable sources of unearned income; such questions
come up, for example, in discussing how universal basic income would affect employment.
In an attempt to get some insight about this effect, \citet*{imbens2001estimating}
study a sample of people who won a major lottery whose prize is paid out in installments
over 20 years. The authors then ask how \$1 in yearly lottery
income affects the earnings of the winner.

To do so, the authors consider $n = 194$ people who all won the lottery,
but got prizes of different sizes (\$1,000--\$100,000 per year).\footnote{The
paper also considers some people who won very large prizes (more than \$100k per year) and
some who won smaller prizes (not paid in installments); however, we restrict our analysis
to the smaller sample of people who won prizes paid out in installments worth \$1k--\$100k per year.}
They effectively use a causal model $\E[Y_i(w) \mid X_i=x] = m(x) + \tau w$
for observations $Y_i=Y_i(W_i)$ of the average yearly earnings in the 6 years following winning
$W_i$ in yearly lottery payoff, where $X_i$ denotes a set of $p = 12$ pre-win covariates
(year won, number of tickets bought, age at win, gender, education, whether
employed at time of win, earnings in 6 years prior to win). Here $Y_i(w)$
represents the average yearly earnings that would have occurred had, possibly contrary to fact,
unit $i$ won a prize paying $w$ dollars annually \citep[e.g.,][]{imbens2015causal}. 
The authors also consider several other model specifications.

As discussed at length by \citet*{imbens2001estimating}, although the lottery winnings
were presumably randomly assigned, we cannot assume exogeneity of the form
$W_i \indep \set{ Y_i(w) : w \in \R}$ because of survey non-response. The data was collected by
mailing out surveys to lottery winners asking about their earnings, etc., so
there may have been selection effects in who responded to the survey. 
A response rate of 42\% was observed, and older people with big 
winnings appear to have been relatively more likely to respond than young people with
big winnings. For this reason, the authors only assume exogeneity
conditionally on the covariates, i.e., $W_i \indep \set{ Y_i(w) : w \in \R} \cond X_i$,
which suffices to establish that the aforementioned causal model is identified 
as a regression model $m(x) + \tau w = \E[Y_i \mid X_i=x, W_i=w]$.

Here, we examine the robustness of the conclusions of \citet*{imbens2001estimating}
to potential effect heterogeneity. Instead of assuming that the slope $\tau$ in this model
 is a constant, we let it vary with $x$ and seek to estimate $\psi = \EE{\tau(X)}$;
this corresponds exactly to an average partial effect in the conditionally linear model, which we
studied in Section \ref{sec:ape}. In our comparison, we consider 3 estimators that implicitly
assume constant slope and estimate $\tau$, and 6 that allow
$\tau(x)$ to vary and estimate $\EE{\tau(X)}$.

Among methods that assume constant slope, the first runs ordinary least squares for $Y_i$ on $W_i$, ignoring potential confounding due to
non-response. The second, which most closely resembles the method used by \citet*{imbens2001estimating},
controls for the $X_i$ using ordinary least squares, i.e., it regresses $Y_i$ on $(X_i, W_i)$ and considers the coefficient on $W_i$.
The third uses the method of \citet{robinson1988root} with cross-fitting as in \citet{chernozhukov2016double}:
it first estimates the marginal effect of $X_i$ on $W_i$ and $Y_i$ via a non-parametric adjustment and then
regresses residuals \smash{$Y_i - \hEE{Y_i \cond X_i}$} on \smash{$W_i - \hEE{W_i \cond X_i}$}.
In each case, we report robust standard errors obtained via the \texttt{R}-package \texttt{sandwich} \citep{sandwich}.

The 6 methods that allow for treatment effect heterogeneity correspond to the 5 methods discussed
in Section \ref{sec:ape}, along with a pure weighting estimator using the estimated Riesz representer,
\smash{$\hpsi = n^{-1} \sum_{i = 1}^n \hriesz[\psi](X_i) Y_i$}, with the same choice of
\smash{$\hriesz[\psi](\cdot)$} as used in \eqref{eq:ape_dr}.
For all non-parametric regression adjustments, we run penalized regression as in Section \ref{sec:ape},
on a basis obtained by taking order-3 Hermite interactions of the 10 continuous features, and then
creating full interactions with the two binary variables (gender and employment), resulting in a total
of 1140 basis elements. For AML+, we include propensity
strata of widths $0.05$, $0.1$, and $0.2$ in the class $\hh_+$ .

\begin{table}[t]
\begin{center}
\begin{tabular}{|ll|cc|}
  \hline
  estimand & estimator & estimate & std.~err \\ 
  \hline
  partial effect & OLS without controls & -0.176 & 0.039 \\ 
  partial effect & OLS with controls & -0.106 & 0.032 \\ 
  partial effect & residual-on-residual OLS & -0.110 & 0.032 \\ 
  avg.~partial effect & plugin Riesz weighting & -0.175 & --- \\ 
  avg.~partial effect & doubly robust plugin & -0.108 & 0.042 \\ 
  avg.~partial effect & minimax linear weighting & -0.074 & --- \\ 
  avg.~partial effect & augm.~minimax linear & -0.091 & 0.044 \\ 
  avg.~partial effect & minimax linear+ weighting & -0.083 & --- \\ 
  avg.~partial effect & augm.~minimax linear+ & -0.097 & 0.045 \\ 
   \hline
\end{tabular}
\caption[Estimates for the effect of unearned income on earnings using data from \citet*{imbens2001estimating}.]{Various estimates, estimators, and estimands for the effect of unearned income on earnings,
using the dataset of \citet*{imbens2001estimating}.
The first 3 methods are justified under the assumption of no heterogeneity in $\tau(x)$ (i.e., $\tau(x) = \tau$),
and estimate $\tau$, while the latter 6 allow for heterogeneity and estimate $\EE{\tau(X)}$.}
\label{tab:IRS}
\end{center}
\end{table}

Table \ref{tab:IRS} reports results using the 9 estimators described above, along with standard
error estimates. We do not report standard errors for the 3 pure weighting methods, as these may not be
asymptotically unbiased and so confidence intervals should also account for bias.
The reported estimates are unitless; in other words, the majority of the estimators suggest that
survey respondents on average respond to a \$1 increase in unearned yearly income by reducing
their yearly earnings by roughly \$0.10.

Substantively, it appears reassuring that most point estimates are consistent
with each other, whether or not they allow for heterogeneity in $\tau(x)$. The only two divergent estimators
are the one that doesn't control for confounding at all,
and the one that uses pure plug-in weighting (which may simply be unstable here).
From a methodological perspective, it is encouraging that our method (and here, also the plug-in doubly robust method)
can rigorously account for potential heterogeneity in $\tau(x)$ without excessively inflating uncertainty.

\section*{Acknowledgments}

We are grateful for stimulating
discussions with Timothy Armstrong, Vitor Hadad, Guido Imbens, Whitney Newey, Jamie Robins, Florian Stebegg, and Jos\'e Zubizarreta,
as well as for comments from seminar participants at several venues.
We also thank Guido Imbens for sharing the lottery data with us.
We initiated this research while D.H.~was a Ph.D. candidate at Columbia University and S.W.~was visiting Columbia as a postdoctoral research scientist.

\ifaos
\else
\bibliographystyle{plainnat-abbrev}
\fi
\bibliography{references}

\begin{thebibliography}{94}
\providecommand{\natexlab}[1]{#1}
\providecommand{\url}[1]{\texttt{#1}}
\expandafter\ifx\csname urlstyle\endcsname\relax
  \providecommand{\doi}[1]{doi: #1}\else
  \providecommand{\doi}{doi: \begingroup \urlstyle{rm}\Url}\fi

\bibitem[Aliprantis and Border(2006)]{aliprantis2006infinite}
C.~D. Aliprantis and K.~C. Border.
\newblock \emph{Infinite Dimensional Analysis: a Hitchhiker's Guide}.
\newblock Springer, Berlin; London, 2006.
\newblock ISBN 9783540326960 3540326960.
\newblock \doi{10.1007/3-540-29587-9}.

\bibitem[Argyriou and Dinuzzo(2014)]{argyriou2014unifying}
A.~Argyriou and F.~Dinuzzo.
\newblock A unifying view of representer theorems.
\newblock In \emph{International Conference on Machine Learning}, pages
  748--756, 2014.

\bibitem[Armstrong and Koles{\'a}r(2018)]{armstrong2015optimal}
T.~B. Armstrong and M.~Koles{\'a}r.
\newblock Optimal inference in a class of regression models.
\newblock \emph{Econometrica}, 86\penalty0 (2):\penalty0 655--683, 2018.

\bibitem[Armstrong and Koles{\'a}r(2017)]{armstrong2017finite}
T.~B. Armstrong and M.~Koles{\'a}r.
\newblock Finite-sample optimal estimation and inference on average treatment
  effects under unconfoundedness.
\newblock \emph{arXiv preprint arXiv:1712.04594}, 2017.

\bibitem[Athey et~al.(2018)Athey, Imbens, and Wager]{athey2016approximate}
S.~Athey, G.~W. Imbens, and S.~Wager.
\newblock Approximate residual balancing: Debiased inference of average
  treatment effects in high dimensions.
\newblock \emph{Journal of the Royal Statistical Society: Series B (Statistical
  Methodology)}, 80\penalty0 (4):\penalty0 597--623, 2018.

\bibitem[Athey et~al.(2019)Athey, Tibshirani, and Wager]{athey2016generalized}
S.~Athey, J.~Tibshirani, and S.~Wager.
\newblock Generalized random forests.
\newblock \emph{The Annals of Statistics}, 47\penalty0 (2):\penalty0
  1148--1178, 2019.

\bibitem[Bartlett et~al.(2005)Bartlett, Bousquet, and
  Mendelson]{bartlett2005local}
P.~L. Bartlett, O.~Bousquet, and S.~Mendelson.
\newblock Local rademacher complexities.
\newblock \emph{The Annals of Statistics}, 33\penalty0 (4):\penalty0
  1497--1537, 2005.

\bibitem[Bickel et~al.(1998)Bickel, Klaassen, Ritov, and Wellner]{bickel}
P.~Bickel, C.~Klaassen, Y.~Ritov, and J.~Wellner.
\newblock \emph{Efficient and Adaptive Estimation for Semiparametric Models}.
\newblock Springer-Verlag, 1998.

\bibitem[Bousquet et~al.(2002)Bousquet, Koltchinskii, and
  Panchenko]{bousquet2002some}
O.~Bousquet, V.~Koltchinskii, and D.~Panchenko.
\newblock Some local measures of complexity of convex hulls and generalization
  bounds.
\newblock In \emph{International Conference on Computational Learning Theory},
  pages 59--73. Springer, 2002.

\bibitem[Breiman et~al.(1984)Breiman, Friedman, Stone, and
  Olshen]{breiman1984classification}
L.~Breiman, J.~Friedman, C.~J. Stone, and R.~A. Olshen.
\newblock \emph{Classification and Regression Trees}.
\newblock CRC press, 1984.

\bibitem[Cai and Low(2003)]{cai2003note}
T.~T. Cai and M.~G. Low.
\newblock A note on nonparametric estimation of linear functionals.
\newblock \emph{Annals of Statistics}, pages 1140--1153, 2003.

\bibitem[Cand\`es and Tao(2007)]{candes2007dantzig}
E.~Cand\`es and T.~Tao.
\newblock The {D}antzig selector: Statistical estimation when $p$ is much
  larger than $n$.
\newblock \emph{The Annals of Statistics}, pages 2313--2351, 2007.

\bibitem[Cassel et~al.(1976)Cassel, S{\"a}rndal, and Wretman]{cassel1976some}
C.~M. Cassel, C.~E. S{\"a}rndal, and J.~H. Wretman.
\newblock Some results on generalized difference estimation and generalized
  regression estimation for finite populations.
\newblock \emph{Biometrika}, 63\penalty0 (3):\penalty0 615--620, 1976.

\bibitem[Chan et~al.(2015)Chan, Yam, and Zhang]{chan2015globally}
K.~C.~G. Chan, S.~C.~P. Yam, and Z.~Zhang.
\newblock Globally efficient non-parametric inference of average treatment
  effects by empirical balancing calibration weighting.
\newblock \emph{Journal of the Royal Statistical Society: Series B (Statistical
  Methodology)}, 2015.

\bibitem[Chen et~al.(2008)Chen, Hong, and Tarozzi]{chen2008semiparametric}
X.~Chen, H.~Hong, and A.~Tarozzi.
\newblock Semiparametric efficiency in {GMM} models with auxiliary data.
\newblock \emph{The Annals of Statistics}, pages 808--843, 2008.

\bibitem[Chernozhukov et~al.(2016)Chernozhukov, Escanciano, Ichimura, and
  Newey]{chernozhukov2016locally}
V.~Chernozhukov, J.~C. Escanciano, H.~Ichimura, and W.~K. Newey.
\newblock Locally robust semiparametric estimation.
\newblock \emph{arXiv preprint arXiv:1608.00033}, 2016.

\bibitem[Chernozhukov et~al.(2018{\natexlab{a}})Chernozhukov, Chetverikov,
  Demirer, Duflo, Hansen, Newey, and Robins]{chernozhukov2016double}
V.~Chernozhukov, D.~Chetverikov, M.~Demirer, E.~Duflo, C.~Hansen, W.~Newey, and
  J.~Robins.
\newblock Double/debiased machine learning for treatment and structural
  parameters.
\newblock \emph{The Econometrics Journal}, 21\penalty0 (1):\penalty0 C1--C68,
  2018{\natexlab{a}}.

\bibitem[Chernozhukov et~al.(2018{\natexlab{b}})Chernozhukov, Newey, and
  Robins]{chernozhukov2018double}
V.~Chernozhukov, W.~Newey, and J.~Robins.
\newblock Double/de-biased machine learning using regularized riesz
  representers.
\newblock \emph{arXiv preprint arXiv:1802.08667}, 2018{\natexlab{b}}.

\bibitem[Combari et~al.(1996)Combari, Laghdir, and Thibault]{combari1996note}
C.~Combari, M.~Laghdir, and L.~Thibault.
\newblock A note on subdifferentials of convex composite functionals.
\newblock \emph{Archiv der Mathematik}, 67\penalty0 (3):\penalty0 239--252,
  1996.

\bibitem[Crump et~al.(2009)Crump, Hotz, Imbens, and Mitnik]{crump}
R.~K. Crump, V.~J. Hotz, G.~W. Imbens, and O.~A. Mitnik.
\newblock Dealing with limited overlap in estimation of average treatment
  effects.
\newblock \emph{Biometrika}, page asn055, 2009.

\bibitem[D'Amour et~al.(2017)D'Amour, Ding, Feller, Lei, and
  Sekhon]{d2017overlap}
A.~D'Amour, P.~Ding, A.~Feller, L.~Lei, and J.~Sekhon.
\newblock Overlap in observational studies with high-dimensional covariates.
\newblock \emph{arXiv preprint arXiv:1711.02582}, 2017.

\bibitem[Domahidi et~al.(2013)Domahidi, Chu, and Boyd]{domahidi2013ecos}
A.~Domahidi, E.~Chu, and S.~Boyd.
\newblock {ECOS}: {A}n {SOCP} solver for embedded systems.
\newblock In \emph{European Control Conference (ECC)}, pages 3071--3076, 2013.

\bibitem[Donoho(1994)]{donoho1994statistical}
D.~L. Donoho.
\newblock Statistical estimation and optimal recovery.
\newblock \emph{The Annals of Statistics}, pages 238--270, 1994.

\bibitem[Donoho and Liu(1991)]{donoho1991geometrizing}
D.~L. Donoho and R.~C. Liu.
\newblock Geometrizing rates of convergence, {III}.
\newblock \emph{The Annals of Statistics}, pages 668--701, 1991.

\bibitem[Fan et~al.(2016)Fan, Imai, Liu, Ning, and Yang]{fan2016improving}
J.~Fan, K.~Imai, H.~Liu, Y.~Ning, and X.~Yang.
\newblock Improving covariate balancing propensity score: A doubly robust and
  efficient approach.
\newblock Technical report, Technical report, Princeton Univ, 2016.

\bibitem[Friedman et~al.(2010)Friedman, Hastie, and
  Tibshirani]{friedman2010regularization}
J.~Friedman, T.~Hastie, and R.~Tibshirani.
\newblock Regularization paths for generalized linear models via coordinate
  descent.
\newblock \emph{Journal of Statistical Software}, 33\penalty0 (1):\penalty0 1,
  2010.

\bibitem[Fu et~al.(2017)Fu, Narasimhan, Diamond, and Miller]{CVXR}
A.~Fu, B.~Narasimhan, S.~Diamond, and J.~Miller.
\newblock \emph{CVXR: Disciplined Convex Optimization}, 2017.
\newblock URL \url{https://CRAN.R-project.org/package=CVXR}.
\newblock R package version 0.94-4.

\bibitem[Gin{\'e} and Nickl(2015)]{gine2015mathematical}
E.~Gin{\'e} and R.~Nickl.
\newblock \emph{Mathematical foundations of infinite-dimensional statistical
  models}.
\newblock Cambridge University Press, 2015.

\bibitem[Graham et~al.(2012)Graham, Pinto, and Egel]{graham1}
B.~Graham, C.~Pinto, and D.~Egel.
\newblock Inverse probability tilting for moment condition models with missing
  data.
\newblock \emph{Review of Economic Studies}, pages 1053--1079, 2012.

\bibitem[Graham and Pinto(2018)]{graham2018semiparametrically}
B.~S. Graham and C.~C. d.~X. Pinto.
\newblock Semiparametrically efficient estimation of the average linear
  regression function.
\newblock Technical report, National Bureau of Economic Research, 2018.

\bibitem[Graham et~al.(2016)Graham, Pinto, and Egel]{graham2}
B.~S. Graham, C.~C. d.~X. Pinto, and D.~Egel.
\newblock Efficient estimation of data combination models by the method of
  auxiliary-to-study tilting ({AST}).
\newblock \emph{Journal of Business \& Economic Statistics}, 34\penalty0
  (2):\penalty0 288--301, 2016.

\bibitem[Gy{\"o}rfi et~al.(2006)Gy{\"o}rfi, Kohler, Krzyzak, and
  Walk]{gyorfi2006distribution}
L.~Gy{\"o}rfi, M.~Kohler, A.~Krzyzak, and H.~Walk.
\newblock \emph{A distribution-free theory of nonparametric regression}.
\newblock Springer Science \& Business Media, 2006.

\bibitem[Hainmueller(2012)]{hainmueller}
J.~Hainmueller.
\newblock Entropy balancing for causal effects: A multivariate reweighting
  method to produce balanced samples in observational studies.
\newblock \emph{Political Analysis}, 20\penalty0 (1):\penalty0 25--46, 2012.

\bibitem[Hirano et~al.(2003)Hirano, Imbens, and Ridder]{hirano2003efficient}
K.~Hirano, G.~W. Imbens, and G.~Ridder.
\newblock Efficient estimation of average treatment effects using the estimated
  propensity score.
\newblock \emph{Econometrica}, 71\penalty0 (4):\penalty0 1161--1189, 2003.

\bibitem[Hirshberg et~al.(2019)Hirshberg, Maleki, and
  Zubizarreta]{hirshberg2019minimax}
D.~A. Hirshberg, A.~Maleki, and J.~Zubizarreta.
\newblock Minimax linear estimation of the retargeted mean.
\newblock \emph{arXiv preprint arXiv:1901.10296}, 2019.

\bibitem[Ibragimov and Khas'minskii(1985)]{ibragimov1985nonparametric}
I.~A. Ibragimov and R.~Z. Khas'minskii.
\newblock On nonparametric estimation of the value of a linear functional in
  {G}aussian white noise.
\newblock \emph{Theory of Probability \& Its Applications}, 29\penalty0
  (1):\penalty0 18--32, 1985.

\bibitem[Imai and Ratkovic(2014)]{imai2014covariate}
K.~Imai and M.~Ratkovic.
\newblock Covariate balancing propensity score.
\newblock \emph{Journal of the Royal Statistical Society: Series B (Statistical
  Methodology)}, 76\penalty0 (1):\penalty0 243--263, 2014.

\bibitem[Imbens and Wager(2019)]{imbens2017optimized}
G.~Imbens and S.~Wager.
\newblock Optimized regression discontinuity designs.
\newblock \emph{Review of Economics and Statistics}, 101\penalty0 (2):\penalty0
  264--278, 2019.

\bibitem[Imbens(2000)]{imbens2000role}
G.~W. Imbens.
\newblock The role of the propensity score in estimating dose-response
  functions.
\newblock \emph{Biometrika}, 87\penalty0 (3):\penalty0 706--710, 2000.

\bibitem[Imbens and Rubin(2015)]{imbens2015causal}
G.~W. Imbens and D.~B. Rubin.
\newblock \emph{Causal Inference in Statistics, Social, and Biomedical
  Sciences}.
\newblock Cambridge University Press, 2015.

\bibitem[Imbens et~al.(2001)Imbens, Rubin, and Sacerdote]{imbens2001estimating}
G.~W. Imbens, D.~B. Rubin, and B.~I. Sacerdote.
\newblock Estimating the effect of unearned income on labor earnings, savings,
  and consumption: Evidence from a survey of lottery players.
\newblock \emph{American Economic Review}, 91\penalty0 (4):\penalty0 778--794,
  2001.

\bibitem[Javanmard and Montanari(2014)]{javanmard2014confidence}
A.~Javanmard and A.~Montanari.
\newblock Confidence intervals and hypothesis testing for high-dimensional
  regression.
\newblock \emph{The Journal of Machine Learning Research}, 15\penalty0
  (1):\penalty0 2869--2909, 2014.

\bibitem[Johnstone(2015)]{johnstone2015-gaussian-sequence}
I.~M. Johnstone.
\newblock Gaussian estimation: Sequence and wavelet models.
\newblock \emph{Manuscript}, 2015.

\bibitem[Juditsky and Nemirovski(2000)]{juditsky2000functional}
A.~Juditsky and A.~Nemirovski.
\newblock Functional aggregation for nonparametric regression.
\newblock \emph{The Annals of Statistics}, 28\penalty0 (3):\penalty0 681--712,
  2000.

\bibitem[Juditsky and Nemirovski(2009)]{juditsky2009nonparametric}
A.~B. Juditsky and A.~S. Nemirovski.
\newblock Nonparametric estimation by convex programming.
\newblock \emph{The Annals of Statistics}, 37\penalty0 (5A):\penalty0
  2278--2300, 2009.

\bibitem[Kallus(2020)]{kallus2016generalized}
N.~Kallus.
\newblock Generalized optimal matching methods for causal inference.
\newblock \emph{Journal of Machine Learning Research}, 21\penalty0
  (62):\penalty0 1--54, 2020.

\bibitem[Kallus(2018)]{kallus2018balanced}
N.~Kallus.
\newblock Balanced policy evaluation and learning.
\newblock In \emph{Advances in Neural Information Processing Systems}, pages
  8909--8920, 2018.

\bibitem[Kennedy(2020)]{kennedy2020optimal}
E.~H. Kennedy.
\newblock Optimal doubly robust estimation of heterogeneous causal effects.
\newblock \emph{arXiv preprint arXiv:2004.14497}, 2020.

\bibitem[Koltchinskii(2006)]{koltchinskii2006local}
V.~Koltchinskii.
\newblock Local rademacher complexities and oracle inequalities in risk
  minimization.
\newblock \emph{The Annals of Statistics}, 34\penalty0 (6):\penalty0
  2593--2656, 2006.

\bibitem[Lang(1993)]{lang1993real}
S.~Lang.
\newblock \emph{Real and functional analysis}.
\newblock Springer-Verlag, New York, 1993.

\bibitem[Lecu{\'e} and Mendelson(2017)]{lecue2017regularization}
G.~Lecu{\'e} and S.~Mendelson.
\newblock Regularization and the small-ball method ii: complexity dependent
  error rates.
\newblock \emph{Journal of Machine Leaning Research}, 18\penalty0
  (146):\penalty0 1--48, 2017.

\bibitem[Lecu{\'e} and Mendelson(2018)]{lecue2018regularization}
G.~Lecu{\'e} and S.~Mendelson.
\newblock Regularization and the small-ball method i: sparse recovery.
\newblock \emph{The Annals of Statistics}, 46\penalty0 (2):\penalty0 611--641,
  2018.

\bibitem[Ledoux and Talagrand(1991)]{ledoux1991probability}
M.~Ledoux and M.~Talagrand.
\newblock \emph{Probability in Banach Spaces: isoperimetry and processes}.
\newblock Springer, 1991.

\bibitem[Li et~al.(2018)Li, Morgan, and Zaslavsky]{li2018balancing}
F.~Li, K.~L. Morgan, and A.~M. Zaslavsky.
\newblock Balancing covariates via propensity score weighting.
\newblock \emph{Journal of the American Statistical Association}, 113\penalty0
  (521):\penalty0 390--400, 2018.

\bibitem[Lugosi and Zeger(1995)]{lugosi1995nonparametric}
G.~Lugosi and K.~Zeger.
\newblock Nonparametric estimation via empirical risk minimization.
\newblock \emph{IEEE Transactions on Information Theory}, 41\penalty0
  (3):\penalty0 677--687, 1995.

\bibitem[Massart(2000)]{massart2000some}
P.~Massart.
\newblock Some applications of concentration inequalities to statistics.
\newblock In \emph{Annales-Faculte des Sciences Toulouse Mathematiques},
  volume~9, pages 245--303. Universit{\'e} Paul Sabatier, 2000.

\bibitem[Megginson(2012)]{megginson2012introduction}
R.~E. Megginson.
\newblock \emph{An introduction to Banach space theory}, volume 183.
\newblock Springer Science \& Business Media, 2012.

\bibitem[Mendelson(2017)]{mendelson2017extending}
S.~Mendelson.
\newblock Extending the small-ball method.
\newblock \emph{arXiv preprint arXiv:1709.00843}, 2017.

\bibitem[Mukherjee et~al.(2017)Mukherjee, Newey, and
  Robins]{mukherjee2017semiparametric}
R.~Mukherjee, W.~K. Newey, and J.~M. Robins.
\newblock Semiparametric efficient empirical higher order influence function
  estimators.
\newblock \emph{arXiv preprint arXiv:1705.07577}, 2017.

\bibitem[Newey(1994)]{newey1994asymptotic}
W.~K. Newey.
\newblock The asymptotic variance of semiparametric estimators.
\newblock \emph{Econometrica}, 62\penalty0 (6):\penalty0 1349--1382, 1994.

\bibitem[Newey and Robins(2018)]{newey2018cross}
W.~K. Newey and J.~R. Robins.
\newblock Cross-fitting and fast remainder rates for semiparametric estimation.
\newblock \emph{arXiv preprint arXiv:1801.09138}, 2018.

\bibitem[Nie and Wager(2017)]{nie2017learning}
X.~Nie and S.~Wager.
\newblock Quasi-oracle estimation of heterogeneous treatment effects.
\newblock \emph{arXiv preprint arXiv:1712.04912}, 2017.

\bibitem[Ning et~al.(2017)Ning, Peng, and Imai]{ning2017high}
Y.~Ning, S.~Peng, and K.~Imai.
\newblock High dimensional propensity score estimation via covariate balancing,
  2017.

\bibitem[Peypouquet(2015)]{peypouquet2015convex}
J.~Peypouquet.
\newblock \emph{Convex Optimization in Normed Spaces: Theory, Methods and
  Examples}.
\newblock Springer, 2015.

\bibitem[Powell et~al.(1989)Powell, Stock, and
  Stoker]{powell1989semiparametric}
J.~L. Powell, J.~H. Stock, and T.~M. Stoker.
\newblock Semiparametric estimation of index coefficients.
\newblock \emph{Econometrica}, pages 1403--1430, 1989.

\bibitem[Robins and Rotnitzky(1995)]{robins1}
J.~Robins and A.~Rotnitzky.
\newblock Semiparametric efficiency in multivariate regression models with
  missing data.
\newblock \emph{Journal of the American Statistical Association}, 90\penalty0
  (1):\penalty0 122--129, 1995.

\bibitem[Robins et~al.(2007)Robins, Sued, Lei-Gomez, and
  Rotnitzky]{robins2007comment}
J.~Robins, M.~Sued, Q.~Lei-Gomez, and A.~Rotnitzky.
\newblock Comment: Performance of double-robust estimators when ``inverse
  probability'' weights are highly variable.
\newblock \emph{Statistical Science}, 22\penalty0 (4):\penalty0 544--559, 2007.

\bibitem[Robins et~al.(2009)Robins, Tchetgen, Li, and van~der
  Vaart]{robins2009semiparametric}
J.~Robins, E.~T. Tchetgen, L.~Li, and A.~van~der Vaart.
\newblock Semiparametric minimax rates.
\newblock \emph{Electronic journal of statistics}, 3:\penalty0 1305, 2009.

\bibitem[Robins et~al.(1994)Robins, Rotnitzky, and Zhao]{robins1994estimation}
J.~M. Robins, A.~Rotnitzky, and L.~P. Zhao.
\newblock Estimation of regression coefficients when some regressors are not
  always observed.
\newblock \emph{Journal of the American Statistical Association}, 89\penalty0
  (427):\penalty0 846--866, 1994.

\bibitem[Robinson(1988)]{robinson1988root}
P.~M. Robinson.
\newblock Root-n-consistent semiparametric regression.
\newblock \emph{Econometrica: Journal of the Econometric Society}, pages
  931--954, 1988.

\bibitem[Rosenbaum and Rubin(1983)]{rosenbaum1983central}
P.~R. Rosenbaum and D.~B. Rubin.
\newblock The central role of the propensity score in observational studies for
  causal effects.
\newblock \emph{Biometrika}, 70\penalty0 (1):\penalty0 41--55, 1983.

\bibitem[Rosenbaum and Rubin(1984)]{rosenbaum1984reducing}
P.~R. Rosenbaum and D.~B. Rubin.
\newblock Reducing bias in observational studies using subclassification on the
  propensity score.
\newblock \emph{Journal of the American statistical Association}, 79\penalty0
  (387):\penalty0 516--524, 1984.

\bibitem[Schick(1986)]{schick1986asymptotically}
A.~Schick.
\newblock On asymptotically efficient estimation in semiparametric models.
\newblock \emph{The Annals of Statistics}, pages 1139--1151, 1986.

\bibitem[Sch{\"o}lkopf et~al.(2001)Sch{\"o}lkopf, Herbrich, and
  Smola]{scholkopf2001generalized}
B.~Sch{\"o}lkopf, R.~Herbrich, and A.~J. Smola.
\newblock A generalized representer theorem.
\newblock In \emph{International conference on computational learning theory},
  pages 416--426. Springer, 2001.

\bibitem[Stone(1977)]{stone1977consistent}
C.~J. Stone.
\newblock Consistent nonparametric regression.
\newblock \emph{The Annals of Statistics}, pages 595--620, 1977.

\bibitem[Tibshirani(1996)]{tibshirani1996regression}
R.~Tibshirani.
\newblock Regression shrinkage and selection via the lasso.
\newblock \emph{Journal of the Royal Statistical Society: Series B (Statistical
  Methodology)}, pages 267--288, 1996.

\bibitem[Tikhomirov(1993)]{tikhomirov1993varepsilon}
V.~Tikhomirov.
\newblock $\varepsilon$-entropy and $\varepsilon$-capacity of sets in
  functional spaces.
\newblock In \emph{Selected works of AN Kolmogorov}, pages 86--170. Springer,
  1993.

\bibitem[Tsiatis(2007)]{tsiatis2007semiparametric}
A.~Tsiatis.
\newblock \emph{Semiparametric Theory and Missing Data}.
\newblock Springer Science \& Business Media, 2007.

\bibitem[Van Der~Laan and Dudoit(2003)]{van2003unified}
M.~J. Van Der~Laan and S.~Dudoit.
\newblock Unified cross-validation methodology for selection among estimators
  and a general cross-validated adaptive epsilon-net estimator: Finite sample
  oracle inequalities and examples.
\newblock 2003.

\bibitem[van~der Laan and Rubin(2006)]{van2006targeted}
M.~J. van~der Laan and D.~Rubin.
\newblock Targeted maximum likelihood learning.
\newblock \emph{The International Journal of Biostatistics}, 2\penalty0
  (1):\penalty0 1--40, 2006.

\bibitem[van~der Laan et~al.(2019)van~der Laan, Benkeser, and
  Cai]{van2019efficient}
M.~J. van~der Laan, D.~Benkeser, and W.~Cai.
\newblock Efficient estimation of pathwise differentiable target parameters
  with the undersmoothed highly adaptive lasso.
\newblock \emph{arXiv preprint arXiv:1908.05607}, 2019.

\bibitem[van~der Vaart(1991)]{van1991differentiable}
A.~van~der Vaart.
\newblock On differentiable functionals.
\newblock \emph{The Annals of Statistics}, pages 178--204, 1991.

\bibitem[van~der Vaart(1994)]{van1994bracketing}
A.~van~der Vaart.
\newblock Bracketing smooth functions.
\newblock \emph{Stochastic Processes and their Applications}, 52\penalty0
  (1):\penalty0 93--105, 1994.

\bibitem[van~der Vaart(2002)]{van2000semiparametric}
A.~van~der Vaart.
\newblock Semiparametric statistics.
\newblock In \emph{Lectures on Probability Theory (St. Flour, 1999).} Springer,
  2002.

\bibitem[van~der Vaart and
  Wellner(1996)]{vandervaart-wellner1996:weak-convergence}
A.~W. van~der Vaart and J.~A. Wellner.
\newblock \emph{Weak Convergence and Empirical Processes}.
\newblock Springer, 1996.

\bibitem[Vershynin(2018)]{vershynin2018high}
R.~Vershynin.
\newblock \emph{High-dimensional probability: An introduction with applications
  in data science}, volume~47.
\newblock Cambridge university press, 2018.

\bibitem[Wang and Zubizarreta(2017)]{wang2017approximate}
Y.~Wang and J.~R. Zubizarreta.
\newblock Approximate balancing weights: Characterizations from a shrinkage
  estimation perspective.
\newblock \emph{arXiv preprint arXiv:1705.00998}, 2017.

\bibitem[Wong and Chan(2017)]{wong2017kernel}
R.~K. Wong and K.~C.~G. Chan.
\newblock Kernel-based covariate functional balancing for observational
  studies.
\newblock \emph{Biometrika}, 105\penalty0 (1):\penalty0 199--213, 2017.

\bibitem[Zeileis(2004)]{sandwich}
A.~Zeileis.
\newblock Econometric computing with hc and hac covariance matrix estimators.
\newblock \emph{Journal of Statistical Software}, 11\penalty0 (10):\penalty0
  1--17, 2004.
\newblock URL \url{http://www.jstatsoft.org/v11/i10/}.

\bibitem[Zhang and Zhang(2014)]{zhang2014confidence}
C.-H. Zhang and S.~S. Zhang.
\newblock Confidence intervals for low dimensional parameters in high
  dimensional linear models.
\newblock \emph{Journal of the Royal Statistical Society: Series B (Statistical
  Methodology)}, 76\penalty0 (1):\penalty0 217--242, 2014.

\bibitem[Zhao(2019)]{zhao2016covariate}
Q.~Zhao.
\newblock Covariate balancing propensity score by tailored loss functions.
\newblock \emph{The Annals of Statistics}, 47\penalty0 (2):\penalty0 965--993,
  2019.

\bibitem[Zhao et~al.(2017)Zhao, Small, and Ertefaie]{zhao2017selective}
Q.~Zhao, D.~S. Small, and A.~Ertefaie.
\newblock Selective inference for effect modification via the lasso.
\newblock \emph{arXiv preprint arXiv:1705.08020}, 2017.

\bibitem[Zheng and van~der Laan(2011)]{zheng2011cross}
W.~Zheng and M.~J. van~der Laan.
\newblock Cross-validated targeted minimum-loss-based estimation.
\newblock In \emph{Targeted Learning}, pages 459--474. Springer, 2011.

\bibitem[Zubizarreta(2015)]{zubizarreta2015stable}
J.~R. Zubizarreta.
\newblock Stable weights that balance covariates for estimation with incomplete
  outcome data.
\newblock \emph{Journal of the American Statistical Association}, 110\penalty0
  (511):\penalty0 910--922, 2015.

\end{thebibliography}

\ifaos
\else

\newpage

\begin{appendix}
\section{Proof of Finite Sample Results}
\label{sec:finite-sample-proofs}
In this section, we prove the finite sample bounds on which Theorem~\ref{theo:simple-rate} is based.
Here and throughout the appendix we will write $\Pn f$ and $\P f$ for averages of the function $f$ over the empirical and population distributions of $Z$ respectively in accordance with convention in the empirical process literature \citep[see e.g.][]{vandervaart-wellner1996:weak-convergence},
As a slight abuse of notation, we also write $\Pn$ to indicate a sample average in other contexts. We will write $\gapprox$
with the same meaning as $\tilde\riesz$ in Theorem~\ref{theo:simple-rate}, as it will be helpful to distinguish between vectors of weights $\gamma$
and functions $g$ which, when evaluated, give those weights.

\subsection{Setting}
\label{sec:appendix-setting}
We observe iid $(Y_1,Z_1) \ldots (Y_n,Z_n)$ with $Y_i \in \R$ and $Z_i$ in an arbitrary set $\zz$
and define $m(z) = \E[Y_i \mid Z_i=z]$ and $v(z) = \Var{Y_i \mid Z_i=z}$. We assume that $m$ is in a closed subspace $\mm$ 
of the $\P$-square integrable functions. Our estimand is defined as $\psi(m) = \P h(Z,m)$ in terms of a family $\set{h(z,\cdot) : z \in \zz}$ of linear functionals on $\mm$, and we assume that $\psi(\cdot) = \P h(Z,\cdot)$ is continuous on $\mm$.

\subsection{Consistency of the Minimax Linear Weights}
\label{sec:consistency}

In this section, we will prove the following consistency result.
It is stated as a deterministic consequence of two empirical process bounds
that will be shown to hold with high probability in Section~\ref{sec:putting-it-all-together}.

\begin{lemm}
\label{lemma:abstract-finite-sample}
Let $\F \subset \mm$ be absolutely convex with the property that the linear functionals $f \to f(z)$ and $f \to h(z,f)$ for $z \in \set{Z_1 \ldots Z_n}$ are continuous with respect to its gauge $\norm{\cdot}_{\F}$.\footnotemark\  Let $\riesz[\psi]$ be the Riesz representer for $\psi$ on the span of $\F$,
and consider, for $Q \in \set{\P,\Pn}$,
\begin{align*}
\hgamma &= \argmin_{\gamma \in \R^n}\sup_{f \in \F}  \sqb{\Pn h(Z_i, f) - \gamma_i f(Z_i)}^2 + (\sigma^2/n^2)\norm{\gamma}^2, \\
\tilde g  &= \argmin_{g} \ \norm{g - \riesz[\psi]}_{L_2(Q)}^2 + (\sigma^2/n)\norm{g}_{\F}^2. 
\end{align*}
These minimizers exist, are unique, and satisfy
\[ \Pn (\hgamma_i - \tilde g)^2 \le 2\alpha \eta_M r^2 \text{ for } \alpha = \max\set{3(\norm{\gapprox}_{\ff} + \eta_M r^2 n/\sigma^2),\ 2\eta_M/\eta_Q} \]
if $\F$ is $\norm{\cdot}_{L_2(Q)}$-closed and bounded and for all $f \in \F$, 
\begin{equation}
\label{eq:ratio-process-bounds}
\begin{aligned}
& \Pn f^2 \ge \eta_Q \P f^2 
&&\text{ if }\  \P f^2 \ge r^2 &&, \\
&\abs{(\Pn-\P)[ h(\cdot,f) - \gapprox f]} \le \eta_M r^2\
&&\text{ if  }\ \P f^2 \le r^2 && \text{ for }\ Q=\P \\
 &\abs{(\Pn-\P)[ h(\cdot,f) - \riesz[\psi] f]} \le \eta_M r^2\
&&\text{ if  }\ \P f^2 \le r^2 && \text{ for }\ Q=\Pn \\
\end{aligned}
\end{equation}
\end{lemm}
\footnotetext{Gauge-continuity is a convenient rephrasing of the pointwise boundedness assumption
of Theorems~\ref{theo:simple}-\ref{theo:simple-rate}.}
We begin by showing existence and uniqueness. It suffices to show that the functions minimized are lower-semicontinuous,
as they are proper and strongly convex and minimized over reflexive spaces $(\R_n, \norm{\cdot}_2)$ 
and $(\vspan \F, \norm{\cdot}_{L_2(Q)})$ respectively \citep[Corollary 2.20]{peypouquet2015convex}. The first is continuous, as 
a convex function is continuous if it is bounded on an open set \citep[Theorem 5.43]{aliprantis2006infinite},
and it is bounded on any bounded subset of $\R^n$. And the second is lower-semicontinuous,
as it is the sum of the continuous function mapping \smash{$g \to \norm{g - \riesz[\psi]}_{L_2(Q)}^2$} and 
the square of the gauge of the absorbing closed convex set $(\sqrt{n}/\sigma)\F$, which is lower-semicontinuous \citep[Theorem 5.52]{aliprantis2006infinite}.

To show that our weights converge to $\gapprox$, we will characterize them as the solution to a least squares problem for estimating $\gapprox$.
This least squares problem is the dual of the problem \eqref{eq:aml} solved by our weights $\hgamma$.
We use the following lemma to establish duality.
\begin{lemm}
\label{lemma:duality}
Let $\S$ be a normed vector space with norm $\norm{\cdot}$   
and $L_0:\S \to \R$ and $\bar L:\S \to \R^n$ be continuous linear maps.
Define a primal $p:\R^n \to R$ and dual $d:\S \to \R$ by
\begin{align*}
p(\gamma) &=  \frac{1}{2}\sup_{\norm{f} \le 1}\sqb{ L_0(f) - \gamma^T \bar L(f)}^2 +\frac{1}{2} \norm{\gamma}_2^2, \\
d(g)      &= \frac{1}{2} \norm{\bar L(g)}_2^2  - L_0(g)  + \frac{1}{2}\norm{g}^2
\end{align*}
Then:
\begin{enumerate}
\item $\min_{\gamma \in \R^n} p(\gamma) = -\inf_{g} d(g)$. 
\item $p$ has a unique minimum at a vector $\hgamma \in \R^n$.
\item For every sequence $\hat g^j$ along which $d(\hat g^j) \to \inf_{g} d(g)$, \\
      $\bar L(\hat g^j) \to \hgamma$.
\end{enumerate}
\end{lemm}
\noindent In our estimator \eqref{eq:aml}, we use the weights $\hriesz$ that minimize 
\smash{$(2\sigma^2/n^2) p(\gamma)$} where \smash{$L_0(f) = \sum_{i=1}^n h(Z_i, f)$}, \smash{$\gamma^T \bar L(f) = \sum_{i=1}^n \gamma_i f(Z_i)$},
and $\norm{\cdot}$ is $\sigma$ times the gauge of $\F$, and we can characterize our weights as the limit of a minimizing sequence for the corresponding dual $d(g)$.
\begin{equation}
\label{eq:duality-explicit}
\begin{aligned}  
&\hgamma_i = \lim_{j \to \infty} \hg_j(Z_i) \quad \text{ if } \quad d(\hat g_j) \to \inf_{g}d(g) \quad \text{ for } \\
&(2/n) d(g) = \Pn g^2  - 2\Pn h(\cdot,g)  + (\sigma^2/n)\norm{g}_{\F}^2.
\end{aligned}
\end{equation}
We will show that $g_j \approx \tilde g$ whenever $d(g_j) \le d(\tilde g)$.
This characterizes $\hgamma$, as each of its coordinates $\hgamma_i$ 
is the limit of $g_j(Z_i)$ for a sequence of functions with this property.

To do this, we will show that the excess loss $d(g) - d(\tilde g)$ is large
unless $g \approx \tilde g$. We begin by lower bounding the excess loss.
Via the elementary identity $g^2 - \tilde g^2 = (g-\tilde g)^2 + 2\gapprox (g-\gapprox)$,
\begin{equation}
\label{eq:concrete-dual}
\begin{aligned}
(2/n)[d(g) - d(\tilde g)] 
&= \Pn (g-\tilde g)^2 + 2\Pn \tilde g (g - \tilde g) - 2 \Pn h(\cdot, g - \tilde g)\\
& + (\sigma^2/n)[\norm{g}_{\F}^2 - \norm{\tilde g}_{\F}^2].
\end{aligned} 
\end{equation}
To lower bound this, we use a convenient property of our approximation $\tilde g$
\begin{equation}
\label{eq:projection-theorem}
0 \le  Q (\tilde g - \riesz[\psi]) (g - \tilde g) + (\sigma^2/n)\norm{\tilde g}_{\F}(\norm{g}_{\F} - \norm{\tilde g}_{\F})
 \ \text{ for all }\  g \in L_2(Q). 
\end{equation}
This is implied by the following generalization of the Hilbert space projection theorem.
The relevant Hilbert space is $\vspan\F \subseteq L_2(Q)$ 
and we take $\phi(x)=\sigma^2 x^2 /(2n)$ and $\rho(g) = \norm{g}_{\F}$.
\begin{lemm}
\label{lemma:projection-theorem}
Let $\phi$ be a nondecreasing convex differentiable function on the nonnegative reals;
$\rho$ be a proper, nonnegative, convex, and lower-semicontinuous function on a Hilbert space;
and $g_{\star}$ be a vector in that space.
Letting $\tilde g = \argmin_{g}  (1/2)\norm{g - g_{\star}}^2 + \phi(\rho(g))$,
$\inner{  \tilde g - g_{\star}, g - \tilde g} + \phi'(\rho(\tilde g))(\rho(g) - \rho(\tilde g))$
is nonnegative for all $g$.
\end{lemm}
Subtracting from the excess loss twice the non-negative right side of \eqref{eq:projection-theorem} yields a simple lower bound.
It is the sum of the empirical mean squared error, a mean-zero empirical process, and a regularization term:
\begin{equation*}
\begin{aligned}
&\Pn (g-\tilde g)^2 + 2(\Pn-Q) \tilde g (g - \tilde g) - 2[\Pn h(\cdot, g - \tilde g) -  Q \riesz[\psi](g-\tilde g)] && \\
& + (\sigma^2/n)[\norm{g}_{\F}^2 - \norm{\tilde g}_{\F}^2 - 2\norm{\tilde g}_{\F}(\norm{g}_{\F}-\norm{\tilde g}_{\F}) ] && \\
&= \Pn (g-\tilde g)^2 + 2(\Pn-\P) [\tilde g (g - \tilde g) \  - \ h(\cdot, g - \tilde g)] + (\sigma^2/n)(\norm{g}_{\F} - \norm{\tilde g}_{\F})^2 && \text{ for }\ Q=\P, \\
&= \Pn (g-\tilde g)^2 + 2(\Pn-\P) [\riesz[\psi](g-\tilde g) - h(\cdot, g - \tilde g)] + (\sigma^2/n)(\norm{g}_{\F} - \norm{\tilde g}_{\F})^2 && \text{ for }\ Q=\Pn. \\
\end{aligned} 
\end{equation*}
Here we've used the Riesz representation property $\P h(\cdot, f) = \P \riesz[\psi] f$ to simplify the first expression
in the two cases $Q=\P$ and $Q = \Pn$. We will use another lower bound that is a function of $\delta = g-\tilde g$,
\begin{equation}
\label{eq:excess-riesz}
\begin{aligned}
&\excess(\delta) 
= \Pn \delta^2 - 2 \abs{M(\delta)} + (\sigma^2/n) (\norm{\delta}_{\ff} - 2\norm{\gapprox}_{\ff})_+^2 \quad \text{ where }\\
&M(\delta) = \begin{cases} 
(\Pn-\P) [h(\cdot,\delta) - \tilde g \delta] & \text{ for } Q = \P, \\
(\Pn-\P) [h(\cdot,\delta) - \riesz[\psi] \delta] & \text{ for } Q = \Pn.
\end{cases} \quad \text{ and }\quad x_+^2 := x^2 1(x \ge 0).
\end{aligned}
\end{equation} 
This bound is derived from the previous one by (i) replacing the second term with its negated absolute value
and (ii) substituting a lower bound on the third term implied by the triangle inequality 
$\norm{\delta}_{\ff} -  \norm{\tilde g}_{\ff} \le \norm{g}_{\ff}$, the increasingness of $x_+^2$,
and the bound $x_+^2 \le x^2$.

By the lemma below, this excess loss lower bound $\excess(\delta)$ can be zero or negative only if 
$\Pn \delta^2 \le 2\alpha \eta_M r^2$. And because $\hgamma_i - \gapprox(Z_i)$ is the limit of a sequence $\delta_j(Z_i)$ 
with $\excess(\delta_j) \le 0$, it follows that $\Pn (\hgamma_i - \gapprox(Z_i))^2 \le 2\alpha\eta_M r^2.$

\begin{lemm}
\label{lemma:consistency-deterministic}
Let $\F$ be a class of functions that is star-shaped around zero,
define $\excess$ as in \eqref{eq:excess-riesz},
and suppose that for all $\delta \in \ff$,
\begin{equation}
\label{eq:ratio-bounds}
\begin{aligned}
& \Pn  \delta^2 \ge \eta_Q \P \delta^2
&& \text{ if }\ \P \delta^2 \ge r^2 \\
&\abs{M(\delta)} \le \eta_M r^2
&& \text{ if }\ \P \delta^2 \le r^2. 
\end{aligned}
\end{equation}
Let $\alpha = \max\set{3(\norm{\gapprox}_{\ff} + \eta_M r^2 n/\sigma^2),\ 2\eta_M/\eta_Q}$. 
Then $\excess(\delta) \le 0$ only if $\norm{\delta}_{\ff} \le \alpha$, 
$\P \delta^2 \le (\alpha r)^2$, and $\Pn \delta^2 \le 2\alpha \eta_M r^2$.
Furthermore, $\excess(\delta) \le \xi$ only if $\norm{\delta}_{\ff} \le \alpha + (\xi n)^{1/2}/\sigma$.
\end{lemm}
\noindent We conclude our proof of Lemma~\ref{lemma:abstract-finite-sample}
by proving Lemmas~\ref{lemma:duality}-\ref{lemma:consistency-deterministic}.

\begin{proof}[Proof of Lemma~\ref{lemma:duality}]
Because $p$ is a proper, strictly convex, coercive, and lower-semicontinuous function on the reflexive space $\R^n$,
it has a unique minimum $\hat p$ at some vector $\hgamma \in \R^n$ \citep[Corollary 2.20]{peypouquet2015convex}. 
Letting $A:\R^n \to \S^{\star}$ be the linear map $A \gamma := -\gamma^T\bar L$,
our primal $p:\R^n \to \R$ has the form of a primal in Fenchel-Rockafellar duality,
\begin{align*}
p(\gamma) &= s(\gamma) + r(A \gamma) \quad \text{ where } \\
s(\gamma) &=  (1/2)\norm{\gamma}^2 \\
r(L) &= (1/2)\norm{L_0 + L}_{\S^{\star}}^2,
\end{align*}
so its dual,
\[  d:\S^{\star\star} \to \R\ \text{ by }\ d(L^{\star}) := s^{\star}(-A^{\star}L^{\star}) + r^{\star}(L^{\star}), \]
has a minimum, and each argmin $\hat L^{\star}$ satisfies $-A^{\star}\hat L^{\star} \in \partial s(\hat \gamma) = \set{ \hat \gamma }$ \citep[Theorem 3.51]{peypouquet2015convex}. Here $s^{\star}$ and $r^{\star}$ are the convex conjugates of $s$ and $r$; $\partial s(\hat \gamma)$ is the subgradient of $s$ at $\hat \gamma$; and 
$A^{\star}$ is the adjoint of $A$, i.e., $A^{\star} L^{\star}$ is the vector in $\R^n$  
satisfying  $\inner{A^{\star} L^{\star}, e_i} = \inner{L^{\star}, A e_i}$ for the standard basis vectors $e_1 \ldots e_n$.

We will now characterize $\hat L^{\star}$ more explicitly, as a minimizer of 
\[ d(L^{\star}) = \frac{1}{2}\sum_{i=1}^n \inner{L^{\star}, A e_i}^2  - \inner{L^{\star}, L_0}  + \frac{1}{2} \norm{L^{\star}}_{\S^{\star\star}}^2. \]
To do this, we first calculate $r^{\star}$ and $s^{\star}$.
\begin{align*}
r^{\star}(L^{\star}) 
&= \sup_{L \in \S^{\star}} \inner{L^{\star}, L} - (1/2)\norm{L_0 + L}_{\S^{\star}}^2 \\
&= \sup_{L' \in \S^{\star}} \inner{L^{\star}, L' - L_0} - (1/2)\norm{L'}_{\S^{\star}}^2 \\ 
&= -\inner{L^{\star}, L_0} + \sup_{t \in \R } \sup_{\norm{L''}_{\S^{\star}}=1} \inner{L^{\star},  t L''} - (1/2)\norm{tL''}_{\S^{\star}}^2 \\
&= -\inner{L^{\star}, L_0} + \sup_{t \in \R} t \norm{L^{\star}}_{\S^{\star\star}} - t^2/2 \\
&= -\inner{L^{\star}, L_0} + (1/2)\norm{L^{\star}}_{\S^{\star\star}}^2.
\end{align*}
In the first step, we reparameterize in terms of $L' = L_0 + L$; in the second, we reparameterize again in terms of $tL''=L'$; 
in the third we substitute $\norm{\cdot}_{\S^{\star\star}}$ for its definition; and in the fourth we use the identity $\max_{t \in \R} at - t^2/2 = a^2/2$.
Similarly, for $y \in \R^n$, 
\begin{align*}
s^{\star}(y) 
&= \sup_{x \in \R^n} \inner{y,x} - \norm{x}^2/2 \\
&= \sup_{t \in \R}\sup_{x' \in \R^n : \norm{x'}=1} \inner{y,tx'} - t^2/2 \\
&= \sup_{t \in \R}t\norm{y} - t^2/2 \\
&= \norm{y}^2/2.
\end{align*}
Taking $y=-A^{\star}L^{\star}$ and establishes our claimed characterization of $d(L^{\star})$.

Now suppose that $L^{\star}$ is an evaluation functional $J_{g} \in \S^{\star\star}$, defined $J_{g}(L) := L(g)$.
Then for any $x \in \R^n$ and any $g \in \S$, $\inner{J_{g}, -Ax} =  \inner{ J_{g}, x^T \bar L } = x^T \bar L(g)$,
and it follows that $\sum_{i=1}^n \inner{J_{g}, A e_i}^2 = \norm{\bar L(g)}^2$. Thus,
\[ d(J_{g}) = \frac{1}{2}\norm{\bar L(g)}^2  - L_0(g)  + \frac{1}{2} \norm{g}_{\S}. \]
If an argmin $\hat L^{\star}$ of $d$ were the evaluation functional $J_{g}$, 
then $g$ would minimize the right side above. When every $L^{\star} \in \S^{\star\star}$ is an evaluation functional,
i.e. when $\S$ is reflexive, because $d$ has a minimum $\hat L^{\star}$ over $\S^{\star\star}$ 
it follows that the right side above has a minimum $\hat g$ over $g \in \S$. 
Furthermore, recalling our first-order optimality condition $-A^{\star}\hat L^{\star} =  \hat \gamma$, 
$\hat \gamma = \bar L(\hat g)$.

This is essentially true whether $\S$ is reflexive or not because evaluation functionals are dense in the bidual $\S^{\star\star}$ in an appropriate sense.
By Goldstine's theorem, for every $L^{\star} \in \S^{\star\star}$, there is a sequence $g_j \in \S$ satisfying $\norm{g_j}_{\S} \le \norm{L^{\star}}_{\S^{\star\star}}$
for all $j$ and $L(g_j) \to L^{\star}(L)$ pointwise for each $L \in \S^{\star}$ \citep[e.g.,][Theorem 2.6.26]{megginson2012introduction}.
Consider such a sequence $\hat g_j$ for an argmin $\hat L^{\star}$ of $d$. We can characterize $\hat \gamma$ as $\lim_{j \to \infty} \bar L(g_j)$,
as $\hat \gamma = -A^{\star}\hat L^{\star} = \lim_{j \to \infty} -A^{\star} J_{g_j}$: $A^{\star}\hat L^{\star}$ is the solution to finitely many linear equations 
\smash{$\inner{A^{\star}\hat L^{\star}, y_k} = \inner{\hat L^{\star}, Ay_k}$} $= \lim_{j \to \infty} \inner{J_{\hat g_j}, Ay_k}$ 
for $\set{ y_k }$ forming a basis for $\R^n$, and pointwise convergence is sufficient to imply convergence 
of the finite dimensional vector with elements $\inner{J_{\hat g_j}, Ay_k}$. 
Furthermore, because $d$ is continuous and depends only on $\norm{\cdot}_{\S^{\star\star}}$ and the value of finitely many functionals, in particular \smash{$L_0$} 
and a basis for the image of $A$, $d(J_{\hat g_j}) \to d(\hat L^{\star})$, and it follows that
$d(J_{\hat g_j}) \to \inf_{g \in \S}d(J_{g})$.

We conclude our proof by showing that every sequence $g_j$ along which $d(J_{g_j})$ converges to its infimum has the 
same limiting value of $\bar L(g_j)$, which therefore must converge to \smash{$\lim_{j \to \infty} \bar L(\hat g_j) = \hat \gamma$}.
This is the case because every term in $d(J_{g})$ is convex in $g$ and 
there is a term that is uniformly convex in $\bar L(g)$: 
if there were two minimizing sequences $g_j$ and $\tilde g_j$ with different limits 
\smash{$\lim \bar L(g_j) \neq \lim \bar L(\tilde g_j)$}, 
their average $(g_j + \tilde g_j) / 2$ would be a sequence along which $d$ converges to 
something strictly smaller than the average of the limit along $g_j$ or $\tilde g_j$, which is its infimum.
\end{proof}

\begin{proof}[Proof of Lemma~\ref{lemma:projection-theorem}]
Let $a(g) = (1/2) \norm{g - g_{\star}}^2 + \phi(\rho(g))$.  
Because it is proper, convex, coercive, and lower-semicontinuous, $a$ has a minimizer $\tilde g$ \citep[Theorem 2.19]{peypouquet2015convex}.
Zero is in its subdifferential $\partial a(\tilde g)$ at its minimizer,
and by a chain rule for subdifferentials \citep[Corollary 3.5]{combari1996note} and 
the Moreau-Rockafellar theorem for subdifferentials of sums \citep[Theorem 3.30]{peypouquet2015convex},
$\partial a(\tilde g)$ is the set of maps $v_a(f) = \inner{  \tilde g - g_{\star}, f} + \phi'(\rho(\tilde g))v_{\rho}(f)$
for $v_{\rho} \in \partial \rho(\tilde g)$. And by definition, $v_{\rho}(g-\tilde g) \le \rho(g) - \rho(\tilde g)$,
so all functionals $v_a \in \partial a(\tilde g)$ satisfy $v_a(g-\tilde g) \le \inner{\tilde g - g_{\star}, g-\tilde g} + \phi'(\rho(\tilde g))(\rho(g) - \rho(\tilde g))$.
This bound implies the claimed nonnegativity property, as $0 = v_a$ for some $a$.
\end{proof}

We prove Lemma~\ref{lemma:consistency-deterministic} with the aid of the following scaling result.
\begin{lemm}
\label{lemm:rescaling}
Let $\ff$ be a set that is star-shaped around zero, $L$ be a homogeneous functional on $\vspan \ff$, 
and $\norm{\cdot}$ be a norm on $\vspan \ff$. If
$L(f) \le \eta r^2$ for all $f \in \ff$ with $\norm{f} \le r$, then
$L(f) \le (\eta/\alpha)\max\set{\norm{f},\ \alpha r}^2$ for all $f \in \alpha \ff$
for every $\alpha > 0$.
\end{lemm}
\begin{proof}
For $f \in \alpha \ff$ with $\norm{f} \le \alpha r$, consider $f' = f/\alpha$. 
Because $f' \in \ff$ and $\norm{f'} \le r$, our assumed bound implies that 
$L(f)=\alpha L(f') \le \eta \alpha r^2=(\eta/\alpha)(\alpha r)^2$.
For $f \in \alpha \ff$ with $\norm{f} \ge \alpha r$, consider $f' = r f / \norm{f}$.
Because $f' \in \ff$ and $\norm{f'} \le r$, our assumed bound implies that 
$L(f) = L(f') \norm{f}/r \le \eta  r \norm{f} \le (\eta/\alpha) \norm{f}^2$,
using in the last step the property $\norm{f} \ge \alpha r$.
\end{proof}

\begin{proof}[Proof of Lemma~\ref{lemma:consistency-deterministic}]
Given our assumed bounds, if $\delta \in \alpha \F$, 
\begin{equation}
\label{eq:rewritten-ratio-bounds}
\begin{aligned}
& \Pn \delta^2 \ge \eta_Q \P \delta^2 && \quad \text{ when }\ \P \delta^2 \ge (\alpha r)^2 \\ 
&\abs{M(\delta)} \le (\eta_M/\alpha) \P \delta^2  && \quad \text{ when }\ \P \delta^2 \ge (\alpha r)^2  \\
&\abs{M(\delta)} \le \eta_M \alpha r^2 && \quad \text{ when }\ \P \delta^2 \le (\alpha r)^2  \\
\end{aligned}
\end{equation}
The first of these is an immediate consequence of the invariance of the ratio $\Pn f^2 / \P f^2$ to scaling
and the second and third follow from Lemma~\ref{lemm:rescaling} with $L(\cdot)=\abs{M(\cdot)}$.  We will now prove our claims using these bounds.

We begin by showing that $\excess(\delta) > 0$ for all $\delta$ with $\norm{\delta}_{\ff} \ge \alpha$.
It suffices to consider $\delta$ with $\norm{\delta}_{\ff} = \alpha$,
as we can write the others as $\delta=s\delta'$ for $s > 1$ and $\norm{\delta'}_{\ff} = \alpha$,
and $\excess(s\delta') \ge s \excess(\delta')$ for $s \ge 1$ when $\alpha \ge 2\norm{\gapprox}_{\ff}$: 
for such $s$, $\delta'$, and $\alpha$,
\begin{align*}
\excess(s\delta') - s \excess(\delta') &= (s^2-s)\Pn (\delta')^2 + (\sigma^2/n)[(s^2-s)\alpha^2 + (1-s)4\norm{\gapprox}^2] \\  
				       &= (s^2-s)\Pn (\delta')^2 + (\sigma^2/n)(s-1)(s\alpha^2 - 4\norm{\gapprox}^2) \ge 0.
\end{align*}
If $\P \delta^2 \ge (\alpha r)^2$, then 
$\excess(\delta) \ge (\eta_Q - 2\eta_M/\alpha) \P \delta^2 + (\sigma^2/n)(\alpha - 2\norm{\gapprox}_{\ff})_+^2$.
If instead $\P \delta^2 \le (\alpha r)^2$, then 
$\excess(\delta) \ge -2\eta_M \alpha r^2 + (\sigma^2/n)(\alpha - 2\norm{\gapprox}_{\ff})_+^2$.
Thus, $\excess(\delta) > \xi$ for all $\delta$ with $\norm{\delta}_{\ff} \ge \alpha$ 
so long as $\eta_Q - 2\eta_M/\alpha \ge 0$ and $(\sigma^2/n)(\alpha - 2\norm{\gapprox}_{\ff})_+^2 > 2\eta_M \alpha r^2 + \xi$.
These conditions hold for $\alpha > \alpha_0 + (\xi n/\sigma^2)^{1/2}$  where
$\alpha_0 = \max\set{2\eta_M/\eta_Q,\ 3(\norm{\gapprox}_{\ff} + \eta_M r^2 n/\sigma^2)}$.
To see that this lower bound implies the latter condition,
observe that for $\alpha \ge 2\norm{\gapprox}_{\ff}$, it expands to
\[ 0 < \alpha^2 - \alpha \p{4\norm{\gapprox}_{\ff} + 2\eta_M r^2 \lambda} + 4\norm{\gapprox}_{\ff}^2 - \xi \lambda \ \text{ for }\ \lambda = n/\sigma^2,\] 
which holds for $\alpha$ exceeding the larger root of the right side,
\begin{align*} 
&2\norm{\gapprox}_{\ff} + \eta_M  r^2 \lambda 
+ \sqrt{ \p{2\norm{\gapprox}_{\ff} + \eta_M  r^2 \lambda}^2 -  4\norm{\gapprox}_{\ff}^2 + \xi \lambda } \\
&= 2\norm{\gapprox}_{\ff} + \eta_M  r^2 \lambda 
+ \sqrt{ 4\norm{\gapprox}_{\ff}\eta_M  r^2 \lambda + \p{\eta_M  r^2 \lambda}^2 + \xi \lambda } \\
&\le 2\p{\norm{\gapprox}_{\ff} + \eta_M  r^2 \lambda}
+ 2 \sqrt{\norm{\gapprox}_{\ff}\eta_M  r^2 \lambda } + \sqrt{\xi \lambda } \\
& \le 3\p{\norm{\gapprox}_{\ff} + \eta_M r^2 \lambda} + \sqrt{\xi\lambda}
\end{align*}
Here the second expression is derived by
expanding and canceling terms under the square root in the first,
the third is follows via the inequality $\sqrt{a+b+c} \le \sqrt{a}+\sqrt{b} + \sqrt{c}$,
and the fourth follows via the inequality $a+b \ge 2\sqrt{ab}$ relating the arithmetic and geometric means.

Now take $\xi=0$ and consider $\delta \in \alpha\F$ for $\alpha > \alpha_0$.
If $\P \delta^2 \ge (\alpha r)^2$, then 
$\excess(\delta) \ge (\eta_Q - 2 \eta_M/\alpha) \P \delta^2 > 0$.
Otherwise, $\excess(\delta) \ge \Pn \delta^2 - 2 \eta_M \alpha r^2$,
which is positive if $\Pn \delta^2 > 2\eta_M \alpha r^2$.

In summary, we've shown that for $\alpha > \alpha_0$, (i) $\excess(\delta) > 0$ if $\norm{g}_{\ff} \ge \alpha$, $\P \delta^2 \ge (\alpha r)^2$,
or $\Pn \delta^2 > 2\eta_M \alpha r^2$,  and (ii) $\excess(\delta) > \xi$ if $\norm{\delta}_{\ff} \ge \alpha + (\xi n)^{1/2}/\sigma$. 
Taking contrapositives, (i)  $\excess(\delta) \le 0$ only if 
$\norm{\delta}_{\ff} < \alpha$, $\P \delta^2 < (\alpha r)^2$, and $\Pn \delta^2 \le 2\eta_M \alpha r^2$,
and (ii) $\excess(\delta) \le \xi$ only if $\norm{\delta}_{\ff} < \alpha + (\xi n)^{1/2}/\sigma$. 
It follows that for $\alpha=\alpha_0$, nonstrict variants of these bounds hold.
\end{proof}

\subsection{Convergence of the noise term} 
\label{sec:convergence-of-the-noise-term}

In this section, we bound the difference between the noise term in the decomposition \eqref{eq:error-decomp}
and the iid sum $\Pn \gapprox(Z_i) \varepsilon_i,\ \varepsilon_i = Y_i - m(Z_i)$. Because $\hriesz$ is a function of $Z_1 \ldots Z_n$,
we can apply Chebyshev's inequality conditionally on $Z_1 \ldots Z_n$ to the difference between our noise term and this sum. 
With conditional and therefore unconditional probability $1-\delta$,
\begin{equation}
\label{eq:noise-term-deviation}
\begin{aligned}
\abs*{ \Pn (\hriesz_i - \gapprox(Z_i)) \varepsilon_i } 
&\le \delta^{-1/2} n^{-1/2} \sqrt{\Pn [\hriesz_i - \gapprox(Z_i)]^2 v(Z_i)}.  \\
&\le \delta^{-1/2} n^{-1/2} \norm{v}_{\infty} \norm{\hriesz - \gapprox}_{L_2(\Pn)}.
\end{aligned}
\end{equation}
The second bound follows from the first via H\"older's inequality.

\subsection{Bounding the bias term}
\label{sec:bias-term-bound}

In this section, we bound the bias term in the decomposition \eqref{eq:error-decomp}. 
As we work with two function classes $\F$ and $\F'$, to avoid ambiguity we indicate the class with a sub or superscript:
\begin{equation}
\label{eq:weight-def-explicit}
\begin{aligned}
 \hgamma^{\GG} &= \argmin_{\gamma \in \R^n} I_{h,\GG}^2(\gamma) + \frac{\sigma_{\GG}^2}{n}\norm{\gamma}_{L_2(\Pn)}^2, \\
 I_{h,\GG}(\gamma) &= \sup_{f \in \GG} \Pn h(Z_i, f) - \gamma_i f(Z_i).
\end{aligned}
\end{equation}
For any absolutely convex set $\GG$, our bias term satisfies the bound
\[ \abs{\Pn h(Z_i, \hm - m) - \hgamma^{\F}_i (\hm - m)} \le \norm{\hm - m}_{\GG} I_{h,\GG}(\hgamma^{\F}). \]
Rather than using this bound for $\GG=\F$, we use it for \smash{$\F' = \{ f : \norm{f}_{\F}^2 + \rho^{-2}\norm{f}_{L_2(\Pn)}^2 \le 1\}$},
a subset of $\F$ containing only small functions. We control the latter factor as follows.
\begin{equation}
\label{eq:bias-bound-abstract}
\begin{aligned}
&I_{h,\F'}(\hgamma^{\F}) \le \rho\norm{\hgamma^{\F'} - \hgamma^{\F}}_{L_2(\Pn)} + I_{h,\F'}(\hgamma^{\F'}) \\
&\le \rho\norm{\hgamma^{\F'} - \hgamma^{\F}}_{L_2(\Pn)} +  \sqb{I_{h,\F'}^2(\riesz[\psi]) +
\frac{\sigma_{\F'}^2}{n}\p{\norm{\riesz[\psi]}_{L_2(\Pn)}^2 - \norm{\hgamma^{\F'}}_{L_2(\Pn)}^2}}^{1/2} \\
&\le \rho\norm{\hgamma^{\F'} - \hgamma^{\F}}_{L_2(\Pn)} + I_{h,\F'}(\riesz[\psi]) \\
&+ \sqb{\frac{\sigma_{\F'}^2}{n}\cb{\norm{\riesz[\psi]}_{L_2(\Pn)}^2 \wedge 2\norm{\riesz[\psi]}_{L_2(\Pn)}\p{\norm{\hgamma^{\F}-\riesz[\psi]}_{L_2(\Pn)} + \norm{\hgamma^{\F'} - \hgamma^{\F}}_{L_2(\Pn)}}}}^{1/2}
\end{aligned}
\end{equation}
The first bound, via the Cauchy-Schwarz inequality, is implied by the property that $\norm{f}_{L_2(\Pn)} \le \rho$ for all $f \in \F'$.
\begin{align*} 
I_{h,\F'}(\hgamma^{\F}) 
&=   \sup_{f \in \F'} \Pn[ h(Z_i, f) - \hgamma^{\F'} f + (\hgamma^{\F'}-\hgamma^{\F})f] \\
&\le \sup_{f \in \F'} \Pn[ h(Z_i, f) - \hgamma^{\F'} f] + \sup_{f \in \F'} \Pn (\hgamma^{\F'}-\hgamma^{\F})f \\
&\le I_{h,\F'}(\hgamma^{\F'}) + \rho \norm{\hgamma^{\F'}-\hgamma^{\F}}_{L_2(\Pn)}. 
\end{align*}
The second bound is implied by the optimality of the weights $\hgamma^{\F'}$. It is a rearrangement of the 
condition that the function minimized by $\hgamma^{\F'}$ is smaller at its minimizer than at the weights $\gamma_i = \riesz[\psi](Z_i)$.
The third bound follows from the second by some elementary arithmetic. As $a^2-b^2=2a(a-b) - (a-b)^2 \le 2a\abs{a-b}$,
using this bound termwise and then taking Cauchy-Schwarz and triangle inequality bounds,
\begin{align*}
\norm{\riesz[\psi]}_{L_2(\Pn)}^2 - \norm{\hgamma^{\F'}}_{L_2(\Pn)}^2 
&\le 2\Pn \riesz[\psi](Z_i) \abs{\riesz[\psi](Z_i) - \hgamma^{\F'}_i} \\
&\le 2\norm{\riesz[\psi]}_{L_2(\Pn)}\norm{\riesz[\psi]-\hgamma^{\F'}}_{L_2(\Pn)} \\
&\le 2\norm{\riesz[\psi]}_{L_2(\Pn)}\p{\norm{\riesz[\psi] - \hgamma^{\F}}_{L_2(\Pn)}  + \norm{\hgamma^{\F} - \hgamma^{\F'}}_{L_2(\Pn)}}.
\end{align*}

Having established the abstract bound \eqref{eq:bias-bound-abstract}, we will derive a concrete version
by controlling $\norm{\hgamma^{\F'}-\hgamma^{\F}}_{L_2(\Pn)}$.
We will take $\sigma_{\F'}^2 = \sigma_{\F}^2 / (1 - \eta)$ with $\eta = \sigma_{\F}^2/(\rho^2 n)$,
as this allows us to control it well. To do this, we recall that the weights $\hgamma_{\GG}$ satisfy
$\hgamma_i^{\GG}=\lim_{j \to \infty}g_j(Z_i)$ where $g_j$ is a minimizing sequence for the dual $d_{\GG}$ \eqref{eq:concrete-dual},
use the similarity of $d_{\F'}$ and $(\sigma_{\F'}/\sigma_{\F})^2 d_{\F}$ to show that a minimizing sequence for the latter
is almost a minimizing sequence for the former, and use the strong convexity of $d_{\F'}$ 
to show that this implies $\hgamma^{\F'} \approx \hgamma^{\F}$. 

\begin{lemm}
\label{lemma:weight-similarity}
For an absolutely convex set $\F$ and $\sigma_{\F} \ge 0$, 
let 
\[ \F' = \set{ f : \norm{f}_{\F}^2 + \rho^{-2}\norm{f}_{L_2(\Pn)}^2 \le 1} \ \text{ and }\ 
\sigma_{\F'}^2 = \sigma_{\F}^2 / (1-\eta) \ \text{ for } \eta = \sigma_{\F}^2 / (\rho^2 n). \]
Define $\hgamma^{\F}$ and $\hgamma^{\F'}$ as in \eqref{eq:weight-def-explicit}
and corresponding duals $d_{\F}$ amd $d_{\F'}$ as in \eqref{eq:concrete-dual}
and suppose that for some $\gapprox$, 
$\excess_{\F}(\delta) = (2/n)[d_{\F}(\gapprox + \delta) - d_{\F}(\gapprox)]$
has the property that for every $\xi \ge 0$, $\excess_{\F}(\delta) \le \xi$ 
only if $\norm{\delta}_{\F} \le \alpha + (\xi n)^{1/2}/\sigma_{\F}$. For any $\gamma \in \R^n$, 

\begin{align*}
&\rho\norm{\hgamma_{\F} - \hgamma_{\F'}}_{L_2(\Pn)} 
\le 6\eta\sqb{I_{h,\F'}(\gamma) + \rho\norm{\gamma}_{L_2(\Pn)}} \\
&\vee   \sqb{\p{4\rho^2\alpha  + 2I_{h,\F'}(\gamma) 
      + 4\rho^2c_{\star} \norm{\hgamma_{\F} - \gapprox}_{L_2(\Pn)}^{1/2}} 6\eta I_{h,\F'}(\gamma) }^{1/2} \\
&\vee   \sqb{24 \eta \rho^{3/2} c_{\star} I_{h,\F'}(\gamma)}^{2/3} \ \ \text{ with }\ \ c_{\star}^2 = 2\rho^{-3}I_{h,\F'}(\gamma) + 2\rho^{-2} \norm{\gamma}_{L_2(\Pn)}. 
\end{align*}
\end{lemm}

To establish our claim of oracle behavior, in the sense that we get essentially the same bias bound with the weights $\hgamma^{\F}$ 
as we would with $\hgamma^{\F'}$, we need to show that \smash{$\rho \norm{\hgamma_{\F}' - \hgamma_{\F}}_{L_2(\Pn)}$} is small relative to $I_{h,\F'}(\hgamma^{\F'})$.
By working with the bound above, we show that subject to some limits on the range of $\rho$, this is the case.

\begin{coro}
\label{coro:weight-similarity}
Under the assumptions of Lemma~\ref{lemma:weight-similarity},
for $\phi \ge I_{h,\F'}(\gamma)$,
$\rho \norm{\hgamma_{\F}' - \hgamma_{\F}}_{L_2(\Pn)} \le \epsilon \phi$ 
if the bounds below are satisfied.
\begin{equation}
\label{eq:rho-lb-implicit}
\begin{aligned}
\rho^2                   &\ge \p{\epsilon^{-1}12 \vee \epsilon^{-2} 36 \vee \epsilon^{-3/2}48 } \sigma_{\F}^2/n, \\
\phi       &\ge \p{\epsilon^{-2} 72 \alpha } \sigma_{\F}^2/n, \\
\rho \phi   &\ge  \p{\epsilon^{-1} 64 \norm{\riesz}_{L_2(\Pn)} 
    			     \vee \epsilon^{-2} 144\norm{\hriesz_{\F} - \gapprox}_{L_2(\Pn)}^{1/2} \norm{\riesz}_{L_2(\Pn)}^{1/2}} \sigma_{\F}^2/n, \\
\rho^3 \phi &\ge \p{\epsilon^{-4} 144^2 \norm{\hriesz_{\F} - \gapprox}_{L_2(\Pn)}} \sigma_{\F}^4/n^2.
\end{aligned}
\end{equation}
\end{coro}

Each of these conditions is a lower bound on an increasing function of $\rho$, as $I_{h,\F'}(\gamma)$ is increasing in $\rho$,
so this is implictly a lower bound on $\rho$. We can simplify these conditions if we can bound
\smash{$\norm{\hriesz_{\F} - \gapprox}_{L_2(\Pn)}$} in terms of $\alpha$ as in Lemma~\ref{lemma:abstract-finite-sample}.

\begin{coro}
\label{coro:weight-similarity-simplified}
Under the assumptions of Corollary~\ref{coro:weight-similarity},
if $\phi \ge I_{h,\F'}(\gamma) \vee \epsilon^{-2} 72 \alpha_{\phi} \sigma_{\F}^2/n$, $\epsilon \le 9/16$,
and $\norm{\hriesz_{\F} - \gapprox}_{L_2(\Pn)}^2 \le \alpha_{\phi} s^2$ and $\alpha_{\phi} \ge n s^2 / \sigma_{\F}^2$
for some $s$ and $\alpha_{\phi} \ge \alpha$, the bounds \eqref{eq:rho-lb-implicit}
are satisfied and therefore $\rho \norm{\hgamma_{\F}' - \hgamma_{\F}}_{L_2(\Pn)} \le \epsilon \phi$ if 
\begin{align*}
\rho &\ge  \epsilon^{-1} 6 \sigma_\F/n^{1/2} \vee (1/2) \norm{\riesz}_{L_2(\Pn)}  \sigma_\F^2/(s^2 n). 
\end{align*}
\end{coro}
\noindent Using the bound on $\norm{\hriesz_{\F} - \gapprox}_{L_2(\Pn)}^2$ from Lemma~\ref{lemma:abstract-finite-sample}, 
we can take $\alpha_{\phi} = \alpha$, $s^2=2\eta_Mr^2$.
Taking $\gamma=\riesz[\psi]$ in \eqref{eq:bias-bound-abstract}
and substituting $\phi \ge I_{h,\F'}(\riesz[\psi]) \vee \epsilon^{-2} 72 \alpha \sigma_{\F}^2/n$
for $I_{h,\F'}(\riesz[\psi])$, when $\rho$ satisfies the lower bound from Corollary~\ref{coro:weight-similarity-simplified}, we get the following oracle bias bound. 
\begin{equation}
\label{eq:bias-bound-concrete}
\begin{aligned}
&I_{h,\F'}(\hgamma^{\F}) \le (1+\epsilon)\phi \\
&+ \sqb{\frac{\sigma_{\F'}^2}{n}\cb{ \norm{\riesz[\psi]}^2_{L_2(\Pn)} \wedge 2\norm{\riesz[\psi]}_{L_2(\Pn)}\p{\norm{\hgamma^{\F}-\riesz[\psi]}_{L_2(\Pn)} + \frac{\epsilon \phi}{\rho} }}}^{1/2} \\
&\le (1+\epsilon + \epsilon') \phi \\ 
&+ \sqb{\frac{\sigma_{\F}^2}{(1-\eta) n}\cb{\norm{\riesz[\psi]}^2_{L_2(\Pn)} \wedge 2\norm{\riesz[\psi]}_{L_2(\Pn)}\norm{\hgamma^{\F}-\riesz[\psi]}_{L_2(\Pn)}}}^{1/2}, \\
(\epsilon')^2 
&= \frac{2\sigma_{\F'}^2\norm{\riesz[\psi]}_{L_2(\Pn)}\epsilon}{n \rho \phi}
\le \frac{2\sigma_{\F'}^2 \norm{\riesz[\psi]}_{L_2(\Pn)} \epsilon }{128\epsilon^{-1}\norm{\riesz[\psi]}_{L_2(\Pn)}\sigma_{\F}^2} = \frac{\epsilon^2}{64(1-\eta)}.
\end{aligned}
\end{equation}
This definition of $(\epsilon')^2$  equates the bracketed term involving $\phi$ and $(\epsilon')^2\phi^2$, 
so the second bound follows by the elementary inequality $\sqrt{a+b}\le \sqrt{a} + \sqrt{b}$.
To bound $(\epsilon')^2$, we've substituted in the denominator one of the lower bounds on $\rho \phi$ from Corollary~\ref{coro:weight-similarity}. 
We conclude the section by proving our lemma and corollaries.

\begin{proof}[Proof of Lemma~\ref{lemma:weight-similarity}]

The bulk of our proof will be devoted to bounding $\norm{\hg_{\F'}- \hg_{\F}}_{L_2(\Pn)}$ where
each $\hg_{\GG}$ is an approximate minimizer of the dual $d_{\GG}$.
We will consider $\hg_{\GG}$ satisfing $d_{\GG}(\hg_{\GG}) \le \min(d_{\GG}(\tilde g), \inf_g d_{\GG}(g) + \epsilon)$
for $\epsilon > 0$.  To simplify our notation, we will work with $d_n := (2/n)d_{\F}$, $d_n' := (2/n)d_{\F'}$, and $\epsilon_n = (2/n)\epsilon$,
and let \smash{$\hg = \hg_{\F}$} and \smash{$\hg' = \hg_{\F'}$} and \smash{$\sigma = \sigma_{\F}$} and \smash{$\sigma' = \sigma_{\F'}$}.

We define $\sigma'$ as we do because it allows us to write $d_n'(g)$ as the sum of $(\sigma'/\sigma)^2 d_n(g)$ and a small remainder.
Observe that  
\[ 
\frac{(\sigma')^2}{\sigma^2} =  \frac{1}{1-\eta} = \frac{\rho^2 n}{\rho^2 n - \sigma^2} = 1 + \frac{\sigma^2}{\rho^2 n - \sigma^2} =  1+\frac{\sigma^2}{\rho^2 n(1-\eta)} = 1+\frac{(\sigma')^2}{\rho^2 n}, \] 
so expanding $\norm{\cdot}_{\F'}^2 = \norm{\cdot}_{\F}^2 + \norm{\cdot}_{L_2(\Pn)}^2/\rho^2$ in the definition \eqref{eq:concrete-dual} of $d_n'$, 
\begin{equation*}
\begin{aligned} 
d_n'(g) 
&= \Pn g^2 - 2\Pn h(\cdot, g)  + ((\sigma')^2/n)(\norm{g}_{\F}^2 + \norm{g}_{L_2(\Pn)}^2/\rho^2) \\
&= (1+(\sigma')^2/(\rho^2 n)) \Pn g^2 - 2\Pn h(\cdot, g)  + ((\sigma')^2/n)\norm{g}_{\F}^2 \\
&= (\sigma'/\sigma)^2 \sqb{\Pn g^2 - 2\Pn h(\cdot, g) +  (\sigma^2/n)\norm{g}_{\F}^2} + 2[(\sigma'/\sigma)^2 - 1]\Pn h(\cdot, g) \\
&= (\sigma'/\sigma)^2 d_n(g) + 2[(\sigma'/\sigma)^2 - 1]\Pn h(\cdot, g).
\end{aligned} 
\end{equation*}
As $\hg'$ and $\hg$ approximately minimize $d_n'$ and $d_n' - 2[(\sigma'/\sigma)^2 - 1]\Pn h(\cdot, g)$,
\begin{align*}
&d_n'(\hat g')  \le d_n'(\hat g) + \epsilon_n, \\
&d_n'(\hat g) - 2[(\sigma'/\sigma)^2 - 1]\Pn h(\cdot, \hat g) \le d_n'(\hat g') - 2[(\sigma'/\sigma)^2 - 1]\Pn h(\cdot, \hat g') + \epsilon_n, \text{ and therefore } \\ 
&d_n'(\hat g')  \le d_n'(\hat g) + \epsilon_n \le d_n'(\hat g') + 2[(\sigma'/\sigma)^2 - 1]\Pn h(\cdot, \hat g - \hat g') + 2\epsilon_n. 
\end{align*}
This implies a bound on the suboptimality of $\hat g$,
\[ d_n'(\hat g) - d_n'(\hat g') \le 2[(\sigma'/\sigma)^2 - 1]\Pn h(\cdot, \hat g - \hat g') + 2\epsilon_n. \]
Furthermore, because $d_n'$ is $2(\sigma'/\sigma)^2$-strongly convex with respect to $\norm{\cdot}_{L_2(\Pn)}$,
\begin{align*}
(1/2)(\sigma'/\sigma)^2 \norm{\hat g - \hat g'}_{L_2(\Pn)}^2 
&\le (1/2)d_n'(\hat g) + (1/2)d_n'(\hat g') - d_n'((\hat g + \hat g')/2) \\
&= (1/2)[d_n'(\hat g) - d_n'(\hat g')] + [d_n'(\hat g') - d_n'((\hat g + \hat g')/2)] \\
&\le [(\sigma'/\sigma)^2 - 1]\Pn h(\cdot, \hat g - \hat g') + 2\epsilon_n. 
\end{align*}
Here we've used the suboptimality bound above and our assumption that $\hat g'$ approximately
minimizes $d_n'$. As $(\sigma/\sigma')^{2}[(\sigma'/\sigma)^2 - 1] = 1- (\sigma/\sigma')^2 = \eta$,
it follows that
\[ \norm{\hat g - \hat g'}_{L_2(\Pn)}^2 \le 2\eta  \Pn h(\cdot, \hat g - \hat g') + 4 (\sigma/\sigma')^2\epsilon_n, \]
and as $\sigma/\sigma' \le 1$,
\begin{equation}
\label{eq:two-gs-l2-bound}
\begin{aligned}
&\norm{\hat g - \hat g'}_{L_2(\Pn)}^2/(2\eta) - (2/\eta)\epsilon_n\\ 
&\le \Pn h(\cdot, \hat g - \hat g')  \\
&=   \sqb{ \Pn[ h(\cdot, \hat g - \hat g') - \gamma (\hat g - \hat g')] 
		   			         	+ \Pn \gamma (\hat g - \hat g') }  \\
&\le I_{h,\F'}(\gamma)\norm{\hat g - \hat g'}_{\F'} + \norm{\gamma}_{L_2(\Pn)} \norm{\hat g - \hat g'}_{L_2(\Pn)} \\ 
&= I_{h,\F'}(\gamma)\sqrt{\norm{\hat g - \hat g'}_{\F}^2 + \rho^{-2}\norm{\hat g - \hat g'}_{L_2(\Pn)}^2} + \norm{\gamma}_{L_2(\Pn)} \norm{\hat g - \hat g'}_{L_2(\Pn)} \\
&\le I_{h,\F'}(\gamma)\norm{\hat g - \hat g'}_{\F} + ( \rho^{-1} I_{h,\F'}(\gamma) + \norm{\gamma}_{L_2(\Pn)}) \norm{\hat g - \hat g'}_{L_2(\Pn)}.
\end{aligned}
\end{equation}

We eliminate the dependence of this bound on $\hg$ by substituting a bound on $\norm{\hat g - \hat g'}_{\F}$.
By the triangle inequality, $\norm{\hat g - \hat g'}_{\F} \le \norm{\hat g - \tilde g}_{\F} + \norm{\hat g' - \tilde g}_{\F}$,
and as $d_n(\hg) \le d_n(\tilde g)$, our assumption about $\excess_{\F}$
implies that the first term is bounded by $\alpha$ and the second by $\alpha + \sqrt{\xi n/\sigma^2}$ if $d_n(\hg') - d_n(\gapprox) \le \xi$.
To establish a bound like this, we use the similarity of
$d_n$ and $d_n'$ like we did above. As $d_n(g) = q d_n'(g) - 2\eta \Pn h(\cdot,g)$ for $q=(\sigma/\sigma')^2$
and $d_n'(\hg') \le d_n'(\gapprox)$, either $d_n(\hg') \le d_n(\gapprox)$ or
\[ \underset{d_n(\gapprox)}{ q d_n'(\gapprox) - 2\eta \Pn h(\cdot, \gapprox)}  \le
   \underset{d_n(\hg')}{q d_n'(\hg') - 2\eta \Pn h(\cdot,\hg')} \le    
   q d_n'(\gapprox) - 2\eta \Pn h(\cdot,\hg'). \]
Consequently, $d_n(\hg') - d_n(\gapprox) \le \xi$ for $\xi= 2\eta\max(0, \Pn h(\cdot, \gapprox - \hg'))$.
It follows that $\norm{\hg' - \tilde g}_{\F} \le \alpha + \sqrt{\xi n/\sigma^2}$,
and the bound remains valid if we subsititute an upper bound on $\sqrt{\xi}$.
We derive an upper bound as follows.
\begin{align*} 
\abs{\Pn h(\cdot, \gapprox - \hg')} 
&\le \abs{\Pn h(\cdot, \gapprox - \hg') - \gamma (\gapprox - \hg')} + \abs{\Pn \gamma (\gapprox - \hg')} \\
&\le \norm{\hg' - \gapprox}_{\F'} I_{h,\F'}(\gamma) + \norm{\gamma}_{L_2(\Pn)} \norm{\hg'- \gapprox}_{L_2(\Pn)} \\
&\le \norm{\hg' - \gapprox}_{\F} I_{h,\F'}(\gamma) + (\rho^{-1}I_{h,\F'}(\gamma) + \norm{\gamma}_{L_2(\Pn)}) \norm{\hg'- \gapprox}_{L_2(\Pn)},
\end{align*}
so 
\[ \norm{\hg' - \tilde g}_{\F} \le \alpha 
+ \sqrt{2\eta n/\sigma^2}\sqb{\norm{\hg' - \gapprox}_{\F}^{1/2} I_{h,\F'}^{1/2}(\gamma) 
	   + \sqrt{\rho^{-1}I_{h,\F'}(\gamma) + \norm{\gamma}_{L_2(\Pn)}} \norm{\hg'- \gapprox}_{L_2(\Pn)}^{1/2}}.
\]
Here $\sqrt{2\eta n/\sigma^2} = \sqrt{2/\rho^2}$, and this 
is a quadratic inequality $y^2 \le  by + c$ for 
\begin{align*} 
y&=\norm{\hg' - \tilde g}_{\F}^{1/2},\\
b&=\sqrt{2} \rho^{-1} I_{h,\F'}^{1/2}(\gamma),\\
c&= \alpha + c_{\star}\norm{\hg'- \gapprox}_{L_2(\Pn)}^{1/2}, 
\quad  c_{\star}^2 = 2\rho^{-3}I_{h,\F'}(\gamma) + 2\rho^{-2} \norm{\gamma}_{L_2(\Pn)}. 
\end{align*}
Its solutions satisfy $y^2 \le (b + \sqrt{b^2+4c})^2/4 \le b^2 + 4c$, so
\begin{align*}
\norm{\hg' - \tilde g}_{\F} 
&\le 2\rho^{-2} I_{h,\F'}(\gamma) + 4\alpha + 4 c_{\star} \p{\norm{\hg-\gapprox}_{L_2(\Pn)}^{1/2} + \norm{\hg'-\hg}_{L_2(\Pn)}^{1/2}}.
\end{align*}
Substituting this in \eqref{eq:two-gs-l2-bound},
\begin{align*}
&\norm{\hat g - \hat g'}_{L_2(\Pn)}^2/(2\eta) \\ 
&\le \sqb{\rho^{-1}I_{h,\F'}(\gamma) + \norm{\gamma}_{L_2(\Pn)}} \norm{\hg' - \hg}_{L_2(\Pn)} \\
&+ \sqb{(2/\eta)\epsilon_n + \p{4\alpha  + 2\rho^{-2}I_{h,\F'}(\gamma) 
      + 4c_{\star} \norm{\hg - \gapprox}_{L_2(\Pn)}^{1/2}} I_{h,\F'}(\gamma)} \\
&+ \sqb{4c_{\star} I_{h,\F'}(\gamma)} \norm{\hg'-\hg}_{L_2(\Pn)}^{1/2}.
\end{align*}
This is $ax^2 \le bx + c + dx^{1/2}$ for $x=\norm{\hat g - \hat g'}_{L_2(\Pn)}$, $a=1/(2\eta)$, and successive bracketed factors $b$,$c$,$d$.
This implies that $ax^2 \le 3\max(bx, c, dx^{1/2})$ and therefore that $x \le \max(3b/a, (3c/a)^{1/2}, (3d/a)^{2/3})$. 
Expanding $a$,$b$,$c$,$d$,
\begin{align*}
&\norm{\hat g - \hat g'}_{L_2(\Pn)} 
\le 6\eta\sqb{\rho^{-1}I_{h,\F'}(\gamma) + \norm{\gamma}_{L_2(\Pn)}} \\
&\vee   \sqb{12\epsilon_n + 6\eta\p{4\alpha  + 2\rho^{-2}I_{h,\F'}(\gamma) 
      + 4c_{\star} \norm{\hg - \gapprox}_{L_2(\Pn)}^{1/2}} I_{h,\F'}(\gamma) }^{1/2} \\
&\vee   \sqb{24 \eta c_{\star} I_{h,\F'}(\gamma)}^{2/3}. 
\end{align*}
This bound is satisfied with $\hgamma$ and $\hgamma'$ in place of $\hg$ and $\hg'$ and $\epsilon_n=0$, 
as $\hgamma_i=\lim_{j\to\infty}\hg_j(Z_i)$ and $\hgamma_i'=\lim_{j \to \infty}\hg_j'(Z_i)$
for approximate minimizers $\hg_j$ and $\hg_j'$ satisying our conditions for $\epsilon^j \to 0$. 
We derive our claimed bound by multiplying by $\rho$.
\end{proof}

\begin{proof}[Proof of Corollary~\ref{coro:weight-similarity}]
Throughout this proof, we will write $\sigma$ meaning $\sigma_\F$.
We work with the bound from Lemma~\ref{lemma:weight-similarity},
which we relax by substituting the upper bound $\phi$ for $I_{h,\F}(\gamma)$.
Then, within each branch of the maximum, we will allocate to each term in our bound a fraction of $\epsilon \phi$.

Consider the first branch. Recalling that $\eta = \sigma^2/(\rho^2 n)$,
\begin{align*}
&6\eta \phi \le \epsilon_{1,1} \phi 
    && \text{ if } 6\sigma^2/n \le \epsilon_{1,1} \rho^2, \\
&6\eta\rho\norm{\gamma}_{L_2(\Pn)} \le \epsilon_{1,2}  \phi
    && \text{ if } 6\norm{\gamma}_{L_2(\Pn)}\sigma^2/n \le \epsilon_{1,2} \rho \phi
\end{align*}
It is bounded by $\epsilon_1 \phi$ for $\epsilon_1 = \epsilon_{1,1}+\epsilon_{1,2}$ if these conditions are satisfied.

Now consider the second branch. $(a\phi)^{1/2} \le \epsilon_2 \phi$ if
$a \le \epsilon_2^2 \phi$, so we will show that each term $a_j$ in $a$ satifies
$a_j \le \epsilon_{2,j}^2 \phi$. It will follow that their sum satisfies
$a \le \epsilon_2^2 \phi$ for $\epsilon_2^2 = \sum_{j}\epsilon_{2,j}^2$,
and therefore that the second branch is bounded by $\epsilon_2 \phi$.
We now bound each term $a_j$.
\begin{align*} 
&24\eta\rho^2\alpha \le \epsilon_{2,1}^2 \phi
    && \text{ if } 24\alpha \sigma^2/n \le \epsilon_{2,1}^2 \phi, \\
&12\eta \phi \le \epsilon_{2,2}^2 \phi
    && \text{ if } 12\sigma^2/n \le \epsilon_{2,2}^2 \rho^2, \\
&24\eta \rho^2 c_{\star} \norm{\hriesz_{\F} - \gapprox}_{L_2(\Pn)}^{1/2} \le \epsilon_{2,3}^2  \phi
    && \text{ if }\ \  48\norm{\hriesz_{\F} - \gapprox}_{L_2(\Pn)}^{1/2}\sigma^2/n \le \epsilon_{2,3}^2 \rho^{3/2} \phi^{1/2},  \\
& && \text{and }    48\norm{\hriesz_{\F} - \gapprox}_{L_2(\Pn)}^{1/2}\norm{\gamma}_{L_2(\Pn)}^{1/2}\sigma^2/n \le \epsilon_{2,3}^2 \rho \phi.
\end{align*}
For the third term, we've used the bound $c_{\star} \le \max(4\rho^{-3}\phi,\ 4\rho^{-2}\norm{\gamma}_{L_2(\Pn)})^{1/2}$.

Finally, consider the third branch. $(a\phi)^{2/3} \le \epsilon_3 \phi$ if
$a^2 \le \epsilon_3^3 \phi$, so we will show that 
we will show that the two terms $a_j^2$ in $a^2=48^2\eta^2\rho^3c_{\star}^2$ 
satisfy $a_j^2 \le \epsilon_{3,j}^3 \phi$ for $\epsilon_3^3 = \sum_{j}\epsilon_{3,j}^3$.
It will follow that the third term is bounded by $\epsilon_3 \phi$.
We now bound these two terms.
\begin{align*}
& 24^2 \eta^2 \rho^3 \cdot 2 \rho^{-3} \phi \le \epsilon_{3,1}^3 \phi
&&\text{ if } 1152\sigma^4/n^2 \le \epsilon_{3,1}^3 \rho^4, \\
& 24^2 \eta^2 \rho^3 \cdot 2 \rho^{-2} \norm{\gamma}_{L_2(\Pn)} \le \epsilon_{3,2}^3 \phi
&&\text{ if } 1152\norm{\gamma}_{L_2(\Pn)}\sigma^4/n^2 \le \epsilon_{3,2}^3 \rho^3 \phi. 
\end{align*}

Bounding the maximum over the three branches by the maximum of our bounds,
$\rho\norm{\hriesz_{\F} - \hriesz_{\F'}}_{L_2(\Pn)} \le \epsilon \phi$
for $\epsilon=\max_{i \in 1 \ldots 3}\epsilon_{i}$ if 
\begin{align*}
\rho^2 &\ge \p{\epsilon_{1,1}^{-1}6 \vee \epsilon_{2,2}^{-2} 12 \vee \epsilon_{3,1}^{-3/2}1152^{1/2} } \sigma^2/n, \\
\phi      &\ge \p{\epsilon_{2,1}^{-2} 24\alpha } \sigma^2/n, \\
\rho \phi &\ge  \p{\epsilon_{1,2}^{-1} 6 \norm{\riesz}_{L_2(\Pn)} 
    			\vee \epsilon_{2,3}^{-2} 48\norm{\hriesz_{\F} - \gapprox}_{L_2(\Pn)}^{1/2} \norm{\riesz}_{L_2(\Pn)}^{1/2}} \sigma^2/n, \\
\rho^3 \phi &\ge \p{\epsilon_{2,3}^{-4} 48^2 \norm{\hriesz_{\F} - \gapprox}_{L_2(\Pn)} \vee \epsilon_{3,2}^{-3} 1152\norm{\gamma}_{L_2(\Pn)}} \sigma^4/n^2.
\end{align*}

To simplify these conditions, first set $\epsilon=\epsilon_1=\epsilon_2=\epsilon_3$ and equally divide
contributions to the $\epsilon_{i}$ between the $\epsilon_{i,j}$ respectively,
taking $\epsilon_{1,j}=\epsilon/2$, $\epsilon_{2,j}^2 = \epsilon^2/3$, and $\epsilon_{3,j}^3 = \epsilon^3/2$.
\begin{align*}
\rho^2 &\ge \p{\epsilon^{-1}12 \vee \epsilon^{-2} 36 \vee \epsilon^{-3/2}48 } \sigma^2/n, \\
\phi      &\ge \p{\epsilon^{-2} 72 \alpha } \sigma^2/n, \\
\rho \phi &\ge  \p{\epsilon^{-1} 12 \norm{\riesz}_{L_2(\Pn)} 
    			\vee \epsilon^{-2} 144\norm{\hriesz_{\F} - \gapprox}_{L_2(\Pn)}^{1/2} \norm{\riesz}_{L_2(\Pn)}^{1/2}} \sigma^2/n, \\
\rho^3 \phi &\ge \p{\epsilon^{-4} 144^2 \norm{\hriesz_{\F} - \gapprox}_{L_2(\Pn)} \vee \epsilon^{-3} 48^2 \norm{\gamma}_{L_2(\Pn)}} \sigma^4/n^2.
\end{align*}
In our lemma statement, we increase the first lower bound on $\rho \phi$ to \smash{$\epsilon^{-1} 64 \norm{\riesz}_{L_2(\Pn)}$}
and then drop the second lower bound on \smash{$\rho^3 \phi$.} The dropped bound is implied by multiplying 
this lower bound on $\rho \phi$ and our lower bound $\epsilon^{-2} 36$ on $\rho^2$.
\end{proof}

\begin{proof}[Proof of Corollary~\ref{coro:weight-similarity-simplified}]
Throughout this proof, we will write $\sigma$ meaning $\sigma_{\F}$.
We will choose $\rho$ so that the bounds \eqref{eq:rho-lb-implicit} are satisfied.
As $\phi \ge \epsilon^{-2}72\alpha_\phi \sigma^2/n$,
we have the lower bounds $\rho^3 \phi \ge \rho^3 \epsilon^{-2}72\alpha_\phi \sigma^2/n$
and $\rho \phi \ge  \rho \epsilon^{-2}72\alpha_\phi \sigma^2/n$. 
These exceed the corresponding bounds
from Corollary~\ref{coro:weight-similarity} as follows. 
\begin{align*} 
&\rho^3 \epsilon^{-2}72\alpha_\phi \sigma^2/n \ge \epsilon^{-4} 144^2 \norm{\hriesz_{\F} - \gapprox}_{L_2(\Pn)} \sigma^4/n^2 \\
&\qquad \text{ if } \quad  \rho^3 \ge \epsilon^{-2} (144^2/72) \norm{\hriesz_{\F} - \tilde g}_{L_2(\Pn)}\sigma^2/(\alpha_\phi n). \\
&\rho \epsilon^{-2}72\alpha_\phi \sigma^2/n \ge \epsilon^{-1} 64 \norm{\riesz}_{L_2(\Pn)} \sigma^2/n \\
&\qquad \text{ if } \quad  \rho   \ge \epsilon (72/64) \norm{\riesz}_{L_2(\Pn)}/\alpha_\phi, \\
&\rho \epsilon^{-2}72\alpha_\phi \sigma^2/n \ge \epsilon^{-2} 144\norm{\hriesz_{\F} - \gapprox}_{L_2(\Pn)}^{1/2} \norm{\riesz}_{L_2(\Pn)}^{1/2} \sigma^2/n \\
&\qquad \text{ if } \quad \rho   \ge (144/72) \norm{\hriesz_{\F} - \gapprox}_{L_2(\Pn)}^{1/2}\norm{\riesz}_{L_2(\Pn)}^{1/2} /\alpha_\phi. 
\end{align*}
Simplifying fractions and substituting the upper bound $\alpha_\phi^{1/2}s \ge \norm{\hriesz_{\F} - \tilde g}_{L_2(\Pn)}$
in the numerator, these bounds hold if
\begin{align*} 
\rho^3 &\ge \epsilon^{-2} 288 s \sigma^2 / (\alpha_\phi^{1/2} n), \\
\rho   &\ge \epsilon (8/9) \norm{\riesz}_{L_2(\Pn)}/\alpha_\phi, \\
\rho   &\ge 2 s^{1/2} \norm{\riesz}_{L_2(\Pn)}^{1/2} /\alpha_\phi^{3/4}, \\
\end{align*} 
And substituting the lower bound $\alpha_\phi \ge n s^2 / \sigma^2$ in the denominators, these hold if
\begin{align*} 
\rho   &\ge \epsilon^{-2/3} 288^{1/3}  \sigma / n^{1/2} && \text{ or equivalently }  \rho^3 \ge \epsilon^{-2} 288  \sigma^3 / n^{3/2}, \\
\rho   &\ge \epsilon (8/9) \norm{\riesz}_{L_2(\Pn)} \sigma^2/(s^2 n), &&\\
\rho   &\ge 2 \norm{\riesz}_{L_2(\Pn)}^{1/2} \sigma^{3/2}/(s n^{3/4}), && \\
\end{align*} 
As $\phi \ge \epsilon^{-2}72\alpha \sigma^2/n$ by construction, 
the bounds \eqref{eq:rho-lb-implicit} from Corollary~\ref{coro:weight-similarity}
hold if the bounds above and the explicit lower bounds on $\rho^2$ from \eqref{eq:rho-lb-implicit} do.
That is, if 
\begin{align*} 
\rho &\ge  \p{\epsilon^{-1/2}\sqrt{12} \vee \epsilon^{-2/3} 288^{1/3} \vee \epsilon^{-3/4}\sqrt{48} \vee \epsilon^{-1} \sqrt{36} }\sigma/n^{1/2} \\
     &\vee \epsilon (8/9) \norm{\riesz}_{L_2(\Pn)} \sigma^2/(s^2 n) \\
     &\vee  2 \norm{\riesz}_{L_2(\Pn)}^{1/2} \sigma^{3/2}/(s n^{3/4}).
\end{align*}

The first term in this lower bound will be $6\epsilon^{-1}\sigma/\sqrt{n}$ 
if $\epsilon \le (36/48)^2 \wedge 36^{3/2}/288 \wedge 36/12 = (36/48)^2 = (3/4)^2$. 
And when this holds, $(1/2)\norm{\riesz}_{L_2(\Pn)}\sigma^2/(s^2 n)$ exceeds the second term.
This yields the simplified bound 
\begin{align*} 
\rho &\ge  \epsilon^{-1} 6 \sigma/n^{1/2} \vee (1/2) \norm{\riesz}_{L_2(\Pn)} \sigma^2/(s^2 n) \vee  2 \norm{\riesz}_{L_2(\Pn)}^{1/2} \sigma^{3/2}/(s n^{3/4}).
\end{align*}
In our stated bound, we drop the third term.
It is not maximal, as it is smaller than the geometric mean of the first two,
which is $\epsilon^{-1/2} \sqrt{3} \norm{\riesz}_{L_2(\Pn)}^{1/2} \sigma^{3/2}/(s n^{3/4})$
with $\epsilon^{-1/2} \sqrt{3} \ge (4/3)\sqrt{3} \ge 2$.

\end{proof}

\subsection{Putting it all together}
\label{sec:putting-it-all-together}
The assumptions of Lemma~\ref{lemma:consistency-deterministic} imply the assumption of
Lemma~\ref{lemma:weight-similarity} concerning $\excess_{\F}$ with the same values of $\tilde g$ and $\alpha$.
Thus, on the intersection of an event of probability $1-\delta$, on which our noise term bound \eqref{eq:noise-term-deviation} holds,
and an event on which the ratio process bounds \eqref{eq:ratio-process-bounds} hold for some $r > 0$,

\begin{equation}
\label{eq:asymptotic-linearity-abstract}
\begin{aligned}
&\abs{\hpsi - \tilde\psi(m) - \sum_{i=1}^n \gapprox(Z_i)(Y_i - m(Z_i))} \le \delta^{-1/2} n^{-1/2} \norm{v}_{\infty} \norm{\hriesz - \gapprox}_{L_2(\Pn)} \\
&+ \norm{\hm-m}_{\F'_{\rho}}\sqb{1+\frac{2\epsilon}{\sqrt{1-\eta_{\rho}}}}  \phi(\rho) \\
&+ \norm{\hm-m}_{\F'_{\rho}}\sqb{\frac{2\sigma^2}{(1-\eta_{\rho}) n}\cb{\norm{\riesz[\psi]}_{L_2(\Pn)}^2 \wedge \norm{\riesz[\psi]}_{L_2(\Pn)}\norm{\hgamma-\riesz[\psi]}_{L_2(\Pn)}}}^{1/2} \\
&\text{for \quad} \eta_\rho = \sigma^2/(\rho^2 n), \\ 
&\text{\hphantom{for \quad}} \F'_{\rho} = \set{f : \norm{f}_{\F}^2 + \rho^{-2}\norm{f}_{L_2(\Pn)}^2 \le 1}, \\
&\text{\hphantom{for \quad}} \phi(\rho) \ge I_{h,\F_{\rho}'}(\riesz[\psi]) \vee \epsilon^{-2} 72 \alpha_{\phi} \sigma_{\F}^2/n.
\end{aligned}
\end{equation}
Here we've used the bounds \eqref{eq:noise-term-deviation} and \eqref{eq:bias-bound-concrete}
on the noise and bias terms in our error decomposition \eqref{eq:error-decomp},
substituting $2\epsilon/\sqrt{1-\eta_{\rho}} \ge \epsilon + \epsilon'$ into \eqref{eq:bias-bound-concrete}.
It holds, with $\epsilon \le 9/16$ and $\alpha_{\phi} = \alpha \vee 2 \eta_M r_{\phi}^2 n / \sigma^2$ for $\alpha$ as in Lemma~\ref{lemma:consistency-deterministic},
when $\rho$ satisfies the lower bound of Corollary~\ref{coro:weight-similarity-simplified} 
with \smash{$s^2=2\eta_Mr_{\phi}^2$} for $r_{\phi} \ge r$, as these are sufficient conditions for the bound \eqref{eq:bias-bound-concrete} to hold
as a consequence of Lemma~\ref{lemma:abstract-finite-sample} and Corollary~\ref{coro:weight-similarity-simplified}.

To complete our proof of Theorem~\ref{theo:simple-rate}, 
we show in Section~\ref{sec:ratio-process-bounds} that the ratio process bounds \eqref{eq:ratio-process-bounds} are satisfied with high probability,
show in Section~\ref{sec:bounding-I} that a certain function $\phi$ bounds $I_{h,\F_{\rho}'}(\riesz[\psi])$ with high probability,
and \ldots in Section \ldots.
In the first two steps, we will use the assumption that $\F$ is uniformly bounded, giving bounds that depend on 
$M_{\infty}(\F)=\sup_{f\in\F}\norm{f}_{\infty}$. 
After we have concluded our proof, in Section~\ref{sec:not-uniformly-bounded}, we will briefly discuss techniques for relaxing this assumption.

\subsubsection{Ratio Process Bounds}
\label{sec:ratio-process-bounds}
Our first bound in \eqref{eq:ratio-process-bounds}, 
a uniform lower bound on the ratio $\Pn f^2/ \P f^2$, holds under a wide range of conditions. These are summarized in \citet{mendelson2017extending}, where
Corollary 3.6 addresses the uniformly bounded case we consider here. 
It establishes that for any $\eta_Q < 1$,
the bound $\Pn f^2 \ge \eta_Q \P f^2$ holds for all $f \in \F$
satisfying $\P f^2 \ge r^2$ with probability $1-2\exp(-c_2nr^2/M_{\infty}^2(\F))$ if 
$R_n(\F_{c_0r}) \le c_1 r^2/M_{\infty}(\F)$ for constants $c_0,c_1,c_2$ that depend only on $\eta_Q$.
And by a scaling argument of \citet*[Lemmas 3.2, 3.4]{bartlett2005local}, 
there is a unique positive $r_Q$ that satisfies the fixed point condition $R_n(\F_{c_0r}) \le c_1 r^2/M_{\infty}(\F)$ 
with equality, and it is satisfied for all $r \ge r_Q$.

For $g=\gapprox$ when $Q=\P$ and $g = \riesz[\psi]$ when $Q=\Pn$, our second bound in \eqref{eq:ratio-process-bounds} is 
on the supremum of the mean-zero empirical process indexed by the 
image $h_{g}(\cdot, \F_r)$ of $\F_r$ under the function $h_{g}(z,f) = h(z,f) - g(z)f(z)$.
By Markov's inequality, this is bounded by 
$\delta^{-1} \E \sup_{h \in h_{g}(\cdot,\F_r)}\abs{(\Pn-\P) h}$ with probability $1-\delta$.
Furthermore, if we prefer to state our bounds in terms of Rademacher complexities,
via symmetrization this is bounded by $2\delta^{-1}R_n(h_{g}(\cdot,\F_r))$ 
\citep[Lemma 2.3.1]{vandervaart-wellner1996:weak-convergence}. 
By the aforementioned scaling argument, 
there is a unique positive $r_M$ satisfying the fixed point condition $2\delta^{-1}R_n(h_{g}(\cdot,\F_r)) \le \eta_M r^2$
with equality, and it is satisfied for all $r \ge r_M$.

In summary, our ratio process bounds \eqref{eq:ratio-process-bounds} hold 
on an event of probability $1-\delta-2\exp(-c_2nr^2/M_{\infty}^2(\F))$ 
for $r \ge r_Q \vee r_M$ where 
\begin{align*}
r_Q &= \inf\set{ r > 0 : R_n(\F_{c_0r}) \le c_1 r^2/M_{\infty}(\F)}, & \\
r_M &= \begin{cases}
    \inf \set{ r > 0 : R_n(h_{\gapprox}(\cdot,\F_r)\ \le \ \delta \eta_M r^2/2} & \text{ for }\quad Q=\P, \\
    \inf \set{ r > 0 : R_n(h_{\riesz[\psi]}(\cdot,\F_r) \le \delta \eta_M r^2/2} & \text{ for }\quad Q=\Pn. \end{cases}
\end{align*}
Here $\eta_Q \in [0,1)$ and $\eta_M > 0$ are arbitrary
and $c_0 \ldots c_2$ are constants dependening on only on $\eta_Q$.
It follows that for such $r$,
the bound \eqref{eq:asymptotic-linearity-abstract} holds with probability $1-2\delta-2\exp(-c_2nr_Q^2/M_{\infty}^2(\F))$
for all $\rho$ satisfying the lower bound from Corollary~\ref{coro:weight-similarity-simplified}.

\subsubsection{Bounding $I_{h,\F_{\rho}'}(\riesz[\psi])$} 
\label{sec:bounding-I}

To bound $I_{h,\F'_{\rho}}(\riesz[\psi])$, we first observe that it is smaller than $I_{h,\GG}(\riesz[\psi])$ 
for any $\GG \supseteq \F'_{\rho}$. As $\F'_{\rho} \subseteq \set{ f \in \F : \Pn f^2 \le \rho^2}$, 
it is contained in $\F_{\sqrt{2}\rho} = \set{ f \in \F : \P f^2 \le 2\rho^2}$ 
for $\rho^2 \ge 20M_{\infty}(\F)R_n(\F_{\rho}) + 26M_{\infty}^2(\F)\log(1/\delta')/n$ on an event of probability $1-\delta'$ \citep[Lemma 3.6]{bartlett2005local}.
Setting $\log(1/\delta') = (20/26)nR_n(\F_{\rho})/M_{\infty}(\F)$, the two terms in this lower bound on $\rho^2$ are equal,
so this containment holds for $\rho^2 \ge 40M_{\infty}(\F)R_n(\F_{\rho})$
on an event of probability $1-\exp\set{- (20/26)nR_n(\F_{\rho})/M_{\infty}(\F)}$.
Using a constant $c_1 \le 1/40$ in the fixed point condition $R_n(\F_{c_0r}) \le c_1 r^2/M_{\infty}(\F)$ 
in the previous section, this condition on $\rho$ is satisfied for $\rho \ge c_0 r_Q$, and for such $\rho$,
$R_n(\F_{\rho}) \ge R_n(\F_{c_0 r_Q}) = c_1 r_Q^2 / M_{\infty}(\F)$, so the probability of this event is
at least \smash{$1-\exp\set{- (20/26)nc_1 r_Q^2/M_{\infty}^2(\F)}$}. By the union bound, it follows
that this containment and 
\eqref{eq:asymptotic-linearity-abstract} hold on an event of probability \smash{$1-2\delta-3\exp(-c_2nr_Q^2/M_{\infty}^2(\F))$},
taking $c_2$ to be no larger than $(20/26)c_1$.

On the intersection of this event and the probability $1-\delta$ event on which Markov's inequality implies 
$I_{h,\F_{\sqrt{2}\rho}}(\riesz[\psi]) \le \delta^{-1}\E I_{h,\F_{\sqrt{2}\rho}}(\riesz[\psi])$,
it follows that $I_{h,\F'_{\rho}}(\riesz[\psi]) \le  \delta^{-1}\E I_{h,\F_{\sqrt{2}\rho}}(\riesz[\psi])$
and a variant of \eqref{eq:asymptotic-linearity-abstract} in which 
$I_{h,\F'_{\rho}}(\riesz[\psi])$ is replaced with the upper bound \smash{$\delta^{-1}\E I_{h,\F_{\sqrt{2}\rho}}(\riesz[\psi])$} holds
 for $\rho$ equal to or exceeding both $c_0 r_Q$ and the lower bound from Corollary~\ref{coro:weight-similarity-simplified}.
We use a deterministic variant of the latter in which $\norm{\riesz[\psi]}_{L_2(\Pn)}$ is replaced with the probability $1-\delta$
Markov's inequality bound \smash{$\delta^{-1/2} \norm{\riesz[\psi]}_{L_2(\P)}$.} Recalling that 
we take $s^2=2\eta_M r_{\phi}^2$ in Corollary~\ref{coro:weight-similarity-simplified}, our bound on $\rho$ is
\begin{equation}
\label{eq:rho-lb}
\rho \ge  c_0 r_Q \vee \frac{6 \epsilon^{-1} \sigma}{n^{1/2}} 
     \vee \frac{\norm{\riesz[\psi]}_{L_2(\P)} \sigma^2}{4\delta^{1/2} \eta_M r_\phi^2 n}.
\end{equation}
The intersection of these events has probability at least
$1-4\delta-3\exp(-c_2nr_Q^2/M_{\infty}^2(\F))$ by the union bound.

\subsubsection{A concrete bound}
We state a bound summarizing the results above. Let
$\gapprox = \argmin_{g} \norm{\riesz[\psi] - g}_{L_2(Q)}^2 + (\sigma^2/n)\norm{g}_{\F}^2$. With probability $1-4\delta-3\exp(-c_2nr_Q^2/M_{\infty}^2(\F))$, 
\begin{equation}
\label{eq:asymptotic-linearity-concrete}
\begin{aligned}
&\norm{\hriesz - \gapprox}_{L_2(\Pn)}^2 \le 2\alpha r^2 \quad \text{ for } \alpha = 3\p{\norm{\gapprox}_{\ff} + r^2 n/\sigma^2} \vee 4, \\
&\sqrt{n}\abs{\hpsi - \tilde\psi(m) - \sum_{i=1}^n \gapprox(Z_i)(Y_i - m(Z_i))} \le \delta^{-1/2} \norm{v}_{\infty} \norm{\hriesz - \gapprox}_{L_2(\Pn)} \\
&+ \norm{\hm-m}_{\F'_{\rho}}\sqrt{n}\phi(\rho) \sqb{1+\frac{2\epsilon}{\sqrt{1-\epsilon^2/36}}}  \\
&+ \norm{\hm-m}_{\F'_{\rho}}\sqb{\frac{2\sigma^2}{\set{1-\epsilon^2/36}}\cb{\norm{\riesz[\psi]}_{L_2(\Pn)}^2 \wedge \norm{\riesz[\psi]}_{L_2(\Pn)}\norm{\hgamma-\riesz[\psi]}_{L_2(\Pn)}}}^{1/2}, \\
& \text{ \quad with \quad }            r = r_Q \vee r_M, \\
& \text{\hphantom{ \quad with \quad }} r_Q = \inf\set{ r > 0 : R_n(\F_{c_0r}) \le c_1 r^2/M_{\infty}(\F)}, \\
& \text{\hphantom{ \quad with \quad }} r_M = \begin{cases} \inf\set{ r > 0 : R_n(h_{\gapprox}(\cdot,\F_r)  \ \le \ \delta r^2/2} & \text{ for } \quad Q=\P,  \\
							   \inf\set{ r > 0 : R_n(h_{\riesz[\psi]}(\cdot,\F_r) \le \delta r^2/2} & \text{ for } \quad Q=\Pn, \end{cases}  \\
& \text{\hphantom{ \quad with \quad }} \phi(\rho) = \delta^{-1}\E I_{h,\F_{\sqrt{2}\rho}}(\riesz[\psi]) \vee \epsilon^{-2} 72 \alpha_{\phi} \sigma^2/n, \\
& \text{\hphantom{ \quad with \quad }} \alpha_{\phi} = \alpha \vee 2nr_{\phi}^2/\sigma^2 \quad \text{ for any } \quad r_{\phi} \ge r \\
& \text{ \quad for any \quad } \rho \ge c_0 r_Q \vee \frac{6 \sigma}{\epsilon \sqrt{n}} 
     \vee \frac{\norm{\riesz[\psi]}_{L_2(\P)} \sigma^2}{4\sqrt{\delta}n r_\phi^2}. 
\end{aligned}
\end{equation}
Here $h_{\gamma}(z,f) = h(z,f) - \gamma(z)f(z)$, $c_0 \ldots c_2$ are universal constants, and $\epsilon \le 9/16$.
To derive this bound, we have taken $\eta_Q=1/2$ and $\eta_M=1$, 
used Lemma~\ref{lemma:abstract-finite-sample} to bound \smash{$\norm{\hriesz - \gapprox}_{L_2(\Pn)}$},
and substituted into \eqref{eq:asymptotic-linearity-abstract} the bounds discussed in the subsections above,
as well as the bound \smash{$\epsilon^2/36 \ge \eta_{\rho}$} implied by the condition $\rho \ge 6 \epsilon^{-1} \sigma / n^{1/2}$.

To simplify our lower bound on $\rho$, we set \smash{$r_\phi^2 = \norm{\riesz[\psi]}_{L_2(\P)} \sigma^2 / (4c_0\sqrt{\delta}n r)$}
to equate $c_0 r$ and \smash{$\norm{\riesz[\psi]}_{L_2(\P)} \sigma^2 / (4\sqrt{\delta}n r_\phi^2)$}. Taking $c_0 \ge 1$, 
this satisfies our assumption $r \ge r_{\phi}$, and by design our lower bound on $\rho$ simplifies 
to $c_0 r \vee  6 \sigma/(\epsilon \sqrt{n})$. For this $r_{\phi}$,
\smash{$\alpha_{\phi}= 3\p{\norm{\gapprox}_{\ff} + r^2 n/\sigma^2} \vee $} \smash{$ \norm{\riesz[\psi]}_{L_2(\P)} / (2c_0\sqrt{\delta} r) \vee 4$},
so the bound above holds for \smash{$\rho \ge c_0 r \vee  6 \sigma/(\epsilon \sqrt{n})$} and
\[
\phi(\rho) = \delta^{-1}\E I_{h,\F_{\sqrt{2}\rho}}(\riesz[\psi]) 
	  \vee \frac{216}{\epsilon^2}\p{\frac{\norm{\gapprox}_{\ff}\sigma^2}{n} + r^2} 
          \vee \frac{36\norm{\riesz[\psi]}_{L_2(\P)}\sigma^2}{\epsilon^2\sqrt{\delta} c_0 nr} \vee \frac{288\sigma^2}{\epsilon^2 n}. 
\]
In our definition of $\phi(\rho)$ in Theorem~\ref{theo:simple-rate}, we substitute the bound \smash{$2R_n(h_{\riesz[\psi]}(\cdot, \F_{\sqrt{2}\rho}))\ge$} \smash{$\E I_{h,\F_{\sqrt{2}\rho}}(\riesz[\psi])$} implied by symmetrization \citep[Lemma 2.3.1]{vandervaart-wellner1996:weak-convergence}.

\paragraph*{Approximately optimizing over $\rho$}
Rather than including $\rho$ explicitly in our bound, 
we approximately optimize over $\rho$ exceeding the lower bound above, which we will call $\rho_{\phi}$.
To do this, we will work with bounds $\norm{\hm - m}_{\F} \le s_{\F}$ and $\norm{\hm-m}_{L_2(\Pn)} \le s_{L_2(\Pn)}$.
Subject to the additional constraint \smash{$\rho \ge s_{L_2(\Pn)}/s_{\F}$}, 
we increase our bound by substituting \smash{$\sqrt{2}s_{\F}$} for $\norm{\hm - m}_{\F_{\rho}'}$, as
\[ \norm{\hm - m}_{\F'_{\rho}}^2 = \norm{\hm - m}_{\F}^2 + \rho^{-2}\norm{\hm - m}_{L_2(\Pn)}^2 
			         \le s_{\F}^2\p{1+ \rho^{-2}s_{L_2(\Pn)}^2/s_{\F}^2}.\]
Thus, our bound \eqref{eq:asymptotic-linearity-concrete} holds for
\smash{$\rho = \rho_{\phi} \vee (s_{L_2(\Pn)}/s_{\F})$}, 
and for this $\rho$, \smash{$\sqrt{2}s_{\F} \ge \norm{\hm - m}_{\F'_{\rho}}$.}
Making these substitutions yields the claim of Theorem~\ref{theo:simple-rate}. 

\subsection{Doing without uniform boundedness}
\label{sec:not-uniformly-bounded}
In Section~\ref{sec:ratio-process-bounds},
we show that the ratio process bounds \eqref{eq:ratio-process-bounds}
hold with high probability when $\F$ is uniformly bounded.
Lower bounds on the ratio process $\Pn f^2 / \P f^2$, like 
our first bound in \eqref{eq:ratio-process-bounds},
hold for classes with $M_p(\F) = \sup_{f \in \F} \norm{f}_{L_p(\P)}$ finite
for $p > 2$. In this case, the fixed point condition determining $r$ is $R_n(\F_{c_0r}) \le c_1 r (r /M_{p}(\F))^{p/(p-2)}$ 
\citep[Corollary 3.6]{mendelson2017extending}. 
The approach we use to establish the second bound in \eqref{eq:ratio-process-bounds}
is based on Markov's inequality and holds without uniform boundedness.
However, if it were known that the class $h_{g}(\cdot,\F_r)$ were uniformly bounded or otherwise had well behaved tails,
a sharper concentration inequality like Talagrand's \citep[e.g.,][Theorem 3.3.9]{gine2015mathematical}
could be used to establish bounds that do not depend strongly on the tail probability $\delta$.

In Section~\ref{sec:bounding-I}, we bound the supremum of the mean-zero empirical process
$(\Pn - \P)h_{\riesz[\psi]}f$ indexed by the random set $\F'_{\rho} \subseteq \set{f \in \F : \Pn f^2 \le \rho^2}$. 
Our approach is based on showing that with high probability, $\F'_{\rho}$ is contained in the 
deterministic set $\set {f \in \F : \P f^2 \le 2\rho^2}$,
and involves the use of bounds based on contraction principle arguments 
that do not generalize well to the unbounded case. In the unbounded case,
it is probably more natural to work with the random set $\F'_{\rho}$ directly, 
for example by using symmetrization to introduce Rademacher multipliers 
and analyzing the resulting Rademacher average conditional on $Z_1 \ldots Z_n$ using
bounds on $L_2(\Pn)$ metric entropy \citep[see e.g.,][Theorem 3.5.1]{gine2015mathematical}.

\section{Asymptotics}
\label{sec:asymptotics}

We will now prove our simple asymptotic result, Theorem~\ref{theo:simple}, using Theorem~\ref{theo:simple-rate} for $Q=\Pn$.
Our assumptions that $\F$ is pointwise closed and therefore $L_2(\Pn)$-closed, that $h(Z,f)$ is pointwise bounded, 
and that $f(Z)$ is uniformly bounded justify the application of the latter.
The following lemma will be used to show that our Rademacher complexity fixed points are $o(n^{-1/4})$.
\begin{lemm}
\label{lemma:fixed-point-fourth-root}
Let $\tau_n(r)$ be a sequence of positive functions, each increasing in $r$, and satisfying $\tau_n(s_n) = o(n^{-1/2})$
for all positive sequences $s_n \to 0$. For any $\eta > 0$, there exists a positive sequence $r_n$ satisfying $r_n = o(n^{-1/4})$ and $\tau_n(r_n) \le \eta r_n^2$ for sufficiently large $n$.
\end{lemm}
\begin{proof}
Let $r_n = \sqrt{\tau_n(n^{-1/4})/\eta}$. Then $r_n = o(n^{-1/4})$ and $\tau(r_n) \le \eta r_n^2 = \tau(n^{-1/4})$ for $n$ sufficiently large that $r_n \le n^{-1/4}$.
\end{proof}

\begin{proof}[Proof of Theorem~\ref{theo:simple}]
We will prove asymptotic linearity \eqref{eq:asymptotic-linearity} here, deferring our claims about regularity  
and efficiency to Section~\ref{sec:proof-of-efficiency} below. 
We begin by showing that $\sqrt{n}R_n(\chi(\F) \cap \omega_{\chi}(r_n)B) \to 0$ 
whenever $r_n \to 0$ for $\chi \in \set{ f \to f,\ f \to \riesz[\psi] f,\ f \to h(\cdot, f) }$
and $\omega_{\chi}(r) = \sup_{f \in \F \cap r B} \norm{f}_{L_2(\P)}$. 

Because each set $\chi(\F)$ is Donsker, the
corresponding Rademacher processes are asymptotically equicontinuous \citep[e.g.,][Theorem 14.6]{ledoux1991probability}
in the sense that $\sqrt{n}R_n(\chi(\F) \cap s_n B) \to 0$ whenever $s_n \to 0$.
Thus,  $\sqrt{n}R_n(\chi(\F) \cap \omega_{\chi}(r_n)B) \to 0$ whenever $r_n \to 0$ 
if $\lim_{r \to 0}\omega_{\chi}(r) = 0$. For $\chi(f)=f$, this holds tautologically; for $\chi(f)=h(\cdot,f)$, this is assumed;
and for $\chi(f)=\riesz[\psi]f$, this follows from the uniform boundedness of $\F$ and square integrability of $\riesz[\psi]$
via a truncation argument: if $\P f^2 \le r^2$,
\[ \P \riesz[\psi]^2 f^2 
    = \P \riesz[\psi]^2 1(\riesz[\psi]^2 \le 1/r) f^2 +
      \P \riesz[\psi]^2 1(\riesz[\psi]^2 > 1/r) f^2 
    \le r + \norm{f}_{\infty} \P \riesz[\psi]^2 1(\riesz[\psi]^2 > 1/r), \]
and this goes to zero as $r \to 0$. And this implies that 
$R_n(h_{\riesz[\psi]}(\cdot, \F \cap r_n B)) \to 0$ as $r_n \to 0$,
as $R_n(h_{\riesz[\psi]}(\cdot, \GG)) \le R_n(h(\cdot,\GG)) + R_n(\riesz[\psi]\GG)$ for any set $\GG$.
Thus, on an event of arbitarily high probability,
$r = o(n^{-1/4})$ via Lemma~\ref{lemma:fixed-point-fourth-root} and $\sqrt{n}\phi(s_n) \to 0$ for any $s_n \to 0$. 
The remainder of our proof is based on these two rates.

As a consequence of our assumed tightness and consistency properties \eqref{eq:consistency-properties},
to establish the asymptotic linearity property 
\[ \sqrt{n}(\hpsi_{AML} - \psi(m) - n^{-1}\sum_{i=1}^n \influence_{\tilde\riesz}(Y_i,Z_i)) \to_P 0, \]
it suffices to show that for any $\delta > 0$, the three-term remainder bound \eqref{eq:remainder-rate} goes to zero
for any constant $s_{\F}$ and with any sequence $s_n \to 0$ in place of $s_{L_2(\Pn)}$.
\begin{enumerate}
\item The first term of our bound goes to zero if $\norm{\hriesz - \tilde\riesz}_{L_2(\Pn)}$ does.
This happens because $\sqrt{n}r^2 \to 0$ and $\norm{\tilde\riesz}_{\F} \le (\sqrt{n}/\sigma)\norm{\riesz[\psi]}_{L_2(\Pn)} = O_P(\sqrt{n})$.
The latter bound holds because 
\[ \norm{\riesz[\psi] - \riesz}_{L_2(\Pn)}^2 + (\sigma^2/n)\norm{\riesz}_\F^2 \ \text{ is smaller at is minimizer }\ 
\riesz=\tilde\riesz \ \text{ than at }\ \riesz=0. \]
\item The second term goes to zero because $\sqrt{n}\phi(s_n) \to 0$ when $s_n \to 0$.
\item The third term goes to zero if \smash{$\norm{\hriesz - \riesz[\psi]}_{L_2(\Pn)}$} does.
By the triangle inequality, this happens if both $\norm{\hriesz - \tilde\riesz}_{L_2(\Pn)}$
and $\norm{\tilde\riesz - \riesz[\psi]}_{L_2(\Pn)}$ do. We have established that the first does.
To show that the second does, observe that there is a sequence of approximations $\tilde\riesz_j \in \vspan \F$ 
converging to any element in its closure, and therefore to $\riesz[\psi]$, 
and it has a convergent subsequence $\tilde\riesz_{j_n}$ satisfying $\norm{\tilde\riesz_{j_n}}_{\F}/\sqrt{n} \to 0$.
It follows that $\norm{\tilde\riesz - \riesz[\psi]}_{L_2(\Pn)} \to 0$ on an event of probability $1-\delta$, as 
\begin{align*} 
    \norm{\tilde\riesz - \riesz[\psi]}_{L_2(\Pn)}^2 + (\sigma^2/n)\norm{\tilde\riesz}_\F^2 
&\le \norm{\tilde\riesz_{j_n} - \riesz[\psi]}_{L_2(\Pn)}^2 + (\sigma^2/n)\norm{\tilde\riesz_{j_n}}_\F^2 \\
&\le \delta^{-1}\norm{\tilde\riesz_{j_n} - \riesz[\psi]}_{L_2(\P)}^2 + (\sigma^2/n)\norm{\tilde\riesz_{j_n}}_\F^2 \to 0
\end{align*} 
Our first comparison is via the optimality of $\tilde\riesz$ and our second on Markov's inequality.
\end{enumerate}
This establishes asymptotic linearity in the sense stated above.
The form of asymptotic linearity we want to prove \eqref{eq:asymptotic-linearity} 
differs in that it has $\influence_{\riesz[\psi]}$ in place of $\influence_{\tilde\riesz}$.
By the triangle inequality, these are equivalent if $\sqrt{n}\Pn( \influence_{\riesz[\psi]} - \influence_{\tilde\riesz}) = o_p(1)$.
And as $\Pn( \influence_{\riesz[\psi]} - \influence_{\tilde\riesz}) = \Pn (\riesz[\psi] - \tilde \riesz)\varepsilon_i$ for $\varepsilon_i = Y_i- m(Z_i)$,
via Chebyshev's inequality as in the derivation of our noise term bound in Section~\ref{sec:convergence-of-the-noise-term},
this goes to zero because $\norm{\riesz[\psi] - \tilde \riesz}_{L_2(\Pn)} \to 0$.
\end{proof}

\subsection{Theorem~\ref{theo:simple-hull}}
\label{sec:appendix-flexible}
We turn our focus to Theorem~\ref{theo:simple-hull}, a variant of the theorem proven above
in which $\F$ is defined as the absolutely convex hull of $\set{m_1 \ldots m_{K_n}} - \GG$ 
for a Donsker class $\GG$. Our claim that this theorem justifies the use of $K=o(n^{1/(2+\alpha)})$ candidates
in ideal conditions, for example when $\sup_{n} \norm{\chi(\F_n)}_{\infty} < \infty$ and $\omega_{\chi,\F_n}(r) \lesssim r$ for all $\chi$
and $\norm{\hm - m}_{L_2(\Pn)} =$ \smash{$ O_p(n^{-1/4})$}, follows from a straightforward covering number bound.

For large $K$, when $\log \hat N(\HH,\tau) \le \tau^{-\alpha}$ for $\alpha < 2$,
\[ \hat N(\HH,\log(K+1)^{-1/2}) \le \exp(\log(K+1)^{\alpha/2}) \le \exp((\alpha/2)\log(K+1)) = (K+1)^{\alpha/2}. \]
In the second comparison, we've used the property that for $a=\log(K+1)$ and $b=\alpha/2 < 1$,
$a^b \le ab \iff a^{b-1} \le b$, and $a^{b-1} \to 0$ as $a=\log(K+1) \to \infty$ whereas $b$ remains constant.
Thus, if we could take $a_n=1$ in condition $(2c)$, it would suffice that $(K_n+1)^{1+\alpha/2} = o(n^{1/2})$,
which would imply our claim. Modification for $a_n \to 0$ slowly is straightforward.

We will now prove Theorem~\ref{theo:simple-hull} and a related claim from Remark~\ref{rema:hull-k-constant}.
Throughout, we will write $rB$ and $r\hat B$ for the radius-$r$ balls in $L_2(\P)$ and $L_2(\Pn)$
and $R_n(\HH)$ and \smash{$\hat R_n(\HH)$} for $\E \sup_{h \in \HH}\abs{\Pn \varepsilon_i h(Z_i) }$ 
and $\E_{\varepsilon} \sup_{h \in \HH}\abs{\Pn \varepsilon_i h(Z_i)}$, with the latter
denoting expectation conditional on $Z_1 \ldots Z_n$. 
Here $\varepsilon_1 \ldots \varepsilon_n$ is a sequence of independent Rademacher random variables independent of $Z_1 \ldots Z_n$,
and we will write $G_n$ and $\hat G_n$ for analogs of $R_n$ and $\hat R_n$ in which 
a sequence of standard normals $\xi_1 \ldots \xi_n$ replaces the Rademacher sequence.
The lemmas below, which we will use in our proof, will be proven afterward.

\begin{lemm}
\label{lemma:empirical-vs-pop-local-complexity}
Let $\F \subseteq L_2(\P)$ be star-shaped around zero, with finite $M_p := \sup_{f \in \F} \norm{f}_{L_p(P)}$ for $p \in (2,\infty]$,
let $\omega(r)$ be a non-decreasing function on the positive reals, and for any $r_L \ge 0$, let 
\begin{align*}
&r_{\star} = \inf\set{ r > r_L : R_n(\F \cap \omega'(r) B) \le \eta r^2} && \text{ and } \\
&\hat r_{\star} = \inf \set{ r > r_L : \hat R_n(\F \cap \omega'(r) \hat B) \le \delta \eta r^2/2} && \text{ for } \\
&\omega'(r) = \omega(r) \vee c_0 r^{2/(1+q)}, \ \ q=p/(p-2), \ \ \eta < c_1/M_p^q.
\end{align*} 
On an event of probability $1-\delta - 2\exp(-c_2 n r_\star^{4q/(1+q)}/M_p^{2q})$, $\hat r_{\star} \ge r_{\star}$.
Here $c_0$ is a universal constant and $c_1,c_2$ depend only on $p$.
\end{lemm}

\begin{lemm}
\label{lemma:non-increasing-modulus}
Let $\F$ be a subset of a space with norm $\norm{\cdot}$ that is star-shaped around zero and $\chi$ be a linear map from $\F$ into a space with norm $\norm{\cdot}'$. For the continuity modulus $\omega(r) = \sup_{f \in \F : \norm{f} \le r} \norm{\chi(g)}'$, $\omega(r)/r$ is nonincreasing. 
\end{lemm}

\begin{coro}
\label{cor:empirical-vs-pop-local-complexity}
Let $\F_n \subseteq L_2(\P)$ be a sequence of sets, each star-shaped around zero, 
let $\chi$ be a linear map from $\cup_{n}\F_n \to L_2(\P)$ 
with $\sup_n \sup_{f \in \ff_n} \norm{f}_{L_p(\P)} < \infty$ for $p \in (2,\infty]$,
and let 
\[ \omega_n'(r) = \sup_{f \in \F \cap r B} \norm{\chi(f)}_{L_2(\P)} \vee c_0 r^{(p-2)/(p-1)}
\ \text{ for a universal constant }\ c_0.
 \]
Let $\eta > 0$ be a constant and $r_n$ and $r_n'$ be deterministic sequences with $r_n' \ge r_n$.
If $\hat R_n(\chi(\F_n) \cap \omega'(\hat r) \hat B) \le (\eta/4) \hat r^2$ with $\hat r = o_P(r_n)$,
then $R_n(\chi(\F) \cap \omega'(r)B) \le \eta r^2$ with $r=o(r_n)$ 
and furthermore $R_n(\chi(\F) \cap \omega'(r_n')B) = O_P( \hat R_n(\chi(\F) \cap \omega'(r_n') \hat B))$.
\end{coro}

\begin{lemm}
\label{lemma:gaussian-comparison-convex-hull}
Let $\F$ be the absolutely convex hull of $\set{m_1 \ldots m_{K}} - \GG$. For any $r, s > 0$,
\[ \hat G_n\p{\F \cap r \hat B} \le 2 \hat G_n\p{[\GG - \GG] \cap s \hat B} + cn^{-1/2} s \sqrt{\log(K+1)} +  n^{-1/2} r \sqrt{(K+1)\hat N(\GG, s)}. \]
Here $c$ is a universal constant and $\hat N(\GG, s)$ is the minimal size of a cover of $\GG$ by $\norm{\cdot}_{L_2(\Pn)}$-balls of radius $s$.
\end{lemm}

\begin{proof}[Proof of Theorem~\ref{theo:simple-hull} and Remark~\ref{rema:hull-k-constant}]
As in the proof of Theorem~\ref{theo:simple} above, 
it suffices to show two rate bounds: 
$R_n(\chi(\F) \cap \omega_{\chi}(r)B) \le \eta r^2$ with $r=o(n^{-1/4})$ for arbitarily small $\eta > 0$
and $\sqrt{n}R_n(\chi(\F) \cap \omega_{\chi}(s_n)B) \to 0$ for all $\chi$
and $\norm{\hm - m}_{L_2(\Pn)} = O_p(s_n)$. 
And by Corollary~\ref{cor:empirical-vs-pop-local-complexity} for $r_n = n^{-1/4}$ and $r_n'=n^{-1/4} \vee s_n$,
it suffices that\footnotemark 
\begin{align}
&\hat R_n(\chi(\F) \omega_{\chi}'(\hat r)\hat B) \le (\eta/4) \hat r^2 \quad \text{ with } \quad \hat r = o_p(n^{-1/4}), \label{eq:hull-fixed-pt}\\
&\sqrt{n}\hat R_n(\chi(\F) \cap \omega_{\chi}'(s_n \vee n^{-1/4})\hat B) = o_P(1). \label{eq:hull-localized}
\end{align}
\footnotetext{In the statement of Theorem~\ref{theo:simple-hull}, we use a simplified definition of $\omega_{\chi'}$ in which the universal constant
$c_0$ from \eqref{cor:empirical-vs-pop-local-complexity} is taken to be one. This does not affect our proof,
which depends on the order of $\omega_{\chi'}(r)$ but not constant factors.}

By a contraction principle for Rademacher averages \citep[Lemma 4.5]{ledoux1991probability}, 
we can bound each Rademacher complexity by a multiple of the analogous Gaussian complexity:
\[ \sqrt{n}\hat R_n\p{\chi(\F) \cap \omega'_{\chi}(r) \hat B} \le (\pi/2)^{1/2} \sqrt{n}\hat G_n\p{\chi(\F) \cap \omega'_{\chi}(r) \hat B}. \] 
And as  $\chi(\F)$ is the absolutely convex hull of $\set{\chi(m_1) \ldots \chi(m_K)} - \chi(\GG)$,
by Lemma~\ref{lemma:gaussian-comparison-convex-hull}, 
\begin{align*} \sqrt{n}\hat G_n\p{\chi(\F) \cap \omega_{\chi}'(r) \hat B)} 
&\le 2\sqrt{n}\hat G_n\p{[\chi(\GG)-\chi(\GG)] \cap s \hat B} 
+ c s \sqrt{\log(K+1)}  \\
&+ \omega_{\chi}'(r) \sqrt{(K+1)\hat N(\chi(\GG), s)} \quad \text{ for any } s. 
\end{align*}
 The first term in this bound is $o_P(1)$ when $s \to 0$ as $n \to \infty$, 
as (i) $\HH := \chi(\GG) - \chi(\GG)$ is Donsker when $\chi(\GG)$ is Donsker with $\sup_{f \in \chi(\GG)}\abs{\P f} < \infty$ \citep[Example 2.10.7]{vandervaart-wellner1996:weak-convergence},
(ii) $\sqrt{n}G_n(\HH \cap t B) \to 0$ 
when $t \to 0$  for Donsker $\HH$ \citep[e.g.,][Theorem 14.6]{ledoux1991probability},
(iii) $\hat G_n(\HH \cap t B) = O_P(G_n(\HH \cap t B))$ by Markov's inequality,
(iv) for any $s \to 0$, there exists $t \to 0$ for which
    $\HH \cap s \hat B \subseteq \HH \cap t B$ with arbitrarily high probability. This last property
holds because $\HH^2 = \set{h^2 : h \in \HH}$ is Glivenko-Cantelli when $\HH$ is Donsker with $\sup_{h\in\HH}\abs{\P h} < \infty$ \citep[Lemma 2.10.14]{vandervaart-wellner1996:weak-convergence}, so $\sup_{h \in \HH \cap s \hat B} \P h^2 \le \sup_{h \in \HH \cap s \hat B} \Pn h^2 + \sup_{h \in \HH}\abs{(\P - \Pn)h^2} \le s^2 + o_P(1)$.
Thus, taking $s=a_n/\sqrt{\log(K+1)}$ with $a_n \to 0$,
\[  \sqrt{n}\hat R_n\p{\chi(\F) \cap \omega'_{\chi}(r) B} = o_P(1) + O_P\p{\omega_{\chi}'(r) \sqrt{(K+1)\hat N(\chi(\GG), a_n/\sqrt{\log(K+1)})}}. \]
We will use this bound to check the aforementioned sufficient conditions. 

For condition \eqref{eq:hull-fixed-pt}, it suffices
that this bound is less than $\sqrt{n}(\eta/4) r^2$ for $r=o_P(n^{-1/4})$. This happens if each term satisfies this condition individually,
and the leading $o_P(1)$ term does, so this condition reduces to
\[ \sqrt{(K+1)\hat N(\chi(\GG), a_n/\sqrt{\log(K+1)})} \le \sqrt{n}(\eta/4) r^2 /\omega_{\chi}'(r)\ \text{ for }\ r=n^{-1/4}a_n',\ a_n' = o_P(1). \]
And as $\omega_{\chi}'$ is increasing, it suffices that
\[ \sqrt{(K+1)\hat N(\chi(\GG), a_n/\sqrt{\log(K+1)})} \le \sqrt{n}(\eta/4) r^2 / \omega_{\chi}'(n^{-1/4}) 
							= (\eta/4)(a_n')^2 / \omega_{\chi}'(n^{-1/4}) \]
or equivalently that
\[ (K+1)\hat N(\chi(\GG), a_n/\sqrt{\log(K+1)}) = o_P(1/\omega_{\chi}'(n^{-1/4})^2). \]
For condition \eqref{eq:hull-localized}, taking $r=n^{-1/4} \vee s_n$  in our bound above,
it suffices that a variant of this condition holds with $\omega_{\chi}'(n^{-1/4} \vee s_n)$ 
in place of $\omega_{\chi}'(n^{-1/4})$, so our assumption involving \smash{$n^{-1/4} \vee s_n$} implies both 
conditions \eqref{eq:hull-fixed-pt} and \eqref{eq:hull-localized}. 
This concludes our proof of Theorem~\ref{theo:simple-hull}.

We will now prove our claim from Remark~\ref{rema:hull-k-constant}. For $K=O(1)$, 
the condition above reduces to \smash{$\omega_{\chi}'(s_n \vee n^{-1/4}) \to 0$}, as by taking $a_n \to 0$ slowly we can take \smash{$\hat N(\chi(\GG),a_n) \to \infty$} arbitrarily slowly ---
in particular, slowly enough that \smash{$\hat N(\chi(\GG),a_n) = o_P(b_n)$} for any sequence $b_n \to \infty$.
That the growth of \smash{$\hat N(\chi(\GG), a_n)$} is bounded as a function of $a_n$ 
in this sense is implied by Sudakov minoration \citep[e.g.,][Theorem 3.18]{ledoux1991probability}:
\[ \sqrt{\log \hat N(\chi(\GG), s)} = O(s^{-1}\hat G_n(\chi(\GG)) = O_P(s^{-1}n^{-1/2}), \]
where in the second comparison we've used the tightness of $\sqrt{n}G_n(\HH)$ for Donsker $\HH$.
And it suffices to assume this condition
on $\omega_{\chi}$ for $\chi(f)=h(\cdot,f)$ only ---
this is clearly satisfied for $\chi(f)=f$, and it was shown that it is satisfied for $\chi(f)=\riesz[\psi]f$
when $\sup_n\sup_{f \in \ff_n}\norm{f}_{\infty} < \infty$ in the proof of Theorem~\ref{theo:simple}.
\end{proof}

\begin{proof}[Proof of Lemma~\ref{lemma:empirical-vs-pop-local-complexity}]
Our proof is based on that of Theorem 4.1 in \citet{bartlett2005local}.
By \citet[Corollary 3.6]{mendelson2017extending}, if $c_1^{-1} M_p^q R_n(\F \cap c_0 s B) \le s^{1+q}$,
then with probability $1-2\exp(-c_2n (s/M_p)^{2q})$,
$\Pn f^2 \ge (1/4)\P f^2$ for all $f \in \F$ with $\P f^2 \ge s^2$. We can assume $c_0 \ge 1$ here,
as if it is true for $c_0 < 1$ it remains true for $c_0=1$.
Taking $s=r^{2/(1+q)}$, if $c_1^{-1} M_p^q R_n(\F \cap \omega''(r) B) \le r^2$
for $\omega''(r)=c_0 r^{2/(1+q)}$, then on an event of probability $1-2\exp(-c_2 n  r^{4q/(1+q)}/M_p^{2q})$,
$\Pn f^2 \ge (1/4)\P f^2$ for all $f \in \F$ with $\P f^2 \ge r^{4/(1+q)}$. And on this event, 
$\F \cap 2s\hat B \supseteq \F \cap sB$ for all $s \ge r^{2/(1+q)}$. 
Furthermore, $R_n(\F \cap sB) < \delta^{-1}\hat R_n(\F \cap s B)$ with probability $1-\delta$
by Markov's inequality. And by the union bound, 
both hold on an event of probability $1-\delta - 2\exp(-c_2 n r^{4q/(1+q)}/M_p^{2q})$.
For the remainder of our argument, we work on this event, and let 
$\psi_0(r) := c_1^{-1} M_p^q R_n(\F \cap \omega'(r) B)$ for $\omega'(r) = \omega(r) \vee \omega''(r)$.
If $r^2 \ge \psi_0(r)$, because $\omega''(r) \ge r^{2/(1+q)}$,
\begin{align*} 
\psi(r) &:= \eta^{-1} R_n(\F \cap \omega''(r)B) \\
&\le  (\delta \eta)^{-1} \hat R_n(\F \cap \omega''(r) B) \\
&\le  (\delta \eta)^{-1} \hat R_n(\F \cap 2 \omega''(r) \hat B)  \\
&\le  \hat\psi(r) := 2(\delta \eta)^{-1} \hat R_n(\F \cap \omega''(r) \hat B).
\end{align*}
It follows that the bound $\psi(r) \le \hat \psi(r)$ holds for any $r$ satisfying $\psi(r) \le r^2$, 
as $\psi_0(r) \le \kappa \psi(r)$ for $\kappa = (c_1^{-1} M_p^q)/\eta^{-1} < 1$. Now suppose that $\psi(r)/r$ is non-increasing --- we will show this below.
This means that this bound holds for any $r > r_{\star}$, as if $\psi(r)/r \le r$ then $\psi(r_+)/r_+ \le r_+$
for all $r_+ \ge r$. Furthermore, our bound above holds for 
$r=r_-$ slightly smaller than $r_{\star}$, as for $r=r_+ := r_-/\sqrt{\kappa} $ slightly larger,
$\psi_0(r_-) \le \psi_0(r_+) \le \kappa \psi(r_+) \le \kappa r_+^2 = r_-^2$.

It follows that the bound $\psi(r_-) \le \hat \psi(r_-)$ holds for $r_-$ smaller than, but sufficiently close to, $r_{\star}$.
If $\hat r_{\star}$ were less than $r_{\star}$, then this bound would hold for some $r_- \in (\hat r_{\star}, r_{\star})$,
and implying that $\psi(r_-) \le \hat \psi(r_-) \le r_-^2$. As by definition $r_{\star}$ is a lower bound on 
the set $\set{r > 0 : \psi(r) \le r^2}$, our premise cannot be true --- it must be the case that $\hat r_{\star} \ge r_{\star}$.

We conclude by showing that $\psi(r)/r$ is non-increasing.
Because $\psi(r)/r = \eta^{-1} R_n(r^{-1}\F \cap [\omega''(r)/r]B)$
and $R_n$ is increasing in the order of inclusion in the sense that $A \subseteq B \implies R_n(A) \le R_n(B)$,
it suffices to show that $r_+^{-1}\F \cap [\omega''(r_+)/r_+] B \subseteq r^{-1}\F \cap [\omega''(r)/r] B$
when $r_+ \ge r$. This holds because $\F$ and $B$ are star-shaped around zero
and $r^{-1}$ and $\omega(r)/r$ are non-increasing.
\end{proof}

\begin{proof}[Proof of Lemma~\ref{lemma:non-increasing-modulus}]
If $s \ge r$ and a sequence $f_j \in \set{\F : \norm{f} \le s}$ satisfy $\lim_{j}\norm{\chi(f_j)}' = \omega(s)$,
then $(r/s)f_j \in \set{\F : \norm{f} \le r }$, so $\omega(r) \ge \lim_{j} \norm{\chi((r/s)f_j)}' = (r/s)\omega(s)$ and equivalently $\omega(r)/r \ge \omega(s)/s$.
\end{proof}

\begin{proof}[Proof of Corollary~\ref{cor:empirical-vs-pop-local-complexity}]
Define $r_{\star}$ and $\hat r_{\star}$ as in Lemma~\ref{lemma:empirical-vs-pop-local-complexity} with $\F = \chi(\F_n)$,
$\omega(r)=\sup_{f \in \F_n \cap r B} \norm{\chi(f)}_{L_2(\P)}$, $\delta=1/2$, and $r_L$ satisfying $n^{(1+q)/4q}r_L \to 0$
for $q=p/(p-2)$. The definition of $\omega'$ that we use here agrees with that of Lemma~\ref{lemma:empirical-vs-pop-local-complexity},
as 
\[ \frac{2}{1 + p/(p-2)} = \frac{2}{2(p-1)/(p-2)} = \frac{p-2}{p-1}. \] 
And as the restriction $r_{\star} \ge r_L$ implies that $n r_\star^{4q/(1+q)} \to \infty$,
$P(r_\star \ge \hat r_{\star}) \to \delta$ by Lemma~\ref{lemma:empirical-vs-pop-local-complexity}.
Furthermore, as $\hat r_{\star} = o_P(r_n)$ by assumption,
$P(\hat r_{\star} \ge \epsilon r_n) \to 0$ for any $\epsilon > 0$, and by the union bound
$P(r_\star \ge \epsilon r_n) \to \delta$. As $r_\star \le \epsilon r_n$ is a deterministic comparison,
and it is true with probability tending to the nonzero limit $1 - \delta=1/2$, 
it follows that it is true deterministically: $r_\star = o(r_n)$. 

Now consider the latter claim. By Markov's inequality, $R_n(\chi(\F_n) \cap \omega_n'(s_n)B) = O_P( \hat R_n(\chi(\F_n) \cap \omega_n'(s_n) B))$,
and using a property used in the proof of Lemma~\ref{lemma:empirical-vs-pop-local-complexity}
--- the property that $\F \cap 2s\hat B \supseteq \F \cap sB$ for all $s > \omega_n'(r_{\star})$ 
with probability tending to one --- it follows that so long as $r_n' > r_{\star}$,
$R_n(\chi(\F_n) \cap \omega_n'(r_n')B) = O_P(R_n(\chi(\F_n) \cap 2\omega_n'(r_n')\hat B))$.
This implies our claim, as $r_n' \ge r_n > r_{\star}$ with probability tending to one, and 
$\chi(\F_n) \cap 2\omega_n'(r_n')\hat B \subseteq  2[ \chi(\F_n) \cap \omega_n'(r_n')\hat B)]$.
\end{proof}

\begin{proof}[Proof of Lemma~\ref{lemma:gaussian-comparison-convex-hull}]
Observe that $\F$ is contained in the Minkowski sum $\conv(\F_0)-\conv(\F_0)$
where $\F_0 := \MM - \GG$ for $\MM = \set{0, m_1, m_2, \ldots, m_K}$. Thus, for any $s$,
\begin{align*} 
\hat G_n(\F \cap r \hat B) 
&\le \hat G_n\p{(\conv(\F_0) - \conv(\F_0)) \cap r \hat B} \\
&\le 2\hat G_n\p{(\F_0 - \F_0] \cap s \hat B} + r \sqrt{\hat N(\F_0, s)/n} \\
&\le 2\hat G_n\p{(\F_0 - \F_0) \cap s \hat B} + r \sqrt{(K+1)\hat N(\GG, s)/n}. 
\end{align*}
The first comparison holds because of this containment, the second via a bound of \citet[Theorem 1]{bousquet2002some}
relating the moduli of continuity of the isonormal gaussian processes $f \to n^{1/2}\Pn \xi_i f(Z_i)$
indexed by a set $\F_0$ and its convex hull, and the third because given any $s$-cover $\GG^s$ 
of $\GG$, $\MM -  \GG^s$ is an $s$-cover of $\F_0$.

To bound the first term here, observe that every function $f \in \F_0-\F_0$
can be written as $m - g$ for $m \in \mathcal M - \mathcal M$ and $g \in \GG-\GG$.
Letting $\Pi$ be the $L_2(\Pn)$-orthogonal projection onto the convex set $\GG-\GG$, we can write this as
$(m - \Pi m) + (\Pi m - g)$, and by the Hilbert space projection theorem \citep[Proposition 1.37]{peypouquet2015convex},
$\Pn (m-\Pi m) (g-\Pi m) \le 0$ for all $g \in \GG-\GG$, so
\begin{align*} 
\norm{f}_{L_2(\Pn)}^2 &= \norm{m - \Pi m}_{L_2(\Pn)}^2 + \norm{\Pi m - g}_{L_2(\Pn)}^2 + 2 \Pn (m-\Pi m)(\Pi m - g) \\
		    &\ge \norm{m - \Pi m}_{L_2(\Pn)}^2 + \norm{\Pi m - g}_{L_2(\Pn)}^2. 
\end{align*}
Thus, every $f \in [\F_0 - \F_0] \cap s\hat B$ can be written as a sum $(m-\Pi m) + (\Pi m - g)$ 
where $\Pi m - g \subseteq [2(\GG-\GG)] \cap s\hat B$ and 
$m-\Pi m \in \MM' \cap s \hat B$ for $\MM' := \set{m_i - m_j - \Pi (\hat m_i - \hat m_j) : i,j \in 0 \ldots K}$ with $m_0=0$.
Thus, for some universal constant $c$, 
\begin{align*}
\hat G_n\p{(\F_0 - \F_0) \cap s \hat B} 
&\le \hat G_n\p{\MM' \cap s B} + \hat G_n\p{[2(\GG-\GG)] \cap s \hat B} \\
&\le c n^{-1/2} s \sqrt{\log(K+1)} +  2 \hat G_n\p{(\GG-\GG) \cap s \hat B}. 
\end{align*}
Here we bound $\hat G_n([2(\GG-\GG)] \cap s \hat B)$ using the inclusion $[2(\GG-\GG)] \cap s \hat B \subseteq 2[(\GG - \GG) \cap s \hat B]$
and $\hat G_n(\MM' \cap s \hat B)$ using the finite class bound 
$\hat G_n(\set{f_1 \ldots f_k}) \le cn^{-1/2}\sup_{j}\norm{f_j}_{L_2(\Pn)} \sqrt{\log(k)}$ \citep[e.g.,][Exercise 7.5.10]{vershynin2018high}.
\end{proof}

\subsection{Regularity and Efficiency}
\label{sec:proof-of-efficiency} 

In this section, we will prove the claims about regularity and efficiency in Theorem~\ref{theo:simple}. We express our estimand
as a functional $\chi(P)$ of the distribution of the observed data, defined by $\chi(P)=\psi_P(m_P)$ for $\psi_P(m) = \E_P h(Z,m)$ and $m_P(z)=E_P[Y \mid Z=z]$.
The first step of our proof is characterizing the tangent space at $\P$ in our model. 
Having done this, we calculate the derivative $\dot\chi_P$ of $\chi$ at $\P$ on this tangent space.
An estimator $\hat \chi$ for a differentiable functional $\chi(\P)$ with the asymptotic characterization 
$\hat \chi - \chi(\P) = \Pn \influence(Y,Z) + o_P(1/\sqrt{n})$ 
is regular iff $\E_P \influence(Y,Z)g(Y,Z) = \dot\chi_P(g)$ for all scores $g$ in the tangent space and
asymptotically efficient iff it is regular and $\influence(y,z)$ is in the closure of the tangent space \citep[Section 1.2 and Example 4.6]{van2000semiparametric}.

\subsubsection{The tangent space}
We will show that the tangent space at $\P_0=\P$ to the set of all submodels $\P_t$ 
for which the regression functions satisfy $m_{\P_t} \in \mm$
and $\norm{m_{\P_t} - m_{\P}}_{L_2(\P)} \to 0$ as $t \to 0$ and the squares of $\epsilon_t=Y_t-m_{P_t}(Z_t)$ for $Y_t,Z_t \sim \P_t$ are uniformly integrable is
\[ \tangent = \set*{ a(z) + b(y,z) \in L_2(\P) : \E_P[a(Z)]=0, \E_P[b(Y,Z) \mid Z]=0, \E_P[Y b(Y,Z) \mid Z=z] \in \mm}. \]
To show that $\tangent$ contains the tangent space, we will show that
the score $g$ of every such submodel $P_t$ is in $\tangent$.
To show that $\tangent$ is contained in the tangent space, we construct such a submodel
for each score $g \in \tangent$. 

\paragraph*{Containment of the tangent space in $\tangent$}
We'll begin with a non-rigorous argument. 
Consider a submodel $P_t$ with factored density $p_t(y,z)=p_t(z) p_t(y \mid z)$ with respect to a product measure $\mu_y \times \mu_z$
and suppose that it is differentiable both pointwise and in quadratic mean, so its score
function is the derivative at $t=0$ of the log likelihood $\log p_t (z) + \log p_t(y \mid z)$.
Call the first term $a(z)$ and the second term $b(y,z)$ --- it is well known that given the score 
$g(y,z)=a(z)+b(y,z)$, we can uniquely recover $a(z)=\E_P[g(Y,Z) \mid Z=z]$ and $b(y,z)=g(y,z)-a(z)$.
A submodel must satisfy $\int y p_t(y \mid z)d\mu_y =\E_{P_t}[Y \mid Z=z] = m_{P_t}(z)$ for $m_{P_t} \in \mm$, 
and assuming we can interchange differentiation and integration,
this implies  
\begin{align*} 
\lim_{t=0}t^{-1}(m_{P_t}(z) - m_P(z))
&= \int y \frac{\partial}{\partial t}\mid_{t=0} p_t(y \mid z) d\mu_y \\
&= \int y \frac{\partial}{\partial t}\mid_{t=0} \log p_t(y \mid z) p_0(y \mid z) d\mu_y \\
&= \E_{P}[Y b(Y,Z) \mid Z=z]. 
\end{align*}
And as $m_{P_t} - m_P \in \mm$ for all $t$, this implies that $\E_{P}[Y b(Y,Z) \mid Z=z] \in \mm$,
so $g \in \tangent$.

To prove this rigorously, we must show that that for any quadratic-mean differentiable submodel $P_t$,
its score $g$ is in $\tangent$. To do this, we begin by simplifying the condition $\E_P[Y b(Y,Z) \mid Z=z] \in \mm$ 
characterizing $\tangent$. This condition is equivalent to the condition $\E_P[f(Z) \E_P[Y b(Y,Z) \mid Z]]=0$
for all $f \in \mm_{\perp}$, the $L_2(\P)$-orthogonal complement of $\mm$. Furthermore, for all bounded $f$,
this is equivalent to the condition $\E_P[f(Z) Y  b(Y,Z)]=0$. Now suppose that this condition holds for all bounded $f$
and recall that we've assumed that $\mm_{\perp}$ has a $\norm{\cdot}_{L_2(P)}$-dense subset of bounded functions.
Each unbounded $f \in \mm_{\perp}$ is the limit of a sequence $f_j \in \mm_{\perp}$ of bounded functions satisfying
$\E_P[f_j(Z) \E_P[Y b(Y,Z) \mid Z]]=0$, and by the Cauchy-Schwarz inequality, $\E_P[(f(Z)-f_j(Z)) \E_P[Y b(Y,Z) \mid Z]] \to 0$ 
and therefore $\E_P[f(Z) \E_P[Y b(Y,Z) \mid Z]] = 0$. Thus, it suffices to show that $\E_P[f(Z) Y b(Y,Z)]=0$ 
for all bounded $f \in \mm_{\perp}$. We will use this to formalize the argument above.

Let $P_t$ be any one-dimensional parametric submodel with score $g$. For a sequence $t_j \to 0$, let $p_t$ and $p$ be densities of $P_{t}$ and $P$ 
with respect to a dominating probability measure $\mu$, so $t^{-1}(\sqrt{p_t}-\sqrt{p}) \to (1/2)g\sqrt{p}$ in $L_2(\mu)$.
Letting $h_t = f(z)(y-m_{P_t}(z))$, 
\begin{equation}
\label{eq:bounded-score-argument}
\begin{aligned} 
&t^{-1}(\E_{P_{t}}[h_{t}(Y,Z)] - \E_{P}[h_{t}(Y,Z)]) - \E_{P}[h_0(Y,Z)g(Y,Z)] \\
&= \int h_t \cdot t^{-1}(p_t - p) d\mu - \int h_0 g p d\mu \\
&= \int h_t \cdot [t^{-1}(\sqrt{p_t} - \sqrt{p})(\sqrt{p_t} + \sqrt{p}) - gp]d\mu + \int (h_t-h_0) g p d\mu.
\end{aligned}
\end{equation}
If this difference goes to zero for every $f \in \mm_{\perp}$, 
this would imply that $g \in \tangent$,
as $\E_{P_{t}}[h_{t}(Y,Z)] - \E_{P}[h_{t}(Y,Z)]=0$ for all $t$ and $\E_{P}[h_0(Y,Z)g(Y,Z)]=\E_P[f(Z)Yb(Y,Z)]$.
To see this, observe that for every $f \in \mm_{\perp}$,
$\E_{P_t}[f(Z)(Y-m_{P_t}(Z))] = 0$ because $\E_{P_t}[Y-m_{P_t}(Z)\mid Z]=0$,
$\E_{P}[f(Z)(Y-m_{P_t}(Z))] = \E_{P}[f(Z)(m_P-m_{P_t})(Z)] = 0$ 
because $f \in \mm_{\perp}$ and $m_P - m_{P_t} \in \mm$ for any submodel,
and $\E_P[f(Z)(Y-m_P(Z))a(Z)] = 0$ and $\E_P[f(Z)m_P(Z)b(Y,Z)]=0$
because $\E_P[Y-m_{P}(Z)\mid Z]=0$ and $\E_P[b(Y,Z) \mid Z]=0$.

If $h_t$ were bounded, this difference would go to zero.
Because $t^{-1}(\sqrt{p_t} - \sqrt{p}) \to (1/2)g\sqrt{p}$ in $L_2(\mu)$,
$\sqrt{p_t} + \sqrt{p} \to 2\sqrt{p}$ in $L_2(\mu)$ and it follows that $t^{-1}(\sqrt{p_t} - \sqrt{p})(\sqrt{p_t} + \sqrt{p}) \to (1/2)g\sqrt{p} \cdot 2\sqrt{p} = gp$ in $L_1(\mu)$.
Thus, by H\"older's inequality, the first term in the last line of \eqref{eq:bounded-score-argument} would go to zero. And because 
$h_t-h_0 = f(z)(m_P - m_{P_t})$ converges to zero in $L_2(\P)$, the Cauchy-Schwarz bound on the second term converges to zero. 
We conclude by using a truncation argument to show that this happens though $h_t$ is not, in general, bounded.

Let $h_t^K = h_t 1(\abs{y - m_{P_t}}  \le K)$, so $\norm{h_t^K}_{\infty} \le K\norm{f}_{\infty}$.  
\begin{equation}
\label{eq:bounded-score-argument-trunc}
\begin{aligned} 
&t^{-1}(\E_{P_{t}}[h_{t}(Y,Z)] - \E_{P}[h_{t}(Y,Z)]) - \E_{P}[h_0(Y,Z)g(Y,Z)] \\
&= \int (h_t - h_t^K) \cdot t^{-1}(p_t - p) d\mu\ +\  \int h_t^K \cdot [t^{-1}(\sqrt{p_t} - \sqrt{p})(\sqrt{p_t} + \sqrt{p}) - gp]d\mu \\
&+ \int (h_t^K-h_0^K) g p d\mu \ +\  \int (h_0^K - h_0) g p d\mu.
\end{aligned}
\end{equation}
Calling these terms $\delta^k_{t,K}$ for $j \in 1\ldots 4$,
we will show that $\lim_{K \to \infty}\lim_{t \to 0}\delta^j_{t_K}=0$.
By H\"older's inequality as described above, for any finite $K$, the second term goes to zero as $t \to 0$.
And as $\abs{h_t^K - h_0^K} \le \abs{h_t - h_0}$, it follows that the Cauchy-Schwarz bound on the third term
is smaller than $\norm{h_t - h_0}_{L_2(\P)}\norm{g}_{L_2(\P)}$, which goes to zero as $t \to 0$. 
This leaves the first and fourth terms. The Cauchy-Schwarz bound on the fourth goes to zero as $K \to \infty$ if $\norm{h_0^K - h_0}_{L_2(\P)} \to 0$,
and we will now show the relevant limit of the first term is zero if, in addition, \smash{$\lim_{K \to\infty}\lim_{t \to 0}\norm{h_t^K - h_t}_{L_2(\P_t)} = 0$.}
We use the following decomposition.
\begin{align*}
\int (h_t-h_t^K) \cdot t^{-1}(p_t - p) d\mu 
&= \int (h_t-h_t^K)(\sqrt{p_t} + \sqrt{p}) \cdot t^{-1}(\sqrt{p_t} - \sqrt{p} - (t/2)g\sqrt{p}) d\mu \\
&+ \int (h_t-h_t^K)(\sqrt{p_t} + \sqrt{p}) \cdot (1/2)g\sqrt{p} d\mu \\
&\le \norm{(h_t-h_t^K)(\sqrt{p_t} + \sqrt{p})}_{L_2(\mu)} \cdot t^{-1}\norm{\sqrt{p_t} - \sqrt{p} - (t/2)g\sqrt{p}}_{L_2(\mu)} \\
& +  \norm{(h_t-h_t^K)(\sqrt{p_t} + \sqrt{p})}_{L_2(\mu)} \cdot (1/2)\norm{g\sqrt{p}}_{L_2(\mu)}. 
\end{align*}
As $t^{-1}\norm{\sqrt{p_t} - \sqrt{p} - (t/2)g\sqrt{p}}_{L_2(\mu)} \to 0$ and $\norm{g\sqrt{p}}_{L_2(\mu)} = \norm{g}_{L_2(\P)} < \infty$,
this is bounded by a constant (in $t$) multiple of $\norm{(h_t-h_t^K)(\sqrt{p_t} + \sqrt{p})}_{L_2(\mu)}$. Furthermore, 
because $(\sqrt{p_t} + \sqrt{p_t})^2 \le 2p_t + 2p$, this is bounded by a constant multiple of $\norm{h_t-h_t^K}_{L_2(\P_t)} + \norm{h_t - h_t^K}_{L_2(\P)}$.
And the relevant limit of this sum goes to zero if 
\begin{equation}
\label{eq:two-limits}
 \lim_{K \to \infty}\lim_{t \to 0}\norm{h_t-h_t^K}_{L_2(\P_t)}=0 \ \text{ and }\ \lim_{K \to \infty}\norm{h_0 - h_0^K}_{L_2(\P)}=0,
\end{equation}
as the second and third terms in the following bound go to zero as $t \to 0$ for any $K$.
\begin{align*} 
\norm{h_t - h_t^K}_{L_2(\P)}
&=\norm{ (h_0 - h_0^K) + (h_t - h_0) + (h_t^K-h_0^K) }_{L_2(\P)} \\
&\le \norm{ h_0 - h_0^K}_{L_2(\P)} + \norm{h_t - h_0}_{L_2(\P)} + \norm{h_t^K-h_0^K}_{L_2(\P)}. 
\end{align*}
All that is left is to show that \eqref{eq:two-limits} holds. It suffices to consider the case of $h_t = y-m_{P_t}$, as
\[ \norm{h_t - h_t^K}_{L_2(\P_t)} = \norm{f \cdot (y-m_{P_t}) 1(\abs{y - m_{P_t}}  > K)}_{L_2(\P_t)} \le \norm{f}_{\infty} \norm{(y-m_{P_t}) 1(\abs{y - m_{P_t}}  > K)}_{L_2(\P_t)}. \]
And for this $h_t$, \eqref{eq:two-limits} is implied by the uniform integrability property we impose on our paths.

\paragraph*{Containment of $\tangent$ in the tangent space}

We will show that each element $a(z)+b(y,z)$ of $\tangent$ 
is the score of a one-dimensional parametric submodel $P_t$
with $E_{P_t}[Y \mid Z=z] = m_P(z) + t m'(z)$ for $m'(z) = \E[ Y b(Y,Z) \mid Z=z]$.
To do this, we use a tilting construction \citep[see e.g.,][Section 4.5]{tsiatis2007semiparametric}. 
Factor $P$ into the product of a regular conditional distribution $P(\cdot \mid Z=z)$ on $Y$ and a marginal $P^M$ on $Z$.
We define this submodel $P_t$ by  
choosing a nonnegative continuously differentiable function $k(x)$ satisfying $k(0)=k'(0)=1$ with $k,k'$ bounded in a neighborhood of zero, 
e.g., $k(x)=2(1+e^{-2x})^{-1}$ \citep[Example 1.12]{van2000semiparametric},
and taking
\begin{equation}
\label{eq:submodels}
\begin{aligned}
&d\P_t(y \mid Z=z)/d\P(y \mid Z=z) = \frac{k(c_t(z)y + tb(y,z))}{\E_P[ k(c_t(Z)Y + tb(Y,Z)) \mid Z=z]}, \\
&d\P_t^M(z) / d\P^M(z) = \frac{k(ta(z))}{\E_P k(ta(Z))}, 
\end{aligned}
\end{equation}
for $c_t$ satisfying 
\[ \E_{P_t}[Y \mid Z=z] = \frac{\E_P[Y k(c_t(Z)Y + tb(Y,Z)) \mid Z=z]}{\E_P[k(c_t(Z)Y + tb(Y,Z)) \mid Z=z]} = m_P(z) + tm'(z). \] 
If we take $c_t(0)=0$, this condition is satisfied for $t=0$, and we will use the implicit function theorem 
to characterize $c_t$ for which it holds on a neighborhood of zero. This requires that the function 
\[ f(t,c) = \frac{\E_P[Y k(cY + tb(Y,Z)) \mid Z=z]}{\E_P[k(c Y + tb(Y,Z)) \mid Z=z]} - \sqb{m_P(z) + tm'(z)} \]
be continuously differentiable and that its partial derivative with respect to $c$ be nonzero at $(c,t)=0$,
and it implies that the solution $c_t$ is continuously differentiable with 
$c_t'(z) = -(\partial f/\partial t)(t, c_t(z)) / (\partial f/\partial c)(t, c_t(z))$.
By the sum and quotient rules of calculus, $f$ is continuously differentiable if the numerator and denominator of
the first term of $f$ are, and this reduces to continuous differentiability of the integrands in the numerator and denominator, 
as we can interchange integration and differentiation because the derivatives of the integrand are dominated.
In particular, because $k'$ is bounded, the partials with respect to $c$ and $t$ 
of the integrand in the numerator are dominated by multiples of $Y^2$ and $Yb(Y,Z)$ respectively;
in the denominator the same goes for $\abs{Y}$ and $\abs{b(Y,Z)}$. Thus, we calculate
\begin{align*} 
(\partial f/\partial t)(t,0) \mid_{t=0} &= \E_P[Yb(Y,Z) \mid Z=z] - \E_P[Y \mid Z=z]\E_P[b(Y,Z) \mid Z=z] - m'(z) \\
                                        &= \E_P[Yb(Y,Z) \mid Z=z] - m'(z) = 0, \\
(\partial f/\partial c)(0,c) \mid_{c=0} &= \E_P[Y^2 \mid Z=z] - (\E_P[Y \mid Z=z])^2 \\
                                        &= \var_P[Y \mid Z=z], \\
c_0'(z) &= -0 \ /\ \var_P[Y \mid Z=z] = 0.
\end{align*}
We will check that this yields a valid submodel. By construction, our densities are nonnegative and integrate to one
and $m_{P_t}(z) = m(z) + tm'(z)$ is in $\mm$ for all $t$ and satisfies $\norm{m_{P_t} - m}_{L_2(\P)} = t\norm{m'}_{L_2(\P)} \to 0$ as $t \to 0$.
The remaining condition is the uniform integrability property $\lim_{K \to \infty}\E_{P_t}(Y-m_{P_t})^2 1((Y-m_{P_t})^2 \ge K)=0$ for $t$ in a neighborhood of zero.
Because the denominators in \eqref{eq:submodels} are near one in a neighborhood of $0$, it suffices that
\[ \lim_{K \to \infty}\E_P[ k(c_t(Z) + t b(Y,Z))k(t(a(Z))) (Y - m(Z)-tm'(Z))^2 1((Y-m(Z) - tm'(Z))^2 \ge K) = 0. \]
Furthermore, because $k$ is bounded this is equivalent to the uniform integrability of $(Y - m(Z)-tm'(Z))^2$ under $\P$,
and this holds because for $t \le 1$, these are dominated by the integrable quantity $2(Y-m(Z))^2 + 2 m'(Z)^2$.

Finally, we check that this submodel is differentiable in quadratic mean with the intended score  $a(Z) + b(Y,Z)$.
By design, this intended score is the derivative at zero of the log of the density $p_t = d\P_t/d\P$.
\begin{align*}
&\frac{d}{dt}\mid_{t=0} \log p_t =\frac{d}{dt}\mid_{t=0} \log \p{ d\P_t^M / d\P^M } +\frac{d}{dt}\mid_{t=0} \log\p{ d\P_t(\cdot \mid Z)/d\P(\cdot \mid Z) } \\ 
&= \frac{d}{dt}\mid_{t=0} \set*{\log k(ta(Z)) - \log \E_P[ k(ta(Z))]} \\
& + \frac{d}{dt}\mid_{t=0}\set*{\log k(c_t(Z)Y + tb(Y,Z)) - \log \E_P[ k(c_t(Z)Y + t b(Y,Z))  \mid Z] } \\ 
&=a(Z) - \E_P[ a(Z) ] + c_0'(Z)Y + b(Y,Z) - \E_P[ c_0'(Z)Y + b(Y,Z)  \mid Z]  \\ 
&=a(Z) + b(Y,Z).
\end{align*}
Here we've used the property $c'_0=0$ derived above. All that is left is to show that our submodel is differentiable in
quadratic mean. Via \citet[Lemma 1.8]{van2000semiparametric}, 
it suffices to show that $\sqrt{p_t}$ is continuously differentiable near zero and
$\int [(p'_t)^2 / p_t] d\P$ is finite. These properties follow from our assumptions on $k$ and
and our characterization of $c_t$ as continuously differentiable in a neighborhood of zero.

\subsubsection{The Pathwise Derivative of $\chi$, Regularity, and Efficiency}
\label{sec:pathwise-derivative}
We will calculate the derivative of our functional $\chi(P)$ on the tangent space $\tangent$.
As discussed above, for any score $g \in \tangent$, there is a submodel $P_t$ of the form defined in \eqref{eq:submodels}, with regression function $m_{P_t}(z) = m_P(z) + t m'(z)$ for $m' \in \mm$. Furthermore, its score satisfies $m'(z) = \E_P[Y b(Y,Z) \mid Z=z]=\E_P[(Y-m_P(Z))g(Y,Z) \mid Z=z]$. 
The form of our path makes it easy to calculate the derivative of $\chi(P_t)$, as $\chi(P_t)$ depends only on $m_{P_t}$ and the marginal distribution of $Z$.

\begin{align*}
&(\partial/\partial t)\mid_{t=0} \chi(P_t) \\ 
&= (\partial/\partial t) \mid_{t=0}   \E_{P_t}[ h(Z,m_{P_t}) ]  \\
&= (\partial/\partial t) \mid_{t=0} \p{  \E_{P}[ k(ta(Z)) h(Z,m_{P_t})] / \E_{P}[ k(ta(Z))] } \\
&=  \p{ (\partial/\partial t) \mid_{t=0} \E_{P}[ k(ta(Z)) h(Z,m_{P_t}) ] } -  \p{(\partial/\partial t)\mid_{t=0}\E_P[ k(ta(Z))]}\E_{P}[ h(Z,m_{P}) ] \\
&= \E_{P}[ (\partial/\partial t)\mid_{t=0}  k(ta(Z)) h(Z,m_P + t m') ] - \E_P[ (\partial/\partial t)\mid_{t=0} k(ta(Z))]\E_{P}[ h(Z,m_{P}) ] \\
&= \E_P[ a(Z) h(Z, m_P)] + \E_P[ h(Z, m') ] - \E_P[a(Z)]\E_{P}[ h(Z,m_{P}) ]
\end{align*}
The third term here is zero, as $\E_P[a(Z)]=0$, and we rely on dominated convergence to interchange integration and differentiation.
Dominatedness follows, via the mean value theorem, from the boundedness of $k'(\cdot)$ and the square integrability of $a(z)$,
 $h(z,m)$, and $h(z,m')$. 

Recalling that $g(Y,Z)=a(Z) + b(Y,Z)$ where $\E_P[b(Y,Z) \mid Z] = 0$ and $\E_P[a(Z)]=0$, 
the first term above is equal to $\E_P[ g(Y,Z) (h(Z,m_P) - \psi_P(m_P))]$.
Furthermore, if $\riesz[\psi]$ is the Riesz representer for $\psi_P(\cdot) = \E_P[ h(Z, \cdot)]$ on a superset of $\mm$, we can write the second term as 
$\E_P[\riesz[\psi](Z) m'(Z)]= \E_P[\riesz[\psi](Z) \E_P[(Y - m_P(Z)) g(Y,Z) \mid Z]]$. 
Thus, $\influence(y,z) = h(z,m_P) - \psi_P(m_P) + \riesz[\psi](z)(y-m_P(z))$ is an influence function, as 
$\E_P[ \influence(Y,Z) g(Y,Z)]=(\partial/\partial t)\mid_{t=0}$ $\chi(P_t)$. This establishes our regularity claim.

Furthermore, $\influence$ is in the closure of the tangent space $\tangent$, and therefore the efficient influence function,
 if and only if $\E_P[ (Y-m_P(Z)) \influence(Y,Z) \mid Z=z]$ is in the closure of $\mm$.
As $\E_P[ (Y-m_P(Z)) \influence(Y,Z) \mid Z=z] = \E_P[(Y-m_P(Z))^2 \mid Z=z] \riesz[\psi](z)$,
this happens if and only if $v \riesz[\psi]$ is in the closure of $\mm$ for $v(z) = \E_P[(Y-m_P(Z))^2 \mid Z=z]$. 
This completes our proof.

Note that if $\psi_P$ is not continuous on $\mm$, $\chi$ is not differentiable at $P$:
for a submodel $P_t$ with constant marginal distribution on $Z$, $t^{-1}(\chi(P_t)-\chi(P)) = \psi_P(m')$.
Thus, as the existence of a regular estimator for $\chi(P)$ implies the differentiability of $\chi$ at $\P$ \citep[Theorem 2.1]{van1991differentiable}, 
it implies the continuity of $\psi_P$ on $\mm$. 

\subsection{Estimating $\riesz[\psi]$ at the optimal rate}
\label{sec:estimating-riesz-optimal}
Here we consider the optimality of the rate $\norm{\hriesz - \riesz[\psi]}_{L_2(P)} = O_P(r)$ discussed in Section~\ref{sec:tuning-sigma}.
We use the notation of Theorem~\ref{theo:simple-rate}. 

If $\ff$ is a class of uniformly bounded functions with empirical metric entropy 
\smash{$\log N(\ff; $} \smash{$\L{2}{\Pn}; \epsilon) = O(\epsilon^{-2\rho})$} for $\rho > 1$, it can be shown that $r_Q = $ \smash{$O(n^{-1/(2+2\rho)})$}
\citep[see e.g.,][Equation 2.4]{koltchinskii2006local}. Furthermore, if $\riesz[\psi]$ is bounded
and the map $f \to h(\cdot,f)$ is well-behaved, $\riesz[\psi]\ff$ and $h(\cdot,\ff)$ will satisfy the same entropy bound, and
$r$ will have this rate as well. When $\ff$ is the unit ball of a H\"older space $C^s$ of functions on $[0,1]^d$
with $s>d/2$, we have such an entropy bound with $\rho = d/(2s)$ \citep{tikhomirov1993varepsilon, van1994bracketing},
and we get the minimax rate $n^{-1/(2+d/s)}$ for estimating a function $f$ with \smash{$\norm{f}_{\ff} < \infty$}
from direct observations of $f(Z_i) + \xi_i$ with gaussian noise $\xi_i$ \citep[Theorem 3.2]{gyorfi2006distribution}.

While the general problem of estimating a Riesz representer is nonstandard, one point of reference is Example~\ref{exam:mar},
in which $\riesz[\psi](x,w)=w/e(x)$ for $e(x)=\E[W_i \mid X_i=x]$. If $e(x)$ is bounded away from zero,
$\riesz[\psi]$ and $e(x)$ are estimable at the same rate, and in the case that \smash{$\norm{e}_{C^s} < \infty$},
the minimax rate for estimating directly observed $e(x)$ is \smash{$n^{-1/(2+d/s)}$.}
In this example, our estimator $\hriesz[\psi]$ for \smash{$\ff = \set{ w g : \norm{g}_{C^s}  \le 1}$} 
attains the minimax rate, as $\norm{\riesz[\psi]}_{\ff} < \infty$ and $\ff$ has entropy comparable to that of the unit ball of our H\"older space.

\section{Simulation Study: Details}
\label{sec:simu_details}

Here we describe the cross-fitting scheme used to estimate 
$\htau(X_i)$ and $\hmu(X_i)$ in the simulation study discussed in Section~\ref{sec:simu}.
Ten-fold cross-fitting is used throughout: where $\htau(X_i)$ and $\hmu(X_i)$ appear in
\eqref{eq:ape_aml} and \eqref{eq:ape_dr}, we use estimators $\htau^{(-i)}$ and $\hmu^{(-i)}$
trained on the folds that do not include unit $i$. This 
reduces dependence on $(Y_i,X_i,W_i)$ and therefore mitigates potential own-observation bias 
in $\hpsi_{DR}$ \citep[see e.g.,][]{chernozhukov2016double}. However, we do get some dependence
through the estimates of $\hat r$ and $\he$ used to train $\htau$ and
through lasso tuning parameters, which are chosen once for all $i$ by cross-validation.
This dependence can be eliminated using a computationally demanding nested sample splitting scheme;
we follow the \texttt{grf} package of \citet*{athey2016generalized} in using the following simplified scheme.

\begin{enumerate}
\item Partition the indices $1 \ldots n$ into $K=10$ folds of equal size, associating each index with a fold $f_i \in 1 \ldots K$.
\item For each fold $j \in 1 \ldots K$, train $\hat r_{j,\lambda_r}$ and $\hat e_{j,\lambda_e}$ on observations for $\set{ i : f_i \neq j}$
      for fixed values $\lambda_r, \lambda_e$ of the lasso penalty parameter. 
      Choose values $\hat\lambda_r, \hat\lambda_e$  by cross-validation, solving
     \[ \hat\lambda_r = \argmin_{\lambda_r} \sum_{i=1}^n (Y_i - \hat r_{f_i, \lambda_r}(X_i))^2, \quad  \hat\lambda_e = \argmin_{\lambda_e}\sum_{i=1}^n (W_i - \hat e_{f_i, \lambda_e}(X_i))^2. \] 
\item For each fold $j \in 1 \ldots K$, train $\htau_{j,\lambda_{\tau}}(x) = \phi(x)^T \hat\beta_j$ for fixed $\lambda_{\tau}$ by  
\[ \hbeta_{j} = 
        \argmin_{\tau} \sum_{i : f_i \neq j} (Y_i - \hat r_{f_i,\hat\lambda_r}(X_i) - (W_i - \hat e_{f_i,\hat\lambda_e}(X_i)) \phi(X_i)^T \beta_j)^2 + \lambda_{\tau}\norm{\beta_j}_{\ell_1}. \]
Choose a value $\hat\lambda_{\tau}$ by cross-validation, solving
\[ \hat\lambda_{\tau} = \argmin_{\lambda_\tau}\sum_{i=1}^n (Y_i - \hat r_{f_i,\hat\lambda_r}(X_i) - (W_i - \hat e_{f_i,\hat\lambda_e}(X_i)) \htau_{f_i,\lambda_{\tau}}(X_i))^2. \] 
\item Define $\htau^{(-i)}=\htau_{f_i, \hat\lambda_{\tau}}$ and $\hmu^{(-i)}=\hat r_{f_i,\hat\lambda_r}(x)- \htau_{f_i, \hat\lambda_\tau}(x) \he_{f_i,\hat\lambda_e}(x)$.
\end{enumerate}

\section{Computing the Weights}
\label{sec:computation}

The optimization problem \eqref{eq:aml-primal} that defines our weights $\hriesz$, 
\[ \hriesz = \argmin_{\gamma \in \R^n} I_{h,\ff}^2(\gamma) + \frac{\sigma^2}{n^2}\norm{\gamma}^2,
\ \ I_{h,\ff}(\gamma) = \sup_{f \in \ff} \frac{1}{n} \sum_{i = 1}^n [\gamma_i f(Z_i) - h(Z_i,f)]. \] 
is strongly convex and not extremely high dimensional, so it is often fairly tractable.
However, because the objective function involves a supremum over a set $\ff$,
 it is helpful to reformulate the problem for implementation. We will discuss the case
that $\ff$ is the absolutely convex hull $\absconv\set{\phi_1,\phi_2, \ldots}$.
We will assume, in addition, that we have decay in $\phi_j$ and $h(\cdot,\phi_j)$ that justifies working with a finite
dimensional approximation to $\ff$, the absolutely convex hull
$\ff_p$ of the first $p$ basis functions. When we do this, our optimization problem above can be expressed as a finite-dimensional quadratic program,
\begin{align*}
&\argmin_{(t, \gamma) \in \RR^{n+1}} t^2 + \frac{\sigma^2}{n^2}\norm{\gamma}^2 && \text{subject to}\ \\
& \abs*{n^{-1}\sum_{i=1}^n[ \gamma_i \phi_j(Z_i) - h(Z_i,\phi_j)] } \le t && \text{for }\ j \in 1 \ldots p,
\end{align*} 
as H\"older's inequality is sharp on $\ell_1$, i.e.,
\begin{align*} 
I_{h,\ff_p}(\gamma)&=\sup_{\norm{\beta}_{\ell_1} \le 1} \sum_{j=1}^p\beta_j \p{n^{-1}\sum_{i=1}^n[\gamma_i \phi_j(Z_i) - h(Z_i,\phi_j)]} \\
    &= \max_{j \le p} \abs*{n^{-1}\sum_{i=1}^n [\gamma_i \phi_j(Z_i) - h(Z_i,\phi_j)]}.
\end{align*} 

Our implementation of the average partial effect estimator described in Section~\ref{sec:simu},  included in the
\texttt{R} package \texttt{amlinear}, uses a variant of this formulation appropriate to the class
$\ff_{\hh}$ \eqref{eq:F_ape}:
\begin{align*}
&\argmin_{(t_{\mu},t_{\tau}, \gamma) \in \RR^{n+2}} t_{\mu}^2+t_{\tau}^2 + \frac{\sigma^2}{n^2}\norm{\gamma}^2 && \text{subject to}\ \\
& \abs*{n^{-1}\sum_{i=1}^n \gamma_i \phi_j(Z_i) } \le t_{\mu} && \text{for }\ j \in 1 \ldots p, \\
& \abs*{n^{-1}\sum_{i=1}^n (W_i \gamma_i - 1) \phi_j(Z_i) } \le t_{\tau} && \text{for }\ j \in 1 \ldots p.
\end{align*} 

\paragraph*{Finite-dimensional approximation}
To determine the number of basis functions we need to include in $\ff_p$, 
we consider the conditional bias term in \eqref{eq:error-decomp}. For any weights $\gamma$, it is bounded by 
\[ \norm{\hm - m}_{\ff}I_{h,\ff}(\gamma) = \norm{\hm - m}_{\ff}\sqb{ I_{h,\ff_p}(\gamma) + I_{h,\ff}(\gamma) - I_{h,\ff_p}(\gamma)}. \]
In particular, for 
\[ \hgamma  = \argmin_{\gamma \in \R^n} I_{h,\ff}^2(\gamma) + \frac{\sigma^2}{n^2}\norm{\gamma}^2,\quad
   \hgamma_p = \argmin_{\gamma \in \R^n} I_{h,\ff_p}^2(\gamma) + \frac{\sigma^2}{n^2}\norm{\gamma}^2, \]
we have the bound
\begin{align*} 
&\norm{\hm - m}_{\ff} I_{h,\ff_p}(\hgamma_p) + \norm{\hm - m}_{\ff}\p{I_{h,\ff}(\hgamma_p) - I_{h,\ff_p}(\hgamma_p)} \\
&\le \norm{\hm - m}_{\ff} I_{h,\ff}(\hgamma) + \norm{\hm - m}_{\ff}\p{I_{h,\ff}(\hgamma_p) - I_{h,\ff_p}(\hgamma_p)}.
\end{align*}
The excess that results from the use of $\hgamma_p$ instead of $\hgamma$ is bounded by the latter term. In this term, the difference 
$I_{h,\ff}(\hgamma_p) - I_{h,\ff_p}(\hgamma_p)$ is bounded by 
\begin{align*} 
&\sup_{f \in \ff}\inf_{f \in \ff_p}\p{\norm{\hgamma_p}\norm{f-f'}_{L_2(\Pn)} + \norm{h(\cdot,f-f')}_{L_2(\Pn)}} \\
&\le \norm{\hgamma_p}\sup_{j > p}\norm{\phi_j}_{L_2(\Pn)} + \sup_{j > p}\norm{h(\cdot,\phi_j)}_{L_2(\Pn)},
\end{align*}
so we choose $p$ to control these suprema. To ensure that this excess is negligible relative to variance, so our estimator is
asymptotically linear under the assumptions of Theorem~\ref{theo:simple}, these suprema must be $o_P(n^{-1/2})$.

\subsection{A dual approach}
It is also possible to use the dual characterization of $\hgamma$ given by Lemma~\ref{lemma:duality},
\[ \hriesz_i = \hg(Z_i) \ \text{ for }\ \hg = \argmin_g n^{-1}\sum_{i=1}^n [g(Z_i)^2 - 2 h(Z_i, g)] + \sigma^2 n^{-1}\norm{g}_{\ff_p}^2, \]
i.e., $\hriesz_i = \phi(Z_i)^T \hat\beta$ for 
\begin{align*}
&\hbeta = \argmin_{\beta \in \R^p} \beta^T \Phi \beta - 2h^T \beta + \sigma^2 n^{-1}\norm{\beta}_{\ell_1}^2, \\ 
&\Phi = n^{-1}\sum_{i=1}^n \phi(Z_i)\phi(Z_i)^T, \ h_j = n^{-1}\sum_{i=1}^n h(Z_i, \phi_j).
\end{align*}
We can solve this by splitting $\beta$ into positive and negative parts, which gives an equivalent second order cone program: $\hbeta=\hbeta^+ - \hbeta^-$ where the latter solve 
\begin{align*}
&\argmin_{(s,t,\beta^+, \beta^-) \in \R^{2p+2}} s - 2\begin{pmatrix} h & -h\end{pmatrix}\begin{pmatrix}\beta^+ \\ \beta^-\end{pmatrix} + \sigma^2 n^{-1} t \quad \text{ subject to} \\
&\beta^+,\ \beta^- \ge 0, \\
&\begin{pmatrix}
\beta^+ \\
\beta^-
\end{pmatrix}^T
\begin{pmatrix}
\hphantom{-}\Phi &\ -\Phi \\ 
-\Phi            &\ \hphantom{-}\Phi
\end{pmatrix}
\begin{pmatrix}
\beta^+ \\
\beta^-
\end{pmatrix} \le s, \\
&\p{ \sum_{j=1}^p \beta^+_j + \sum_{j=1}^p \beta^{-}_j}^2 \le t, \\
\end{align*}

\subsection{Computation in Hilbert Spaces}
The approaches discussed above rely on finite-dimensional approximation and efficient solvers for quadratic and second order cone programs.
In contrast, when $\ff$ is the unit ball of a Reproducing Kernel Hilbert Space, we can often solve the dual 
\[ \hg = \argmin_g n^{-1}\sum_{i=1}^n [g(Z_i)^2 - 2 h(Z_i, g)] + \sigma^2 n^{-1}\norm{g}_{\ff}^2 \]
without approximation by solving a $n \times n$ linear system. In particular, if $h(Z_i, g)$ is a function of $g(Z_i)$,
a well-known Representer theorem  states that the solution to this problem
has the form $\hg(z) = \sum_{j=1}^n \hat\alpha_j K(Z_j,z)$, where $K(z',z)$ is the kernel 
associated with our space \citep{scholkopf2001generalized}.
 When $\hg$ is known to have this form, we can calculate $\hat\alpha$ by substituting $g(z)=\sum_{j=1}^n \alpha_j K(Z_j, z)$
into our dual problem above and solving the resulting unconstrained quadratic optimization problem over $\alpha$,
\begin{align*} 
&\hat \alpha = \argmin_\alpha \alpha^T \bar K \alpha - 2 \bar h^T \alpha + \sigma^2 n^{-1} \alpha^T G \alpha, \\
&\bar{K}_{ij} = n^{-1}\sum_{k=1}^n K(Z_i, Z_k) K(Z_j,Z_k), \\ 
&\bar h_i = n^{-1}\sum_{k=1}^n h(Z_k, K(Z_i,\cdot)),\\
& G_{ij} = K(Z_i,Z_j).
\end{align*}
This approach works in Example~\ref{exam:mar}, as $h(Z_i,g)$ has the required form,
and variations apply in our other examples, 
as similar representer theorems hold under appropriate conditions \citep[see e.g.][]{argyriou2014unifying}.

\section{Consistency of penalized least squares estimators}
\label{appendix:penalized-least-squares}
In this section, we state and prove a consistency result for penalized least squares 
relevant to our claims in Remark~\ref{rema:consistency}. We base our presentation 
on that of \citet{lecue2017regularization}.

\begin{theo}
\label{theo:donsker-consistency}
Let $(Y_1,Z_1)\ldots(Y_n,Z_n) \sim \P$ be independent and identically distributed, 
let $m_{\star} = \argmin_{m \in \mm}\P (Y - m(Z))^2$ for closed convex $\mm \subseteq L_2(\P)$,
and let $\ell(m) =  \Pn (Y_i - m(Z_i))^2 + \lambda \norm{m}_{\F}$ for some norm $\norm{\cdot}_{\F}$.
If $\ell(\hm) \le \ell(m_{\star})$ for $\hm \in \mm$, then 
$\norm{\hm-m}_{\F} \le \alpha$, $\norm{\hm-m}_{L_2(\P)} \le \alpha r$, and 
$\norm{\hm-m}_{L_2(\Pn)} \le (2\eta_M \alpha r^2 + 2\lambda \norm{m_{\star}}_{\F})^{1/2}$   
\[ \text{ where }\ \alpha = \max\set{ 2\eta_M/\eta_Q + \sqrt{2\lambda \norm{m_{\star}}_{\F}/(\eta_Q r^2)},\ 2\norm{m_{\star}}_{\F}/(1-2\eta_M r^2/\lambda)}, \]
on an event on which, for all $f \in \F = \set{f : \norm{f}_{\F} \le 1}$,
\begin{equation}
\label{eq:least-squares-ratio-process}
\begin{aligned}
&\Pn f^2 \ge \eta_Q \P f^2          	       && \text{ if }\quad \P f^2 \ge r^2, \\
&\abs{(\P - \Pn)(Y-m_{\star})f} \le \eta_M r^2  && \text{ if }\quad \P f^2 \le r^2.
\end{aligned}
\end{equation}
\end{theo}

Taking $\mm = L_2(\P)$, this characterizes the consistency of an estimator $\hm$ of 
the regression function $m_{\star}(z) = \E[Y \mid Z=z]$. And with appropriate tuning, this simplifies.

\begin{coro}
Under the assumptions of Theorem~\ref{theo:donsker-consistency}, 
taking $\lambda=c\eta_Mr^2$ for any constant $c > 2$,  
$\norm{\hm-m}_{\F} \lesssim \alpha$, 
$\norm{\hm - m}_{L_2(\P)} \lesssim \alpha r$,
and $\norm{\hm - m}_{L_2(\Pn)} \lesssim \sqrt{\eta_M \alpha}r$ 
for $\alpha = \max(\eta_M/\eta_Q, \norm{m_{\star}})$.
\end{coro} 

The first condition in \eqref{eq:least-squares-ratio-process},
a uniform lower bound, holds with high probability for $r$ satisfying 
$R_n(\F_{c_0r}) \le c_1 r^2/\sup_{f \in \F}\norm{f}_{\infty}$ for constants $c_0$ and $c_1$ dependent on $\eta_Q$,
and variants apply to unbounded classes \citep{mendelson2017extending}.
The second, if the `noise' $\xi_i = Y_i - m_{\star}(Z_i)$ is in some bounded interval $[-b,b]$,
holds with high probability if $R_n(\F_{r}) \le \eta_M r^2/(2b)$
by symmetrization and contraction, and similar claims hold for unbounded but relatively well behaved noise
via multiplier inequalities \cite[see e.g.,][Section 3.14]{gine2015mathematical}.

\begin{proof}
The Hilbert space projection theorem implies that the minimizer $m_{\star}$ exists
and satisfies $\P(Y-m_{\star})(m-m_{\star}) \le 0$ for all $m \in \mm$ \citep[Proposition 1.37]{peypouquet2015convex}.
We use this property to lower bound the excess loss $\ell(m) - \ell(m_{\star})$. 
\begin{align*}
\ell(m) - \ell(m_{\star}) 
&= \Pn [(Y - m)^2 - (Y - m_{\star})^2] + \lambda (\norm{m}_{\F}- \norm{m_{\star}}_{\F}) \\
&= \Pn (m-m_{\star})^2 - 2\Pn (Y-m_{\star})(m-m_{\star}) + \lambda (\norm{m}_{\F}- \norm{m_{\star}}_{\F}) \\
&\ge \Pn (m-m_{\star})^2 + 2(\P - \Pn)(Y-m_{\star})(m-m_{\star})  + \lambda (\norm{m}_{\F} - \norm{m_{\star}}_{\F}) \\
&\ge \Pn (m-m_{\star})^2 - 2\abs{(\P - \Pn)(Y-m_{\star})(m-m_{\star})}  + \lambda (\norm{m-m_{\star}}_{\F} - 2 \norm{m_{\star}}_{\F}).
\end{align*}
In the last step, we use the triangle inequality bound $\norm{a-b} \ge \norm{a} - \norm{b}$ for $a=m-m_{\star}$ and $b=-m_{\star}$.
The last bound is a function of $\delta = m-m_{\star}$, which we will call $\ell'(\delta)$.

Our argument will be based on scaled versions of our assumed bounds \eqref{eq:least-squares-ratio-process}.
Using the scale invariance of the first and Lemma~\ref{lemm:rescaling} for the second, for all $f \in \alpha \F$,
\begin{equation}
\label{eq:erm-ratio-bounds-scaled}
\begin{aligned}
& \Pn f^2 \ge \eta_Q \P f^2                                  & \text{ if }\ \P f^2 \ge (\alpha r)^2, \\ 
&\abs{(\P-\Pn) (Y-m_{\star})f} \le \eta_M \P f^2 / \alpha    & \text{ if }\ \P f^2 \ge (\alpha r)^2, \\
&\abs{(\P-\Pn) (Y-m_{\star})f} \le \eta_M \alpha r^2         & \text{ if }\ \P f^2 \le (\alpha r)^2. \\
\end{aligned}
\end{equation}

We will show that $\ell'(f) > 0$ if $\norm{f}_{\F} \ge \alpha$. 
It suffices to do so for $\norm{f}_{\F}=\alpha$, as if $\norm{f}_{\F} \ge \alpha$,
$f=tg$ for $\norm{g}_{\F} = \alpha$ and $t \ge 1$,
and $\ell'(f) = \ell'(tg) \ge t\ell'(g)$, as  $\ell'(tg) - t\ell'(g)= (t^2-t)\Pn g^2 - 2\lambda (1-t)\norm{m_{\star}}_{\F}$.
When $\norm{f}_{\F}=\alpha$, 
$\ell'(f) \ge (\eta_Q - 2\eta_M/\alpha) \P f^2 +  \lambda (\alpha - 2 \norm{m_{\star}}_{\F})$ if $\P f^2 \ge (\alpha r)^2$
and $\ell'(f) \ge -2\eta_M\alpha r^2 + \lambda (\alpha - 2 \norm{m_{\star}}_{\F})$ otherwise,
so it will be positive if $\eta_Q - 2\eta_M/\alpha \ge 0$ and  $-2\eta_M\alpha r^2 + \lambda (\alpha - 2 \norm{m_{\star}}_{\F}) > 0$.
This holds for $\alpha \ge 2\max\set{\eta_M/\eta_Q, \norm{m_{\star}}_{\F}/(1-2\eta_M r^2/\lambda)}$.

We will now consider the case that $\norm{f}_{\F} \le \alpha$. In this case,
$\ell'(f) \ge (\eta_Q - 2\eta_M/\alpha) \P f^2 -  2 \lambda \norm{m_{\star}}_{\F}$ if $\P f^2 \ge (\alpha r)^2$
and $\ell'(f) \ge \Pn  f^2 - 2\eta_M\alpha r^2 - 2 \lambda \norm{m_{\star}}_{\F}$ otherwise. 
Thus, $\ell'(f) > 0$ if $\P f^2 > \max\set{(\alpha r)^2, 2 \lambda \norm{m_{\star}}_{\F}/(\eta_Q - 2\eta_M/\alpha)}$
or if $\P f^2 \le (\alpha r)^2$ and $\Pn f^2 > 2(\eta_M \alpha r^2 + \lambda \norm{m_{\star}}_{\F})$.
And when $(\alpha r)^2 \ge 2 \lambda \norm{m_{\star}}_{\F}/(\eta_Q - 2\eta_M/\alpha)$, 
and therefore for $\alpha \ge 2\eta_M/\eta_Q + \sqrt{2\lambda \norm{m_{\star}}_{\F}/\eta_Q r^2}$,
this conclusion simplifies to $\ell'(f) > 0$ if $\P f^2 > (\alpha r)^2$
or if $\Pn f^2 > 2(\eta_M \alpha r^2 + \lambda \norm{m_{\star}}_{\F})^2$.
Our claim follows, as $\ell'(\hm-m) \le 0$.
\end{proof}

\end{appendix}

\fi

\end{document}